\documentclass[prd,aps,amsfonts,showpacs,superscriptaddress,nofootinbib,longbibliography,notitlepage]{revtex4-1}

\usepackage{graphicx}
\usepackage{xcolor}
\usepackage{rotating}
\usepackage{amsmath,amssymb,graphics,amsthm}

\usepackage[colorlinks=true, urlcolor=violet, linkcolor=blue, citecolor=red, hyperindex=true, linktocpage=true]{hyperref}
\usepackage[capitalise,compress]{cleveref}

\newtheorem{thm}{Theorem}
\newtheorem{cor}[thm]{Corollary}
\newtheorem{lem}[thm]{Lemma}

\newtheorem{prop}[thm]{Proposition}

\makeatletter
\renewcommand{\p@subsection}{}
\renewcommand{\p@subsubsection}{}
\makeatother

\allowdisplaybreaks

\newcommand{\avg}[1]{\left \langle #1 \right\rangle}
\newcommand{\ket}[1]{\left | #1 \right\rangle}
\newcommand{\oket}[1]{\left\vert#1\right)}
\newcommand{\bra}[1]{\left \langle #1 \right |}

\newcommand{\abs}[1]{\left | #1 \right|}
\newcommand{\OO}[1]{ O\left(#1\right)}
\newcommand{\norm}[1]{\left\|#1\right\|}
\newcommand{\comm}[1]{\left[#1\right]}
\newcommand{\dist}{\mathcal{D}}

\usepackage{xcolor}

\usepackage{mathtools}

\newcommand{\CNOT}{\textsc{CNOT}}

\newcommand{\vac}{\text{vac}}

\newcommand{\ii}{\mathrm{i}}

\newcommand{\PP}{\mathbb{P}}

\usepackage{dsfont}

\begin{document}
\title{Hierarchy of linear light cones with long-range interactions}

\author{Minh C. Tran}
\affiliation{Joint Center for Quantum Information and Computer Science,
NIST/University of Maryland, College Park, MD 20742, USA}
\affiliation{Joint Quantum Institute, NIST/University of Maryland, College Park, MD 20742, USA}

\author{Chi-Fang Chen}
\affiliation{Department of Physics, California Institute of Technology, Pasadena CA 91125, USA}

\author{Adam Ehrenberg}
\author{Andrew Y. Guo}
\author{Abhinav Deshpande}
\affiliation{Joint Center for Quantum Information and Computer Science,
NIST/University of Maryland, College Park, MD 20742, USA}
\affiliation{Joint Quantum Institute, NIST/University of Maryland, College Park, MD 20742, USA}

\author{Yifan Hong}
\affiliation{Department of Physics, University of Colorado, Boulder CO 80309, USA}

\author{Zhe-Xuan Gong}
\affiliation{Department of Physics, Colorado School of Mines, Golden, CO 80401, USA}
\affiliation{National Institute of Standard and Technology, Boulder, CO 80305, USA}

\author{Alexey V. Gorshkov}
\affiliation{Joint Center for Quantum Information and Computer Science,
NIST/University of Maryland, College Park, MD 20742, USA}
\affiliation{Joint Quantum Institute, NIST/University of Maryland, College Park, MD 20742, USA}

\author{Andrew Lucas}
\email{andrew.j.lucas@colorado.edu}
\affiliation{Department of Physics, University of Colorado, Boulder CO 80309, USA}
\affiliation{Center for Theory of Quantum Matter, University of Colorado, Boulder CO 80309, USA}

\begin{abstract}
In quantum many-body systems with local interactions, quantum information and entanglement cannot spread outside of a \emph{linear light cone}, which expands at an emergent velocity analogous to the speed of  light. Local operations at sufficiently separated spacetime points approximately commute---given a many-body state $|\psi\rangle$,  $\mathcal{O}_x(t) \mathcal{O}_y |\psi\rangle \approx \mathcal{O}_y\mathcal{O}_x(t) |\psi\rangle$ with arbitrarily small errors---so long as $|x-y|\gtrsim vt$, where $v$ is finite. Yet most non-relativistic physical systems realized in nature have long-range interactions:  two degrees of freedom separated by a distance $r$ interact with potential energy $V(r) \propto 1/r^{\alpha}$.  In systems with long-range interactions, we rigorously establish a \emph{hierarchy of linear light cones}: at the same $\alpha$, some quantum information processing tasks are constrained by a linear light cone while others are not. In one spatial dimension, this linear light cone exists for every many-body state $|\psi\rangle$ when $\alpha>3$ (Lieb-Robinson light cone); for a typical state $|\psi\rangle$ chosen uniformly at random from the Hilbert space when $\alpha>\frac{5}{2}$ (Frobenius light cone); for every state of a non-interacting system when $\alpha>2$ (free light cone).  These bounds apply to time-dependent systems and are optimal up to subalgebraic improvements.  Our theorems regarding the Lieb-Robinson and free light cones---and their tightness---also generalize to arbitrary dimensions. We discuss the implications of our bounds on the growth of connected correlators and of topological order, the clustering of correlations in gapped systems, and the digital simulation of systems with long-range interactions.  In addition, we show that universal quantum state transfer, as well as many-body quantum chaos, are bounded by the Frobenius light cone, and therefore are poorly constrained by all Lieb-Robinson bounds.

\end{abstract}

\date{\today}

\maketitle
\tableofcontents

\section{Introduction}\label{sec:intro}

While non-relativistic quantum systems do not possess intrinsic absolute speed limits, their dynamics exhibit a form of causality analogous to the speed of light.
Lieb and Robinson first deduced the existence of a finite velocity for the propagation of information in quantum spin systems with finite-range interactions \cite{Lieb1972}.
This leads to ballistic dynamics, out of which a \emph{linear light cone} emerges.

For systems with power-law interactions, i.e.\ those that fall off as $1/r^\alpha$ in the distance $r$ between two degrees of freedom, the story is much richer.
Such long-range interactions have been exhibited in a variety of quantum simulators and technological platforms, including ultracold atomic gases \cite{Aikawa2012}, Rydberg atoms \cite{Saffman2010}, one dimensional chains of trapped ions~\cite{Britton2012}, polar molecules~\cite{Yan2013}, color centers in solid-state systems~\cite{Yao2012}, and atoms trapped in photonic crystals~\cite{Douglas2015}.   More formally, most physical systems consist of objects with electrical charges or electromagnetic dipoles, and so, fundamentally, these systems also exhibit long-range interactions. Since most developments in condensed matter physics and statistical physics are based on systems with short-range, local interactions, it is important to know to what extent the canonical paradigms still hold in the presence of long-range interactions. In addition to being interesting from this fundamental-science perspective, long-range interactions can also be used to significantly improve the performance of various quantum technologies, such as quantum computing \cite{Eldredge2017,linkeExperimentalComparisonTwo2017,Deshpande2018}, quantum simulation \cite{Tran2019,landsmanVerifiedQuantumInformation2019} and quantum metrology \cite{PhysRevA.46.R6797,fossfeig2016entanglement}.

Until recently, it was unknown whether or not there existed a critical value of the power-law exponent $\alpha$ above which a linear light cone is present.
Hastings and Koma~\cite{Hastings2006} first demonstrated a light cone whose velocity diverges exponentially in distance for $\alpha$ greater than the lattice dimension, $d$.
Progressive improvements yielded a series of algebraic light cones for $\alpha > 2d$, which tend to a linear light cone in the limit as $\alpha \to \infty$ \cite{Foss-Feig2015,Tran2019a}.  After numerical simulations suggested the existence of a sharp linear light cone \cite{Chen2019b, Luitz2019, Secular2019}, a proof of generic linear light cones was found for systems with interaction exponent $\alpha>2d+1$ \cite{Chen2019,Kuwahara2019}.

Complementary to the Lieb-Robinson bounds are protocols that achieve the (asymptotically) fastest allowable rates of quantum information processing. One such dynamical task is quantum state transfer, which has been used experimentally to demonstrate the transmission of entanglement in quantum systems \cite{Cirac1997}. These protocols can be directly connected to the Lieb-Robinson bound \cite{Eldredge2017,Epstein2017} and have been a standard way to benchmark the sharpness of these bounds.

The goal of this paper is to answer two important questions: first, can the result in Refs.~\cite{Chen2019,Kuwahara2019} be tightened? In particular, does there exist a universal \emph{linear} light cone for some $\alpha<2d+1$? Second, do the tightest light-cone bounds imply correspondingly tight bounds on interesting measures of information spreading, such as quantum state transfer or scrambling? In other words, are Lieb-Robinson bounds optimal in practice for constraining quantum information dynamics?

Surprisingly, the answer to both questions is ``no.''  In this paper, we show that quantum information can spread at arbitrarily large ``velocities" once the power law exponent $\alpha < 2d + 1$, thus proving the tightness of the recent bounds~\cite{Chen2019,Kuwahara2019}. We also show that a \emph{Frobenius} bound can give tighter constraints on quantum state transfer tasks---as well as many-body quantum chaos---than \emph{Lieb-Robinson} bounds.  
We prove that the light cone given by the Frobenius bound is linear for $\alpha>\frac{5}{2}$ in $d=1$, and conjecture the generalization $\alpha>\frac{3}{2}d+1$ for higher dimensions.  Additionally, in systems that are described by non-interacting bosons or fermions, we prove a linear \emph{free-particle} light cone for $\alpha>d+1$.   All of these cutoffs in this hierarchy of linear light cones are tight: see Figure \ref{fig:lightcones}.

\begin{figure}[t]
\includegraphics[width=0.45\textwidth]{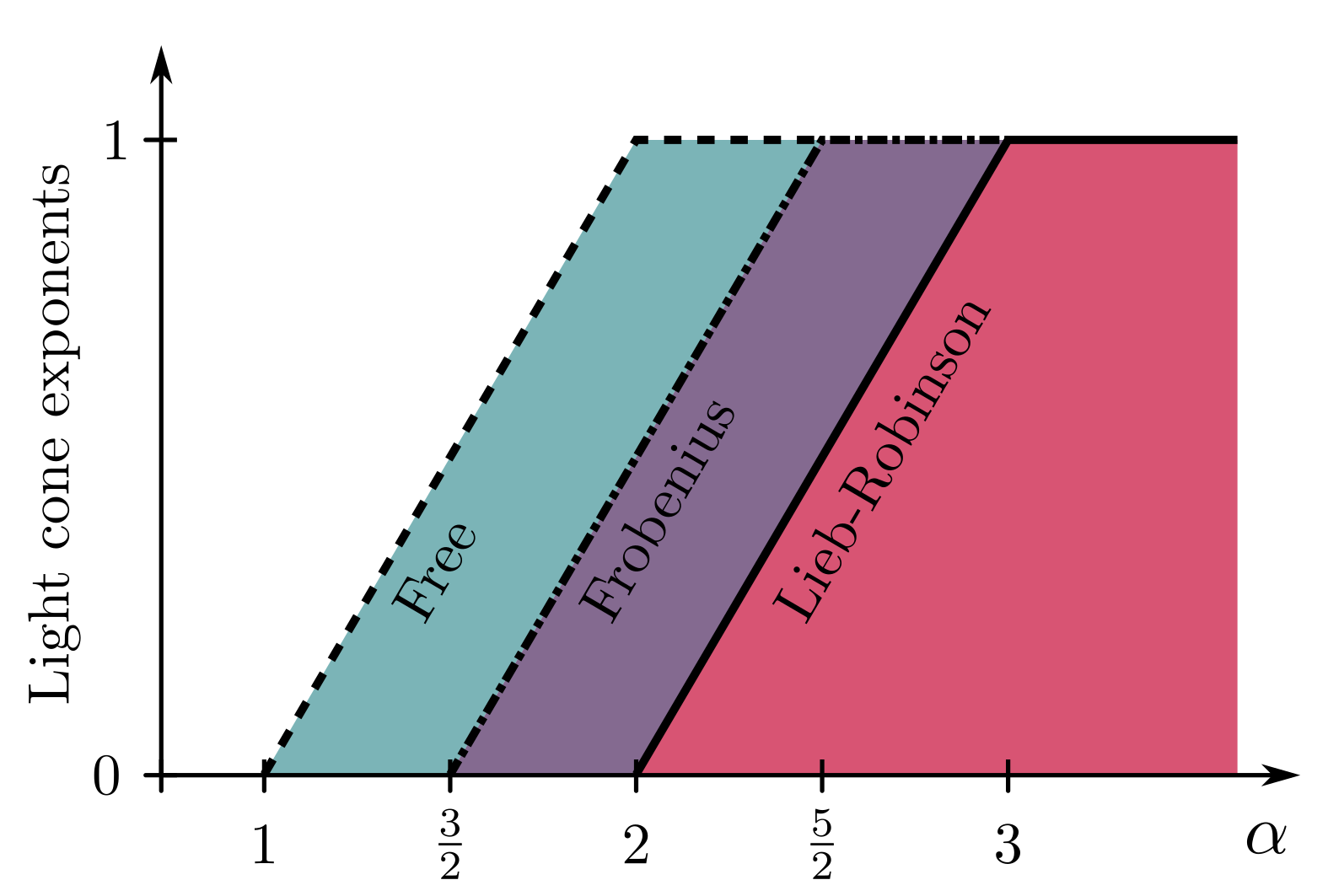}
\caption{The hierarchy of linear light cones in one dimension; we say that a light cone has exponent $\gamma$ if $\lVert [A_0(t),B_r]\rVert$ is large only when $t\gtrsim r^\gamma$.  The plot depicts the exponents of the Lieb-Robinson light cone (solid line)~\cite{Chen2019}, the Frobenius light cone from \cref{thm:mbqw} (dot-dashed line), and the free light cone from \cref{thm:free} (dashed line) as functions of $\alpha$ in one dimension.  The free light cone is known to be a tight bound for all $\alpha$.  We also show that the Lieb-Robinson and Frobenius light cones are not linear below $\alpha=3$ and $\alpha=\frac{5}{2}$ respectively.}
\label{fig:lightcones}
\end{figure}

These results immediately demonstrate that the long-observed mismatch between Lieb-Robinson bounds and state-transfer protocols that aim to saturate the bounds, such as that of Ref.~\cite{Eldredge2017}, is not entirely a limitation of our creativity or mathematical prowess, but is rather linked to a fundamental property of nature.  There are, simply put, \emph{multiple notions of locality} in systems with long-range interactions.  Furthermore, the tensions among these localities manifest themselves within a range of $\alpha$ that is easily accessible in experiment.  This unexpected result is the key finding of our paper.

The hierarchy of linear light cones we demonstrate is not only a profound property of nature, but also has important applications for quantum technologies. For example, systems with long-range interactions can be hard problems to simulate, both on classical and quantum computers. Proving the tightness of the linear Lieb-Robinson light cone at $\alpha>2d+1$ proves that a two-dimensional gas of atoms with dipole-dipole interactions can never be simulated as easily as one with local interactions with a provably small error.  At the same time, the hierarchy of light cones reveals that \emph{some} problems are much easier to simulate than had previously been realized. As a specific example, the Bose-Hubbard model (with long-range hopping) has been argued to be so difficult that its efficient simulation would serve as a demonstration of quantum supremacy \cite{Aaronson2011}.  Our light cones show that it is not difficult to simulate the low-density Bose-Hubbard model for $\alpha > d$, whereas previously this was only known for $\alpha > 2d$ \cite{Maskara2019}; as a result, we have substantially constrained the parameter space in which quantum supremacy can be demonstrated.  This result constrains how and when atoms with dipole-dipole interactions trapped in a two-dimensional optical lattice can perform hard quantum computation or simulation.

High fidelity quantum state transfer can be used to build fast remote quantum gates, which can significantly speed up a large-scale quantum computer. There is a growing interest in designing fully-connected quantum computers that take advantage of long-range interactions among physical qubits \cite{linkeExperimentalComparisonTwo2017}, and finding the optimal quantum state-transfer protocols using long-range interactions is a crucial part of the design.  The hierarchy of light cones we find reveals the fundamental inadequacy of Lieb-Robinson bounds for constraining \emph{universal} state transfer algorithms, which transfer the state of a single qubit independently of the states of other qubits. We develop a quantum walk formalism for constraining universal state-transfer protocols, and obtain parametrically better bounds than the Lieb-Robinson bound.  Furthermore, the framework that has been initiated in this work also reveals novel state-transfer protocols with desirable properties.  Specifically, we present a new method for using long-range interactions for state transfer that has two experimentally desirable features.  Firstly, our new protocol takes place in a constrained subspace of a many-body Hilbert space that is naturally realized in atomic platforms with a conserved magnetization.  Secondly, the protocol is extraordinarily robust to perturbations in the Hamiltonian, a desirable feature on account of the low-precision tunable couplers present in near-term quantum information processors.

Platforms with long-range interactions have been proposed as natural quantum simulators which approximately realize $\alpha = 0$ (i.e.\ all-to-all) interactions.  Systems with such a complete breakdown of locality can be highly desirable.  For example, they may simulate quantum gravity via the holographic correspondence \cite{Maldacena:1997re} and may enable the production of metrologically useful entanglement via spin-squeezing  \cite{PhysRevA.46.R6797,PhysRevA.56.2249,PhysRevLett.121.070403,foss-feig16b,PhysRevResearch.1.033075}. An important open question is how small $\alpha$ needs to be for locality to break down to a degree sufficient for realizing a particular application or particular physics. For example, are dipolar $1/r^3$ interactions in a given 1D, 2D, or 3D system sufficiently non-local? Our results indicate that the answer to these questions may depend on whether there are additional constraints in the system. Indeed, in a highly constrained subspace at high total spin in an SU(2)-symmetric model, we expect that the constraints arising from locality are stronger than the Lieb-Robinson light cone suggests, similar to the stronger light cone that arises for non-interacting particles. Therefore, in such constrained models, reaching non-locality may require a lower value of $\alpha$ compared to unconstrained models.

Lastly, we emphasize that given that there is a hierarchy of different notions of locality, exquisite care must be taken to analyze and interpret experimental results in long-range interacting quantum systems.

\section{Summary of results\label{sec:summary}}
We now provide a heuristic overview of why the hierarchy of light cones arises, along with a myriad of additional applications of these results in near-term quantum simulation experiments.  The remainder of the paper then contains the rigorous proof of all results, along with a brief conclusion.

For illustrative purposes, let us first consider a one-dimensional spin-$\frac{1}{2}$ chain with two-body long-range interactions.  Such models naturally arise in experiments, for example, using the nuclear spin-$\frac{1}{2}$ of an appropriate atom.  Letting $X^\beta_i := (X_i, Y_i,Z_i)$ denote the three Pauli matrices acting on the spin on the $i^{\mathrm{th}}$ site, we can consider a very broad family of time-dependent Hamiltonians of the form \begin{equation}
    H(t) = \sum_{i,\beta}  h_i^\beta(t) X_i^\beta + \sum_{i,j,\beta,\gamma} \frac{J_{ij}^{\beta\gamma}(t)}{|i-j|^\alpha} X^\beta_i X^\gamma_j. \label{eq:introH}
\end{equation}
Roughly speaking, if the coefficients $J_{ij}^{\beta\gamma}$ are all of the same order, we call this Hamiltonian a model with long-range interactions of power law exponent $\alpha$.  Remarkably, even though every spin is coupled with every other spin, this model is, for many practical purposes, \emph{local} for $\alpha$ sufficiently large \cite{Chen2019,Kuwahara2019} (indeed, this Hamiltonian even becomes finite-range in the limit $\alpha \to \infty$). But, what do we mean by locality?  And how small can $\alpha$ get before locality breaks down?
We will see that, in fact, there are multiple notions of locality, depending on the specific quantities of interest: different information processing tasks are sped up by long range interactions at different values of $\alpha$.

\subsection{Lieb-Robinson light cone}
A sensible notion of locality is to demand that any local perturbation, acting at site $x$ will only influence physics at sites within distance $vt$ of the original site $x$, after an amount of time $t$  \cite{PhysRevLett.113.030602}. This notion of locality is imposed by the original Lieb-Robinson bound \cite{Lieb1972}, which implies a ``linear light cone":  the quantity $\lVert [X_0(t),X_r] \rVert$ is small for $r>vt$, where $X_0$ and $X_r$ are local operators on lattice sites 0 and $r$ respectively and $\lVert \cdot \rVert$ denotes the operator norm (the largest magnitude of the operator's eigenvalues). More precisely, a linear light cone here means that for any small (but finite) value of $\epsilon$, we can find a finite velocity $v$ such that $\lVert [X_0(t),X_r] \rVert < \epsilon$ for $|t|< r/v$.

As noted previously, recent works have established linear Lieb-Robinson light cones for $\alpha>2d+1$~\cite{Chen2019,Kuwahara2019}.  The first main result of this work is to prove that linear Lieb-Robinson light cones are \emph{only} guaranteed for any $\alpha>2d+1$.  We prove this by explicitly constructing a Hamiltonian $H(t)$---of the form (\ref{eq:introH}), generalized to any $d$---such that \begin{equation}
    \lVert [X_0(t),X_r] \rVert \gtrsim \frac{t^{2d+1}}{r^\alpha}. \label{eq:introt2dp1}
\end{equation}
for two sites separated by a distance $r$ (Theorem \ref{thm:protocol}).

The construction of $H(t)$ that achieves \cref{eq:introt2dp1} can be broken down into three steps.
In the first step, we use time $\OO{1}$ to expand the operator $X_0$ to an operator $A_1$ supported on $O(t^d)$ sites located in a ball $\mathcal{B}_1$ of radius $t$ (\cref{fig:interaction}).
We then push this operator into another ball $\mathcal{B}_2$ of radius $\OO{t}$, which is centered around site $X_r$, using the Hamiltonian
\begin{equation}
    H_2(t) = \sum_{i \in \mathcal{B}_1} \sum_{j\in\mathcal{B}_2} \frac{Z_i Z_j}{|i-j|^\alpha}.
\end{equation}
Finally, we contract the operator onto site $X_r$ in time $O(t)$.
By a direct calculation, we show that, at the end of this process and to the lowest order in $t$, $\norm{\comm{[X_0(t),X_r]}}$ is proportional to the nested commutator $t\norm{\comm{A_2,\comm{H_2(t),A_1}}}$, where
\begin{equation}
A_{1,2} = \prod_{i\in\mathcal{B}_{1,2}} X_i.
\end{equation}
The nested commutator can be bounded:
 \begin{equation}
    \lVert [A_2, [H(t),A_1]]\rVert \sim t^d \left\lVert \left[A_2, \sum_{j\in \mathcal{B}_2} \frac{Z_j}{r^{\alpha}}\right]\right\rVert\sim \sum_{j\in\mathcal{B}_2} \frac{t^{2d}}{r^{\alpha}}, \label{eq:introHA}
    \end{equation}
and hence we obtain Eq. (\ref{eq:introt2dp1}).
We conclude that in a time $t$, it is faster to use long-range interactions than it is to use finite-range ones to grow $\norm{\comm{X_0(t),X_r}}$ when $\alpha<2d+1$.

\begin{figure}[t]
\centering
 \includegraphics[width=0.9\textwidth]{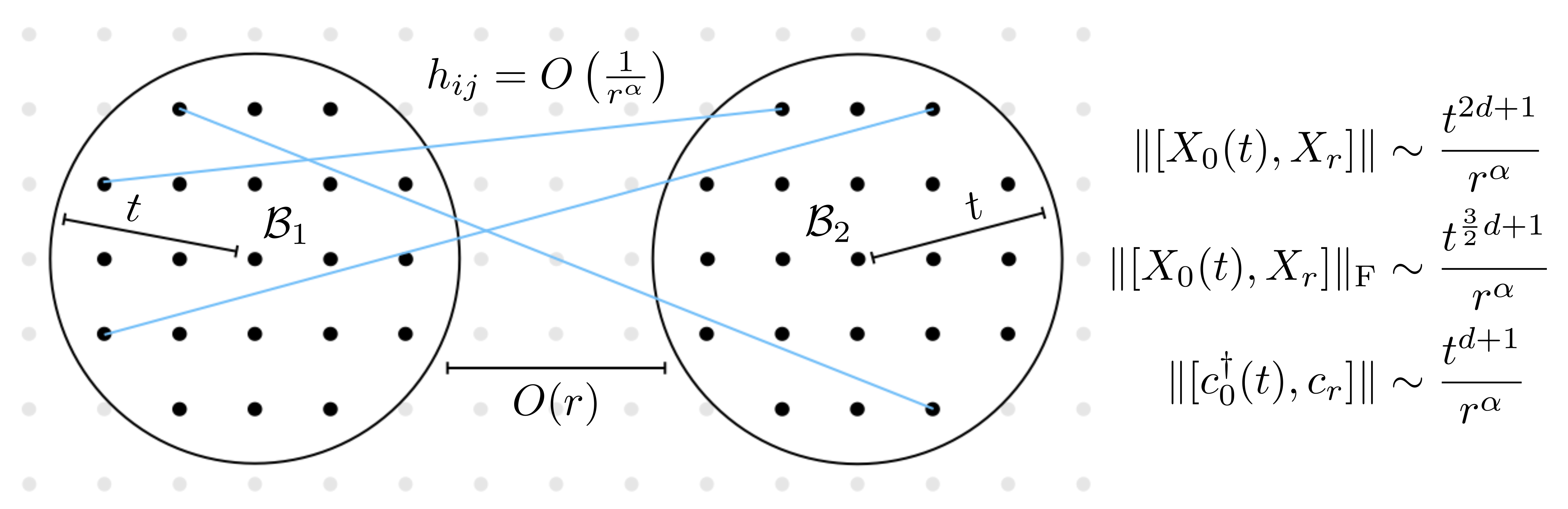}
  \caption{The norm of the interaction between two balls $\mathcal{B}_{1,2}$ of size $t$, separated by a distance $r$, determine the shape of the light cone. The critical values of $\alpha$ after which this norm becomes large differs depending on whether we use the operator or Frobenius norm, and whether the system is interacting or free.
} \label{fig:interaction}
\end{figure}

There are a number of important consequences of the tightness of the linear light cone at $\alpha=2d+1$.  For example, in \cref{sec:connected-corr}, we show that connected correlation functions of the form
\begin{align}
C(t) =	\bra{\psi(t)} X_0X_r\ket{\psi(t)} - \bra{\psi(t)} X_0\ket{\psi(t)}\bra{\psi(t)} X_r\ket{\psi(t)} \gtrsim \frac{t^{2d+1}}{r^\alpha}
\end{align}
can be achieved, even when the initial state $|\psi(0)\rangle$ does not have any entanglement between sites 0 and $r$. Because such correlation functions are routinely measured in quantum simulation experiments, this result resolves a long-standing issue of when the nonlinear light cones for correlations can occur with long-range interactions. For example, the experiment in Ref.\,\cite{richermeNonlocalPropagationCorrelations2014} suggests that $\alpha\approx 1$ marked the transition between linear and nonlinear light cones for spin correlations in a 11-site long-range Ising model.  Our result implies that other quantum systems with long-range interactions can transmit information \emph{much faster} than this Ising model.

Another important application of Lieb-Robinson bounds is to design efficient approximation algorithms for simulating quantum many-body dynamics, with either a classical computer \cite{Osborne2006} or a quantum computer \cite{Haah}. Given an initial state $\ket{\psi}$ and Hamiltonian $H$ of the form (\ref{eq:introH}), we consider the task of estimating the expectation value of the time-evolved observable $\avg{A(t)}:=\bra{\psi}U(t)^\dag AU(t)\ket{\psi}$ on a quantum computer, where $U(t) = \mathrm{e}^{-\mathrm{i}Ht}$ is the time translation operator generated by $H$ (assuming it does not depend on time).   When $A$ is a local operator, Lieb-Robinson bounds suggest that $\avg{A(t)}$ should depend only on the ``local" information stored in the wave function: one may as well trace out and ignore the sites sufficiently far away.  If a Lieb-Robinson bound implies we can trace out all sites a distance $>vt$ away from the support of $A$, computation of $\langle A(t)\rangle$ requires a small fraction of the resources needed to construct the full $U(t)$ acting on the entire many-body Hilbert space.  \cref{lem:loc-ob-app} makes this intuition precise and constrains the computational resources needed for faithful quantum simulation.

In \cref{sec:topo}, we use Lieb-Robinson bounds to constrain the minimum time $\tau^*$ it takes to create topologically ordered states from topologically trivial ones.
This result is of great practical value for experiments either studying topological matter or building topological quantum memories and topological quantum computers.
In finite-range interacting systems, $\tau^*$ scales linearly with the system size \cite{Bravyi2006}.
We extend this result in \cref{thm:topo} to power-law interacting systems with $\alpha>3d+1$.

In \cref{sec:clustering-corr}, we bound the spatial decay of correlation functions in a ground state of a gapped quantum phase with long-range interactions.
In Ref.~\cite{Hastings2006}, the authors show that in a time-independent power-law Hamiltonian with an exponent $\alpha$ and a spectral gap between the ground state and the first excited state, the correlations between distant sites in the ground state of the system also decay with the distance as a power law, with an exponent lower bounded by $\alpha'<\alpha$.
Yet no experiment and no numerical calculation has found a gapped system demonstrating correlation decay with exponent strictly less than $\alpha$. We prove that it is indeed impossible to saturate this bound; we show that the correlation exponent is lower bounded by $\alpha' = \alpha$ whenever $\alpha>2d$.

More broadly, a sharper knowledge of Lieb-Robinson light cones in quantum systems may improve previous bounds on area laws for quantum entanglement \cite{Hastings_2007} and heating rates in periodically driven systems~\cite{Abanin2015,Abanin2017a,Kuwahara2016,PhysRevLett.119.050501,Tran2019,Tran2019a}.

\subsection{Frobenius light cone}
In \cref{sec:fro}, we turn to a stronger notion of light cone, inspired by recent developments in the theory of many-body quantum chaos \cite{Shenker2014,Maldacena2016}.
Instead of the operator norm, we consider the Frobenius norm
\begin{equation}
     \lVert [X_0(t),X_r]\rVert_{\mathrm{F}} := \sqrt{\frac{\mathrm{tr}([X_0(t),X_r]^\dagger [X_0(t),X_r])}{\mathrm{tr}(\mathds{1})}}. \label{eq:frobeniusdef}
\end{equation}
This Frobenius norm, normalized by dimension, can be interpreted as the out-of-time-ordered correlation (OTOC) function used to probe early time chaos in many-body systems \cite{Shenker2014,Maldacena2016} or, equivalently, as the ``fraction" of the operator $X_0(t)$ that has support on the site $r$.  More intuitively, this OTOC can be understood by the following thought experiment: consider an initial quantum state $|\psi\rangle$, and perturb this quantum state by two operators: first, the local operator $X_r$ (which flips spin $r$ in the conventional $Z$-basis), and then the Heisenberg-evolved operator $X_0$, which amounts to flipping spin 0 at a later time $t$.  Does the order of these operations matter?  Clearly not if $t=0$: $X_0X_r = X_rX_0$.  However, if the operations occur at different times $t$, the effect could be significant:  $X_0(t)X_r|\psi\rangle$ might be a very different quantum state than $X_rX_0(t)|\psi\rangle$.  We can quantify how far apart these two states are in Hilbert space by asking for the typical length of $[X_0(t),X_r]|\psi\rangle$, or the value of $C = \langle \psi| [X_0(t),X_r]^\dagger [X_0(t),X_r]|\psi\rangle$.  A suitable notion of ``typical" is to choose a random initial state in the Hilbert space.  Averaging over all initial conditions amounts to replacing $\mathrm{tr}(|\psi\rangle\langle \psi| \cdots) \rightarrow \mathrm{tr}(\frac{1}{\dim(\mathcal{H})} \cdots)$.  Hence the average value of $C$ is given by (\ref{eq:frobeniusdef}). Mathematically, the Frobenius norm gives the average of the squared eigenvalues, while the operator norm used in Lieb-Robinson bounds is the maximal eigenvalue. Certainly, the Frobenius norm is always smaller: $\lVert [X_0(t),X_r]\rVert_{\mathrm{F}} \le \lVert [X_0(t),X_r]\rVert$.

Remarkably, in long-range interacting systems, we can show that the Frobenius norm is not only smaller by a constant prefactor, but is rather constrained by \emph{parametrically stronger} bounds.
Indeed, we prove in Section \ref{sec:fro} that $\lVert [X_0(t),X_r]\rVert_{\mathrm{F}}$ is bounded inside of an even stricter light cone, which is linear in one-dimensional models with two-body interactions so long as $\alpha>\frac{5}{2}$.
When $X_0$ is replaced by an operator on infinitely many sites $0,-1,-2,\ldots$, we also demonstrate the optimality of this bound, up to subalgebraic corrections.

To understand how the Frobenius light cone is deduced, let us revisit the argument for the Lieb-Robinson light cone.  Modifying Eq. (\ref{eq:introHA}) to use the Frobenius norm, we observe that \begin{equation}
    t\lVert [A_2, [H_2(t),A_1]]\rVert_{\mathrm{F}}
    \sim t \left\lVert \left[A_2, \sum_{j\in\mathcal{B}_2} \frac{t^dZ_j}{r^{\alpha}} \right]\right\rVert_{\mathrm{F}}
    \sim t\sqrt{\sum_{j\in\mathcal{B}_2} \left(\frac{t^{d}}{r^{\alpha}}\right)^2} \sim \frac{t^{\frac{3}{2}d+1}}{r^{\alpha}}.
\end{equation}
Hence, in $d$ dimensions, the Frobenius norm of the operator grows faster than in local models using long-range interactions once $\alpha<\frac{3}{2}d+1$.  For $d=1$, Theorem \ref{thm:mbqw} proves that this intuition is correct: the Frobenius light cone is linear when $\alpha>\frac{5}{2}$.    In between $\frac{3}{2}<\alpha<\frac{5}{2}$, in $d=1$, this theorem also guarantees that the Frobenius light cone expands no faster than $t\sim r^{\alpha-\frac{3}{2}}$ (up to logarithmic corrections).

The mathematical method used to prove the Frobenius light cone is based on an interpretation of the time evolution equation for operators as a many-body quantum walk governing the time evolution of a probability distribution.  By bounding the growth of expectation values in this probability distribution using techniques from classical probability theory, we constrain the growth of Eq.  (\ref{eq:frobeniusdef}).  This represents a radical shift in perspective compared with the conventional Lieb-Robinson theorem, which is based on applying the triangle inequality in an appropriate interaction picture (see e.g. Refs. \cite{Foss-Feig2015, Kuwahara2019}).

Since the Frobenius norm (squared) gives infinite temperature OTOCs, the Lieb-Robinson light cone \emph{is not relevant} for infinite temperature many-body quantum chaos and the growth of operators.  A careful determination of bounds on quantum chaos and operator spreading is essential for building on recent experimental progress in measuring OTOCs \cite{Garttner_2017,Li_2017} and quantum information scrambling \cite{landsmanVerifiedQuantumInformation2019} to design optimal information scramblers. Such work will be crucial in developing quantum simulators of holographic quantum gravity \cite{Maldacena:1997re}.

As emphasized before, many quantum state transfer tasks, including a ``background-independent" state transfer where $X_i(t)=X_j$, $Y_i(t)=Y_j$ and $Z_i(t)=Z_j$ (hence state $i$ is transferred to $j$ independently of all other qubits),   are constrained by the Frobenius light cone, which is tighter than the Lieb-Robinson light cone: see Theorem \ref{thm:mbqw}.  

\subsection{Free light cone}
Finally, we consider the light cone in systems of non-interacting particles.
While these systems are rich enough such that they sometimes lie beyond the regime of computability for classical computers, their dynamics can be essentially reduced to the motion of a single particle.
Returning to the same setup of Figure \ref{fig:interaction}, we may again estimate when the linear light cone fails by computing the weight of a single particle hopping the distance $\sim r$ from the ball $\mathcal{B}_1$ to $\mathcal{B}_2$ after time $t$:
\begin{equation}
    t\bigg\lVert \sum_{i\in\mathcal{B}_1,j\in\mathcal{B}_2} \frac{c^\dagger_j c_i}{r^\alpha}
    \bigg \rVert \lesssim \frac{t^{d+1}}{r^{\alpha}}.
\end{equation}
Here $c^\dag_i$ and $c_i$ are the creation and the annihilation operators for the non-interacting particles.  Following our previous logic, the free particle is constrained within a linear light cone when $\alpha>d+1$.
We rigorously prove that the free light cone is linear for $\alpha>d+1$ in Theorem \ref{thm:free}, and prove that no linear light cones exist for $\alpha<d+1$ in Theorem \ref{thm:statetransfer}.  When combined, these two theorems also prove that for $d<\alpha<d+1$, the form of the light cone is no worse than $t\sim r^{\alpha-d}$, and that no further improvement on the exponent $\alpha-d$ can be found.

Specifically, in Theorem \ref{thm:statetransfer} (Section \ref{sec:singlestatetransfer}), we show that this estimated growth rate is  achieved by a novel quantum state-transfer protocol involving a single particle.
The protocol works by successively spreading a particle to larger and larger regions of the lattice, each time doubling the number of sites sharing the particle (\cref{fig:statetransfer}).
Specifically, after the $k$th step of the protocol at time $t_k$, an operator $c_0^\dag$ originally supported at the origin would become \begin{equation}
    c^\dagger_0(t_k) \propto \sum_{\text{sites }x\text{ in a cube of length } O(2^k)} c^\dagger_x,
\end{equation}
where the precise set of sites $x$ is depicted in \cref{fig:statetransfer}.
After spreading the particle to a square large enough to cover both the origin and the target site, we simply reverse the protocol to concentrate the particle on the target site.
In each step of the protocol, the weaker interactions due to the power-law constraint are well compensated by the volume of the squares, making the protocol superlinear for all $\alpha<d+1$.
As emphasized in the introduction, this state-transfer protocol has (at least) two appealing features for experimental implementation, and could enhance the performance of quantum computing architectures assisted by long-range interactions \cite{linkeExperimentalComparisonTwo2017}.

\begin{figure}[t]
 \includegraphics[width=0.75\textwidth]{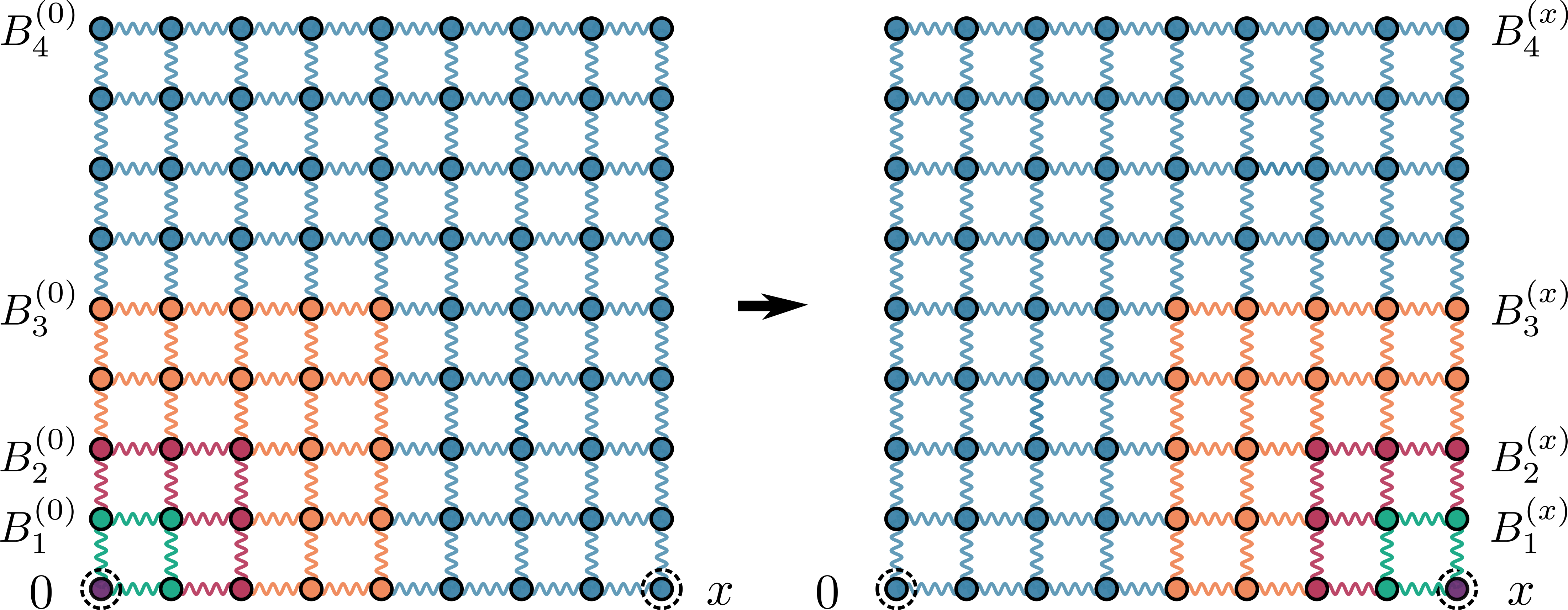} \caption{
 An illustration of our single-particle state-transfer protocol in $d=2$ dimensions.
 Through $4$ steps, we redistribute a particle initially at 0 to the square $B_1^{(0)}$, then to $B_2^{(0)},B_3^{(0)}$, and finally to $B_4^{(0)}$, each time doubling the size of the region sharing the particle.
 We use different colors to mark the additional sites that the particle spreads to in different steps.
 We then reverse the process to concentrate the particle on the target site $x$ and thereby achieve a perfect state transfer.
The protocol is enhanced by the volume $r^d$ of the squares, with $r$ being a typical size of the squares.
This enhancement offsets the penalty $1/r^\alpha$ due to the power-law constraint, resulting in a superlinear state-transfer protocol when $\alpha<d+1$.
 } \label{fig:statetransfer}
\end{figure}

The free light cone is also relevant for early time dynamics in low density models of interacting fermions or bosons, which are readily realized in experiment.  The observed slowdown of dynamics between the Lieb-Robinson and free light cones makes the Hubbard model exponentially easier to simulate in experimentally relevant regimes (e.g. polar molecules) at early times, with implications for demonstrating quantum supremacy (Section \ref{sec:boson-sampling}).

\section{Formal preliminaries}\label{sec:preliminaries}
We now more carefully introduce the problem that we address in this paper.  First, we will give a precise definition of a many-body quantum system with long-range interactions.  We need to first define the distance between two points.  Formally, we do so as follows. Let $\Lambda$ be the vertex set of a $d$-dimensional lattice graph with edge set $E_\Lambda$.   A lattice graph $(\Lambda,E_\Lambda)$ is a graph which is invariant under $d$-dimensional discrete translations: mathematically speaking, $\mathbb{Z}^d \subseteq \mathrm{Aut}(\Lambda,E_\Lambda)$, where Aut denotes the group of graph isomorphisms from $(\Lambda,E_\Lambda)$ to itself.  We assume that all vertices have finite degree in $E_\Lambda$, and that $|\Lambda/\mathbb{Z}^d| < \infty$, i.e. the unit cell has a finite number of vertices, and every vertex has a finite number of (nearest) neighbors.   This graph imbues a natural notion of distance, which we will use for the rest of the paper.  Let $\mathcal{D}: \Lambda \times \Lambda \rightarrow \mathbb{Z}^+$ denote the shortest path length between two vertices, also known as the Manhattan metric.

A many-body quantum system is then defined by placing a finite-dimensional quantum system (e.g. a qubit) on every vertex in $\Lambda$.  Formally we define a many-body Hilbert space \begin{equation}
    \mathcal{H} :=  \bigotimes_{i\in\Lambda} \mathcal{H}_i,
\end{equation}
where we assume that $\dim(\mathcal{H}_i)<\infty$.  In this paper, we will be especially interested in the dynamics of the operators acting on $\mathcal{H}$.  Let $\mathcal{B}$ denote the set of all Hermitian operators acting on $\mathcal{H}$.  $\mathcal{B}$ is a real vector space, and we denote operators $\mathcal{O}\in \mathcal{B}$ with $|\mathcal{O})$ whenever we wish to emphasize that they should be thought of as vectors.   A basis for $\mathcal{B}$ can be found as follows: let $T^a_i$ denote the generators of $\mathrm{U}(\dim(\mathcal{H}_i))$ where $a=0$ denotes the identity operator, which gives a complete basis for Hermitian operators on the local Hilbert space $\mathcal{H}_i$.   $\mathcal{B}$ is simply the tensor product of all these local bases of Hermitian operators: \begin{equation}
    \mathcal{B} = \mathrm{span}\left\lbrace \bigotimes_{i\in\Lambda} T^{a_i}_i, \; \text{for all }\lbrace a_i\rbrace  \right\rbrace. \label{eq:operatorbasis}
\end{equation}
For subset $X\subset\Lambda$, we define $\mathcal{B}_X$ to be the set of all basis vectors which act non-trivially only on the sites of $X$: \begin{equation}
    \mathcal{B}_X := \mathrm{span}\left\lbrace \bigotimes_{i\in X} T^{a_i}_i, \; \text{for all }\lbrace a_i \ne 0\rbrace  \right\rbrace.
\end{equation}
We define the projectors \cite{chen2019operator} \begin{equation}
    \mathbb{P}_i |\otimes T^{a_k}_k) := \left\lbrace \begin{array}{ll} |\otimes T^{a_k}_k) &\  a_i \ne 0 \\ 0 &\ a_i = 0 \end{array}\right.,
\end{equation}
which return the part of the operator that acts non-trivially on site $i$: \begin{equation}
    \mathbb{P}_i \mathcal{O} = \mathcal{O} - \frac{1}{\dim(\mathcal{H}_i)} \underset{i}{\mathrm{tr}} \;\mathcal{O}.
\end{equation}
For a general subset $X\subset\Lambda$, the projectors \begin{equation}
    \mathbb{P}_X := \sum_{Y\in\mathbb{Z}_2^X:|Y|>0} (-1)^{|Y|+1} \prod_{j\in Y}\mathbb{P}_j
\end{equation}
act similarly, and return the part of the operator which acts non-trivially on the subset $X$.
It was proven in Ref.~\cite{Chen2019} that when $|X|<\infty$, \begin{equation}
    \lVert \mathbb{P}_X \mathcal{O} \rVert_\infty \le 2 \lVert \mathcal{O}\rVert_\infty,
\end{equation}
where $\lVert \cdot\rVert_\infty$ is again the operator norm.
We will often drop the $\infty$ subscript for convenience.
In addition, we can relate the commutator in the Lieb-Robinson bound to the projection of an operator using the identity
\begin{align}
	\norm{\comm{\mathcal O_X,\mathcal O_Y}}\leq 2\norm{O_X}\norm{\mathbb P_X \mathcal O_Y},
\end{align}
which holds for all operators $\mathcal O_X \in \mathcal{B}_X, \mathcal O_Y \in \mathcal{B}_Y$.

We define the Hamiltonian $H: \mathbb{R}\rightarrow\mathcal{B}$ as \begin{equation}
    H(t) := \sum_{X\subset\Lambda} H_X(t), \label{eq:Hamiltoniandef}
\end{equation}
where $H_X(t):\mathbb{R}\rightarrow\mathcal{B}_X$.   $H(t)$ is said to be $q$-local if $H_X(t) =0$ for all $|X|>q$: physically speaking, the Hamiltonian operator contains at most $q$-body interactions.   The Hamiltonian generates time evolution on $\mathcal{B}$ according to the Heisenberg equation of motion for operators: we define the Liouvillian $\mathcal{L}(t)$ as the generator of time evolution, \begin{equation}
    \mathcal{L}(t)|\mathcal{O}) := |\mathrm{i}[H(t),\mathcal{O}]). \label{eq:liouvillian}
\end{equation}
We define the time evolved operator $\mathcal{O}(t) : \mathbb{R} \rightarrow\mathcal{B}$ as the solution to the differential equation \begin{equation}
    \frac{\mathrm{d}\mathcal{O}(t)}{\mathrm{d}t} :=  \mathcal{L}(t)\mathcal{O}(t), \;\;\;\; \mathcal{O}(0) := \mathcal{O}.
\end{equation}

  We say that the Hamiltonian $H$ has long-range interactions with exponent $\alpha$ if  \begin{equation}
    \alpha = \sup\left\lbrace \alpha_0 \in (0,\infty) : \text{ there exists } 0<h<\infty\text{ such that } \sum_{X:\lbrace i,j\rbrace \subseteq X} \lVert H_X(t) \rVert \le \frac{h}{\mathcal{D}(i,j)^{\alpha_0}}, \text{ for all } t\in\mathbb{R} \right\rbrace,  \label{eq:sec5alphadef}
\end{equation}
where $\mathcal{D}(i,j)$ denotes the distance between $i,j$.  In physics we often say that the interaction has exponent $\alpha$ when, assuming only two-body interactions, $H_{\lbrace i,j\rbrace} \lesssim h \mathcal{D}(i,j)^{-\alpha}$; strictly speaking though, any Hamiltonian with exponent $\alpha_2$, according to this loose definition, also has exponent $\alpha_1<\alpha_2$.  The formal definition \cref{eq:sec5alphadef} avoids this unwanted feature and assigns a unique exponent $\alpha$ to every problem.

The following identities, which we state without proof, will be useful in the discussion that follows: \begin{prop}[Sums over power laws \cite{Hastings2006, Tran2019a}] If $\alpha>d$, for any $\Lambda$ and $\mathcal{D}$, there exist $0<C_1,C_2<\infty$ such that:
    \begin{equation}
    \sum_{j\in\Lambda: \mathcal{D}(i,j)>r} \frac{1}{\mathcal{D}(i,j)^\alpha} < \frac{C_1}{r^{\alpha-d}},  \label{eq:C1sum}
\end{equation}\begin{equation}
   \sum_{k\in\Lambda \setminus\lbrace i,j\rbrace} \frac{1}{\mathcal{D}(i,k)^\alpha \mathcal{D}(j,k)^\alpha} < \frac{C_2}{\mathcal{D}(i,j)^{\alpha}}.  \label{eq:C2sum}
\end{equation}
\end{prop}

\section{Lieb-Robinson light cone}\label{sec:LR}
We begin by presenting the strictest light cone on the commutators of local operators, representing the generalization of the Lieb-Robinson theorem \cite{Lieb1972} to systems with long-range interactions.
\subsection{The linear light cone}\label{sec:LRbound}
The following proposition controls the growth of commutator norms in a Hamiltonian system with long-range interactions.

\begin{prop}\label{prop:multisite-bound}
Let $X,Y\subset \Lambda$ be disjoint with $\mathcal{D}(X,Y) := r$; $\mathcal O_X$ be an operator supported on $X$ obeying $\lVert \mathcal{O}_X\rVert = 1$; $\mathcal O_X(t)$ be the time-evolved version of $\mathcal O_X$ under a power-law Hamiltonian with an exponent $\alpha>2d+1$.  There exist constants $0<\bar v,c<\infty$ such that, for time evolutions $0 \leq t < r/\bar{v}$ generated by Eq.~(\ref{eq:liouvillian}) obeying Eq.~(\ref{eq:sec5alphadef}),
    \begin{align}
        \norm{\PP_Y\oket{\mathcal O_X(t)}}\leq c \abs{X} \frac{t^{d+1}\log^{2d}r}{(r-\bar vt)^{\alpha-d}}. \label{eq:prop1eq}
    \end{align}
\end{prop}
\begin{proof}
We begin by recalling the following theorem (recast in the language of projectors):
\begin{thm}[Linear light cone \cite{Kuwahara2019}] \label{thm:kuwahara}
Eq. (\ref{eq:prop1eq}) holds for a single-site operator, i.e. when $|X|=1$.
\end{thm}
While the proof presented in Ref.~\cite{Kuwahara2019} applied only to time-independent Hamiltonians, the generalization to time-dependent models is immediate from their results.  Next, we show the following general result.
\begin{lem}	\label{lem:one-to-many}
If for all  $x\in X$, $\lVert \mathbb{P}_Y  |\mathcal{O}_x(t))\rVert \le f(t,\mathcal{D}(x,Y))$, then there exist $0<K<\infty$ such that
\begin{equation}
    \lVert \mathbb{P}_Y |\mathcal {O}_X(t))\rVert \le K\sum_{x\in X} f(t,\mathcal{D}(x,Y)).
\end{equation}
\end{lem}
\begin{proof}
For pedagogical reasons, we demonstrate the proof on a system of spin-$1/2$ particles with $K = 9/2$.
However, the proof applies to any system with finite local Hilbert space dimensions \cite{Chen2019}.
Let $\{S_j:j=1,\dots,d_Y-1\}$ denote the $d_Y -1 = 4^{\abs{Y}}-1$ nontrivial Pauli strings supported on $Y$.
Then \cite{Chen2019}
\begin{align}
	\norm{\PP_{Y}\oket{\mathcal{O}_X(t)}}
	 &= \norm{\frac{1}{2d_Y}\sum_{j=1}^{d_Y-1}\comm{S_j,\comm{S_j,{\mathcal{O}_X(t)}}}} \leq \frac{1}{2d_Y}\sum_{j=1}^{d_Y-1} 2\underbrace{\norm{S_j}}_{=1}\norm{\comm{S_j,{\mathcal{O}_X(t)}}}\nonumber\\
	 &\leq \frac{1}{d_Y}\sum_{j=1}^{d_Y-1} \norm{\comm{\mathcal{O}_X,{S_j}(-t)}} \leq \frac{2}{d_Y}\sum_{j=1}^{d_Y-1} \norm{\PP_{X}\oket{S_j(-t)}}. \label{eq:breakX}
\end{align}
Next, we shall prove that
\begin{align}
	\norm{\PP_{X}\oket{S_j(-t)}} \leq 3\sum_{ x}\norm{ \PP_{ x}\oket{S_j(-t)}}.
\end{align}
To do so, we assign an (arbitrary) ordering of the sites in $X$: i.e. if $X=\lbrace x_1,\ldots, x_n\rbrace$, we choose $x_1<x_2<\cdots <x_n$.  Let $\tilde X_{x} = \{ x'\in X:x'>x\}$ be a subset of $X$ consisting of sites in $X$ that are greater than $x$.
We rewrite
\begin{align}
	\PP_X = \sum_{x} (1-\PP_{\tilde{X}_x})\PP_{x}, \label{eq:A4-fdsfknkc}
\end{align}
and therefore we have
\begin{align}
	\norm{\PP_X \oket{S_j(-t)}}
	&\leq \sum_{x}\norm{ (1-\PP_{\tilde{X}_x})\PP_{x}\oket{S_j(-t)}} \leq \sum_{x}3\norm{ \PP_{x}\oket{S_j(-t)}}.\label{eq:breakY}
\end{align}
In the last line, we have used that $\lVert \mathbb{P}_X \mathcal{O}\rVert \le 2 \lVert \mathcal{O}\rVert$ whenever $|X|<\infty$ \cite{Chen2019}, along with the triangle inequality. Plugging this back into the earlier equation, we have
\begin{align}
	\norm{\PP_{Y}\oket{\mathcal{O}_X(t)}}
	&\leq \frac{6}{d_Y}\sum_{j=1}^{d_Y-1} \sum_x \norm{\PP_{x}\oket{S_j(-t)}}
	\leq \frac{6}{d_Y}\sum_{j=1}^{d_Y-1} \sum_x \frac{1}{8}\sum_{P_x}  \norm{
	\comm{P_x,\comm{P_x, {S_j(-t)}}}}\nonumber\\
	&\leq\frac{3}{2} \underbrace{\frac{1}{d_Y}\sum_{j=1}^{d_Y-1}}_{\leq 1}\sum_x \sum_{P_x} \norm{\PP_{S_j} P_x(t)}
	\leq \frac{9}{2} \sum_x f(t,\dist(x,Y)),
\end{align}
where $P_x\in \{X_x,Y_x,Z_x\}$ denotes one of the three Pauli matrices on site $x$.
In the second from the last line, we have used the assumption $\norm{\PP_{S_{j}} P_x(t)}\leq f(t,\dist(x,Y))$.
\end{proof}
Combining Theorem \ref{thm:kuwahara} with Lemma \ref{lem:one-to-many} proves \cref{eq:prop1eq}, which is tighter than a result of Ref.~\cite{Kuwahara2019} when applied to general operators that are supported on many sites.
\end{proof}

\subsection{Fast operator spreading protocol}\label{sec:andy-protocol}
Proposition \ref{prop:multisite-bound} proves that the support of an operator $\mathcal{O}_i(t)$ is only large inside of a linear light cone when $\alpha>2d+1$. Our first main result is the following theorem, which proves the optimality (up to subalgebraic corrections) of that result.

\begin{thm}\label{thm:protocol}
	Let $\dim(\mathcal{H}_i)=2$ for all $i\in\Lambda$, and let $X_0$ and $X_r$ be two Pauli-$X$ operators supported on two sites $i$ and $j$ respectively, obeying $\mathcal{D}(i,j)=r$.
	For all $\alpha>d$, there exists a time-dependent Hamiltonian $H(t)$ obeying Eq.~(\ref{eq:sec5alphadef}) and constants $0<K,K^\prime<\infty$ such that for $3<t<K^\prime r^{\alpha/(1+2d)}$,
	\begin{align}
		\norm{\comm{X_0(t),X_r}} \geq K \frac{t^{1+2d}}{r^\alpha}.
	\end{align}
\end{thm}

\begin{proof}
We prove the theorem by constructing a fast operator-spreading protocol, which follows three steps, as depicted in Figure~\ref{fig:2dplus1protocol}.
In each step, we evolve the operator using a power-law Hamiltonian for time $t/3$.
For simplicity, we assume $t/3 := \ell \in \mathbb{Z}^+$, and assume that $\ell<\frac{1}{2}r$.

\emph{Step 1---.} In time $t/3$, we use a unitary $\mathcal U_1$ to spread the operator $X_0$ to $\prod_{i\in \mathcal B_\ell} X_i$, where $\mathcal B_\ell$ is a ball of radius $\ell$ centered at site 0.
We denote the volume of this ball by $\mathcal V := |\mathcal{B}_\ell|$.
The unitary $\mathcal U_1$ can be implemented using a series of controlled-NOT operators ($\CNOT$) among nearest neighbors in the lattice.  Note that a CNOT gate $U_{\mathrm{CNOT},i,j}$ for neighbors $i$ and $j$ acts as follows: \begin{equation}
    U_{\mathrm{CNOT},i,j}^\dagger X_i U_{\mathrm{CNOT},i,j} := X_i X_j.
\end{equation}
Under the conditions of Eq.~(\ref{eq:sec5alphadef}), this CNOT gate can be implemented in a time step of $O(1)$.

\begin{figure}[t]
 \includegraphics[width=0.85\textwidth]{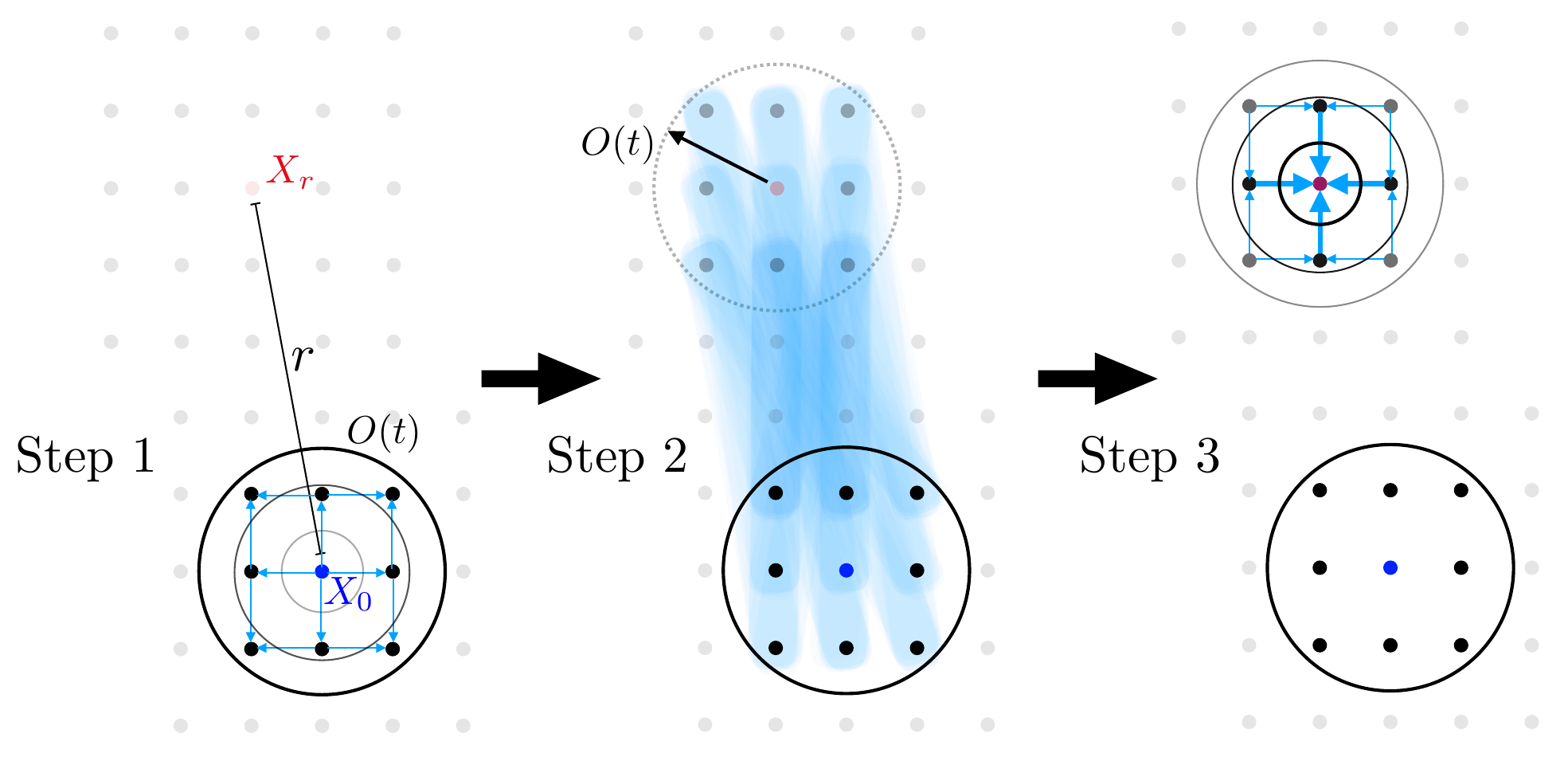} \caption{A protocol for rapid growth of the commutator norm using two-body long-range interactions.  Step 1: we use CNOT gates between nearest neighbor sites to spread a single Pauli $X_0$ to a Pauli string $XX\cdots X$ supported on every site inside a ball of radius $O(t)$ centered at $X_0$. Step 2: we use pairwise $ZZ$ interactions between all sites in the two balls, located distance $O(r)$ apart, which adds an operator of norm $\sim O(t^{2d+1}/r^{\alpha})$ into the second ball a distance $r$ away.  Step 3: we invert Step 1 in the outer ball, pushing all of the operator weight in the outer ball onto a single site.
 } \label{fig:2dplus1protocol}
\end{figure}

\emph{Step 2---.} In the next $t/3$ interval, we apply $\mathcal U_2 = \prod_{j\in\mathcal B_\ell}U_{j}(\tau)$ on the operator, where
\begin{subequations}\begin{align}
	U_{j}(\tau) &:= \cos(\tau\Theta)  + \mathrm{i}\sin(\tau\Theta) Z_j, \\
	\Theta &:= \sum_{k\in \tilde {\mathcal B}_\ell} Z_k,
\end{align}\end{subequations}
$\tilde {\mathcal B}_\ell$ is another ball of radius $d$ centered around the site at distance $r$,
and \begin{equation}
    \tau = \frac{t}{3(2r)^\alpha}. \label{eq:protocoltau}
\end{equation}
It is straightforward to verify that $U_{j}(\tau)$ is a unitary, since \begin{equation}
    U_j(\tau) = \exp\left[-\mathrm{i}\tau \sum_{y\in \tilde {\mathcal B}_\ell} Z_j Z_y\right].
\end{equation}
Since $Z_{j}Z_y$ commutes with $Z_{j'}Z_{y'}$ for all $j,j'\in\mathcal{B}_\ell$ and $y,y'\in\tilde{\mathcal{B}}_\ell$, $U_{j}(\tau)$ and $U_{j'}(\tau)$ can be implemented simultaneously.
In other words, the unitary $\mathcal U_2$ can be generated by a power-law Hamiltonian within time $t/3$: the factor of $2r$ in Eq.~(\ref{eq:protocoltau}) is present because the maximal distance between two sites in $\mathcal{B}_\ell$ and $\tilde{\mathcal{B}}_\ell$ is $r+2\ell<2r$.   The evolved version of the operator under this unitary is
\begin{align}
	 &\mathcal U_2^\dag \left(\prod_{j\in \mathcal B_\ell} X_j\right) \mathcal U_2
	  = \prod_{j\in\mathcal B_\ell} \Bigg[
	 \cos(2\tau \Theta) X_j
	 + \sin(2\tau\Theta) Y_j \Bigg].\label{eq:ProtocolStep2}
\end{align}

\emph{Step 3---.} In the final $t/3$, we apply a unitary $\mathcal U_3$ which is the inverse of $\mathcal U_1$, up to its action on $\tilde {\mathcal B}_\ell$ instead of $\mathcal B_\ell$.
It is easier to instead think of evolving the final operator $X_r$ under $\mathcal{U}_3^{-1}$, which does not change the commutator norm $\norm{\comm{X_0(t),X_r}}$.
	Therefore, after time $t$, we get the commutator norm:
	\begin{align}
		&\norm{\comm{X_1(t),X_r}}
		 = \norm{\comm{\mathcal U_2^\dag\mathcal U_1^\dag X_0 \mathcal U_1 \mathcal U_2,\mathcal U_3 X_r \mathcal U_3^\dag}}= \Bigg\Vert\underbrace{\comm{\prod_{j\in\mathcal B_\ell} \Bigg[
	 \cos(2\tau \Theta) X_j
	 + \sin(2\tau\Theta) Y_j \Bigg],\prod_{k\in\tilde{\mathcal B}_\ell} X_k}}_{\equiv C}\Bigg\Vert. \label{X1Xn}
	\end{align}

	To lower bound the norm of $C$, we consider the matrix elements of $C$ in the eigenbasis of Pauli $Z$ operators.
	We observe that $\bra{e}C\ket{00\dots0} = 0$ for all computational basis states $\ket{e}$ of the two balls except for $\ket{e} = \ket{11\dots1}$.  Hence,
	\begin{align}
		\bra{11\dots1}C\ket{00\dots0}
		&=  \Bigg[
	 \cos(2\tau \mathcal V)
	 - \ii\sin(2\tau\mathcal V) \Bigg]^{\mathcal V}
	 - \Bigg[
	 \cos(2\tau \mathcal V)
	 + \ii\sin(2\tau\mathcal V)  \Bigg]^{\mathcal V}\nonumber\\
	 &=-2\sum_{k~\text{odd}}^{\mathcal V} \binom{\mathcal V}{k} \ii^k\sin(2\tau \mathcal V)^k\cos(2\tau \mathcal V)^{\mathcal V-k} :=  a.
 	\end{align}
 	Therefore, $C$ is block diagonal and has eigenvalues $\pm \abs{a}$ in the sector $\{\ket{00\dots0}\},\{\ket{11\dots1}\}$.
    We note that to the lowest-order, $\abs{a} \approx \mathcal V^2 \tau \propto t^{2d+1}/r^\alpha$.
    Therefore, this operator-spreading protocol saturates the Lieb-Robinson bound in \cref{prop:multisite-bound}.

 	To make the statement rigorous, we lower bound the norm of $C$:
 	\begin{align}
 		\norm C \geq \abs{a}
 		&\geq 2\mathcal V \sin(2\tau \mathcal V) \cos(2\tau\mathcal V)^{\mathcal V-1
 		}
 		- \sum_{k~\text{odd}, k\geq 3}^{\mathcal V} \binom{\mathcal V}{k} \sin(2\tau\mathcal V)^k\cos(2\tau\mathcal V)^{\mathcal V-k} \nonumber\\
 		&\geq 2\mathcal V \sin(2\tau \mathcal V) \cos(2\tau\mathcal V)^{\mathcal V-1
 		}
 		- \frac{\mathcal V^3}{6}\sin(2\tau\mathcal V)^3\sum_{k~\text{even}}^{\mathcal V-3} \binom{\mathcal V-3}{k} \sin(2\tau\mathcal V)^k\cos(2\tau\mathcal V)^{\mathcal V-3-k} \nonumber\\
 		&\geq2\mathcal V \sin(2\tau \mathcal V) \cos(2\tau\mathcal V)^{\mathcal V-1
 		}
 		- \frac{\mathcal V^3}{6}\sin(2\tau\mathcal V)^3
 		[\sin(2\tau\mathcal V)+\cos(2\tau\mathcal V)]^{\mathcal V-3}.
 	\end{align}
 	Now we require that $\mathcal V^2\tau = \epsilon< 1/2$, which is equivalent to $t^{2d+1}\lesssim r^\alpha$.
 	Under this condition,\begin{subequations}
 	\begin{align}
 		&\cos(2\tau \mathcal V)^{\mathcal V-1} \geq \left(1 - \tau^2\mathcal V^2\right)^{\mathcal V} = \left(1-\frac{\epsilon^2}{\mathcal V^2}\right)^{\mathcal V} \geq \frac12, && (\text{for all }\mathcal V \geq 1 > \epsilon^2/10),\\
 		& \left[\sin(2\tau\mathcal V)+\cos(2\tau\mathcal V)\right]^{\mathcal V-3}
 		\leq \left(1+2\tau\mathcal V\right)^{\mathcal V} \leq \left(1+\frac{2\epsilon}{\mathcal V}\right)^{\mathcal V}\leq \mathrm{e}^{2\epsilon},\\
 		& \tau \mathcal V \leq \sin({2\tau\mathcal V}) \leq 2\tau \mathcal V.
 	\end{align}\end{subequations}
 	Therefore,
 	\begin{align}
 		\norm {C}
 		&\geq 2\mathcal V (\tau\mathcal V) \left[\frac{1}{2} - \frac{\mathcal V^2}{12}(2\tau \mathcal V)^2 \mathrm{e}^{2\epsilon}\right]
 		\geq  \mathcal V^2\tau \left(1-\frac{2}{3}\epsilon^2\mathrm{e}^{2\epsilon}\right)
 		\geq \frac{1}{2}\mathcal V^2 \tau \nonumber\\
 		&\geq \frac{1}{2}\left(\frac{t^{d}}{3^d}\right)^2 \frac{t}{3(2r)^\alpha}
 		\geq \frac{1}{3^{1+2d}2^{1+\alpha}}\frac{t^{2d+1}}{r^\alpha}.\label{eq:cp_def}
 	\end{align}
	This protocol shows that if the light cone of a Lieb-Robinson bound is $t\gtrsim r^{\kappa}$, then $\kappa \leq \alpha/(1+2d)$.
\end{proof}

Lastly, we note that it is trivial to remove the restriction $\dim(\mathcal{H}_i)=2$ from the assumptions of Theorem \ref{thm:protocol} by  simply making $H(t)$ act trivially on all but 2 of the basis states in each $\mathcal{H}_i$.

\emph{Note Added.}---After our manuscript was first posted, a second version of Ref.\ \cite{Kuwahara2019} appeared, which includes an analysis like Theorem \ref{thm:protocol} involving applying a very similar $H(t)$ to a particular initial state where all qubits are in $|0\rangle$, except for the qubit at site $0$, which starts in an arbitrary state $a|0\rangle + b|1\rangle$. The result, up to a Hadamard gate, is an encoding of the initial qubit into state $a|000\cdots\rangle_{{B}_\ell}|0\rangle_f + b|111\cdots\rangle_{{B}_\ell}|1\rangle_f$, where site $f$ is distance $r$ away from site $0$ and where all other qubits are in $|0\rangle$ \cite{Kuwahara2019}.

We add that, upon undoing the GHZ-like state on ${{B}_\ell}$ with nearest-neighbor CNOTs, one obtains $a|0\rangle_0|0\rangle_f + b|1\rangle_{0}|1\rangle_f$, where all other qubits are in $|0\rangle$. Reversing this full procedure with the roles of $0$ and $f$ exchanged, one obtains single-qubit state transfer from $0$ to $r$, i.e.\ all qubits end up in $|0\rangle$, except for qubit $f$, which ends up in state $a|0\rangle + b|1\rangle$. Using the language defined in the introduction, this state transfer, like that in Ref.\ \cite{Eldredge2017}, is not universal in that it assumes that all qubits involved, except for the one at site $0$, start in $\ket{0}$.  This protocol takes time $O(r^{\alpha/(2d+1)})$ and is faster than that of Ref.\ \cite{Eldredge2017} for $d+\frac{1}{2}<\alpha<2d+1$. In particular, for these values of $\alpha$, this protocol can be used to shorten, relative to the result of Ref.\ \cite{Eldredge2017}, the preparation time of MERA (multiscale entanglement renormalization ansatz) states of linear size $L$ down to $O(L^{\alpha/(2 d + 1)})$.

For $d < \alpha < d+1$, the state transfer protocol just presented can be sped up with the help of the protocol of Ref.\ \cite{Eldredge2017}. In particular, we can use the protocol from Ref.\ \cite{Eldredge2017} to encode in time $t$ a qubit $a \ket{0} + b \ket{1}$ at site $0$ into a GHZ-like state $a \ket{000 \dots } + b \ket{111 \dots}$ of $\sim t^{d/(\alpha - d)}$ qubits, as compared to $\sim t^d$ qubits in the original protocol above. The same procedure can be used to prepare the GHZ state around site $f$, while the reverse of the procedure can be used to undo the preparation of these GHZ-like states around both site $0$ and site $f$. With this enhancement on the preparation and the undoing of the GHZ-like states, the state transfer protocol takes time $t \sim r^{\alpha (\alpha-d)/(\alpha+d)}$, is faster than both of the original protocols, and can be used to prepare MERA states of linear size $L$ in time $O(L^{\alpha (\alpha-d)/(\alpha+d)})$.

\subsection{Growth of connected correlators}\label{sec:connected-corr}

In this subsection, we explore how fast connected correlators can be generated using a power-law Hamiltonian.
In particular, we use the Lieb-Robinson bounds to show that the growth of connected correlators is constrained to linear light cones for all $\alpha>2d+1$.
In contrast, for $\alpha<2d+1$, we construct---based on our protocol in \cref{thm:protocol}---an explicit example where the growth of connected correlators is not contained within any linear light cone.

We consider a $d$-dimensional lattice $\Lambda$ and a power-law Hamiltonian $H(t)$ with an exponent $\alpha$.
Let $C$ denote a plane that separates $\Lambda$ into two disjoint subsets $L$ and $R$, with $L\cup R=\Lambda$.
Let $A$ and $B$ be two unit-norm operators supported on single sites $x\in L$ and $y\in R$ respectively such that $\dist(x,C),\dist(y,C)>r/2$.
Finally, let $\ket{\psi}$ be a product state between the sublattices $L,R$.
Our object is the connected correlator
\begin{align}
	\mathcal C(t,r) = \avg{A(t)B(t)} - \avg{A(t)}\avg{B(t)},
\end{align}
where $\avg{\cdot} = \bra{\psi}\cdot\ket{\psi}$ and $A(t)$ is the time-evolved version of $A$ under $H$.
While the correlator vanishes at time zero due to the disjoint supports of $A$ and $B$,
it may grow with time as the operators spread under the evolution.

First, we show that $\mathcal C(t,r)$ obey a linear light cone for all $\alpha>2d+1$.
Our strategy is to approximate $A(t)$ by an operator $\tilde A$ supported on a ball of radius $r/2$ centered on $x$ and $B(t)$ by $\tilde B$ supported on a ball of the same radius but centered on $y$.
Since $\tilde A$ and $\tilde B$ have disjoint supports, the connected correlator between them vanishes.
Therefore, the connected correlator between $A(t)$ and $B(t)$ is upper bounded by the errors of the approximations:
\begin{align}
	\mathcal C(t,r) \leq 2\norm{A(t) - \tilde A} + 2\norm{B(t) - \tilde B},
\end{align}
which in turn depend on the constructions of $\tilde A,\tilde B$.

Let $S_A$ contain sites that are at most a distance $r/2$ away from $x$ and $S_A^c$ be all other sites in the lattice.
We construct $\tilde A$ by simply tracing out the part of $A(t)$ that lies outside $S_A$~\cite{Bravyi2006}:
\begin{align}
	\tilde A = \mathrm{tr}_{S_A^{\mathrm{c}}} [A(t)],
\end{align}
where the partial trace is taken over $S_A^c$.
It follows from the definition that $\tilde A$ is supported entirely on $S_A$.
Proposition~\ref{prop:multisite-bound} provides a bound on the error in approximating $A(t)$ by $\tilde A$: there exists a velocity $u$ such that when $r>ut$,
\begin{align}
	\norm{A(t) - \tilde A}
	&\le K \frac{t^{d+1}\log^{2d}r}{r^{\alpha-d}},
\end{align}
for some constant $0<K<\infty$.
Upper bounding the error in approximating $B(t)$ by $\tilde B$ in a similar way, we obtain a bound on the connected correlator:
\begin{align}
	\mathcal C(t,r) \le 4K  \frac{t^{d+1}\log^{2d}r}{r^{\alpha-d}}.
\end{align}
As a result, the light cone of the connected correlator is linear, with velocity no larger than $u$, for $\alpha>2d+1$.

We now provide an example of superlinear growth of connected correlators for $\alpha < 2d+1$
using a slightly altered protocol than that discussed in Section~\ref{sec:andy-protocol}.
In particular, we consider initial operators $A = X_{x}$ and $B = Z_{y}$ supported on $x,y$ respectively such that $\dist(x,y) = r$.

The protocol works as follows.
In the first step of the protocol (again in time $t/3$), we still apply $\mathcal{U}_{1}$ in order to spread $X_{x}$ to $\prod_{i \in \mathcal{B}_\ell} X_{i}$, where $\mathcal{B}_\ell$ is a ball of radius $\ell = t/3$ centered on $x$.
Since $\mathcal{U}_1$ acts trivially on $\tilde{\mathcal{B}}_\ell$ (the ball of radius $\ell$ centered on $y$), we can simultaneously apply a locally rotated version of ${\mathcal{U}}_{1}$ in $\tilde{\mathcal{B}}_\ell$ to spread $Z_{y}$ to $\prod_{\tilde{i} \in \tilde{\mathcal{B}}_\ell} Z_{\tilde{i}}$.
In the next time $t/3$, we again apply $\mathcal{U}_{2}$, which takes $\prod_{i \in \mathcal{B}_\ell} X_{i}$ to the expression in \cref{eq:ProtocolStep2}.
Note that this evolution does not change $\prod_{\tilde{i} \in \tilde{\mathcal{B}}_\ell} Z_{\tilde{i}}$ as it commutes with $\mathcal{U}_{2}$.
For the last $t/3$ we simply do nothing.

Define the state
$\ket{\psi} = \ket{\phi}_{\mathcal{B}_\ell}\ket{\phi}_{\tilde{\mathcal{B}}_\ell},$
where
\begin{equation}
\ket{\phi}_{\mathcal{B}_\ell}
= \frac{1}{\sqrt{2}}
\big(\underbrace{\ket{0 \cdots 0}_{\mathcal{B}_\ell}}_{\equiv \ket{\bar{0}}_{\mathcal{B}_\ell}} + \ii \underbrace{\ket{1 \cdots 1}_{{\mathcal{B}}_\ell}}_{\equiv \ket{\bar{1}}_{\mathcal{B}_\ell}}\big)
\end{equation}
is a state of the sites in $\mathcal B_\ell$.
Throughout our analysis, we will often dispense with the subscripts, but the Hilbert space in question should be clear from context, and we will always list the state on $\mathcal{B}_\ell$ before that for $\tilde{\mathcal{B}}_\ell$.

We will calculate the connected correlator
\begin{align}
	 \mathcal C(t,r) = \avg{Z_{y}(t)X_{x}(t)} - \avg{Z_{y}(t)}\avg{X_{x}(t)},
\end{align}
where $\avg{\cdot} = \bra{\psi}\cdot\ket{\psi}$ and $X_x(t),Z_{y}(t)$ are the operators evolved under the unitaries described above.
Assume for simplicity that $t$ is such that $\mathcal{V}$---the volume of $\mathcal B_\ell$---is odd.
It is straightforward to show that $\avg{Z_{y}(t)} = 0$ and therefore the second term $\mathcal C(t,r)$ vanishes.
Next, we have
\begin{align}
X_{x}(t)\ket{\psi} &= \prod_{j\in\mathcal B_\ell} \left[
	 \cos(2\tau \Theta) X_j
	 + \sin(2\tau\Theta) Y_j \right] \ket{\psi}
     \nonumber\\
	 &= \frac{1}{\sqrt{2}}\prod_{j\in\mathcal B_\ell} \left[
	c X_j
	 + s Y_j \right] \ket{\phi}\ket{\bar{0}} + \frac{\ii}{\sqrt{2}} \prod_{j\in\mathcal B_\ell}\left[
	c X_j
	 - s Y_j \right] \ket{\phi}\ket{\bar{1}} \nonumber\\
	 &= \frac{1}{2}\bigg[(c + \ii s)^{\mathcal{V}}\ket{\bar{1}}\ket{\bar{0}} + \ii (c - \ii s)^{\mathcal{V}}\ket{\bar{0}}\ket{\bar{0}} + \ii (c - \ii s)^{\mathcal{V}}\ket{\bar{1}}\ket{\bar{1}} - (c + \ii s)^{\mathcal{V}}\ket{\bar{0}}\ket{\bar{1}}\bigg],
\end{align}
where $c = \cos(2\tau\mathcal{V})$ and $s = \sin(2\tau\mathcal{V})$.
Next note that:
\begin{align}
\bra{\psi}Z_{y}(t) &= \frac{1}{\sqrt{2}}(\bra{\phi}\bra{\bar{0}} - \ii\bra{\phi}\bra{\bar{1}})Z_{y}(t) = \frac{1}{\sqrt{2}}(\bra{\phi}\bra{\bar{0}} + \ii\bra{\phi}\bra{\bar{1}}) = \frac{1}{2}\left(\bra{\bar{0}}\bra{\bar{0}} + \ii \bra{\bar{0}}\bra{\bar{1}}) - \ii \bra{\bar{1}}\bra{\bar{0}} + \bra{\bar{1}}\bra{\bar{1}}\right).
\end{align}
Thus:
\begin{align}
\mathcal C(t,r) = \avg{Z_{y}(t)X_{x}(t)}
&= \frac{\ii}{2}((c - \ii s)^{\mathcal{V}} - (c + \ii s)^{\mathcal{V}}) \geq \frac{1}{3^{1+2d}2^{2+\alpha}}\frac{t^{2d+1}}{r^\alpha},
\end{align}
where we have used the bound \cref{eq:cp_def}.
Therefore, this demonstrates the connected correlators may grow along a superlinear light cone for all $\alpha<2d+1$.

We note that in our setting, we only assume the initial state is a bipartite product state across the cut $C$.
Our bound therefore also applies to a more restrictive scenario where the initial states are fully product.
However, it is not clear whether the bound can be saturated in this scenario.

\subsection{Simulation of local observables}\label{sec:local-obs}

In this subsection, we use the Lieb-Robinson bounds to improve the estimation of local observables in time-evolved states.
Given an initial state $\ket{\psi}$ and a power-law Hamiltonian $H$, we consider the task of estimating the expectation value of the time-evolved observable $\avg{A(t)}:=\bra{\psi}U(t)^\dag AU(t)\ket{\psi}$, where $U(t)$ is the unitary generated by $H$ at time $t$, for a local operator $A$.
The ability to perform this task for any arbitrary local observable is equivalent to the ability to compute local density matrices of the time-evolved state $U(t)\ket{\psi}$ or the ability to sample local observables in $U(t)\ket{\psi}$.

A typical approach to estimating $\avg{A(t)}$ is as follows.
First the unitary evolution $U(t)$ on the entire system is decomposed into a more tractable sequence of elementary unitaries that are supported on a smaller number of sites to produce an approximation to the time-evolved state $\ket{\psi(t)}$.
The expectation value $\avg{A(t)}$ is then computed by simulating measurements of $A$ on this state.
The number of elementary unitaries in the decomposition of $U(t)$ typically increases with both time $t$ and the number of sites $N$ in the system.

However, in the Heisenberg picture, the intuition from the Lieb-Robinson bounds suggests that the dynamics of $A(t) = U(t)^\dag A U(t)$ is mostly confined inside some light cones and, therefore, it should be sufficient to simulate the unitary generated by the Hamiltonian inside the light cones alone.
The following result provides such an approximation.
\begin{prop}\label{lem:loc-ob-app}
	Let $H$ be a 2-local power-law Hamiltonian (i.e. the sets $X$ in Eq.~(\ref{eq:Hamiltoniandef}) obey $|X|=2$) of exponent $\alpha>2d+1$, and $H_r$ be a Hamiltonian constructed from $H$ by taking only interaction terms supported entirely on sites inside a ball of radius $r\geq 4\bar vt\geq 1$ around the support of the single-site operator $A$ (where $\bar v$ is the same constant as in \cref{prop:multisite-bound}).
	Let $A(t)$ and $\tilde A(t)$ be the versions of $A$ evolved for time $t$ under $H$ and $H_r$, respectively.
	Then there exists $0<K<\infty$ such that
	\begin{align}
	 	 \norm{A(t) - \tilde A(t)} =K\frac{t^{d+2}\log^{2d}r}{r^{\alpha-d}}. \label{eq:propsimulatinglocal}
	 \end{align}
\end{prop}

\begin{proof}
Without loss of generality, assume that $A$ is initially supported at the origin.
Using the triangle inequality, we bound the difference between $A(t)$ and $\tilde A(t)$:
\begin{align}
	\norm{A(t) - \tilde A(t)}
	&\leq \int\limits_0^t \mathrm{d}s \norm{\comm{H-H_r,\tilde A(s)}}\leq \int\limits_0^t \mathrm{d}s \sum_{i:\dist(i,0)\leq r}\norm{\comm{\sum_{j:\dist({j,0})>r} h_{i,j},\tilde A(s)}}\label{eq:loc-ob-diff}.
\end{align}
We then use the bounds in Refs.~\cite{Hastings2006,Kuwahara2016} to bound the commutator norm $\norm{\comm{\sum_{j:\dist({j,0})>r} h_{  i,  j},\tilde A(s)}}$.
For that, we separate the sums over $i$ into terms corresponding to $i$'s inside and outside the linear light cone defined by $\dist(i,0) = 2\bar v s$.

For $i$ such that $\dist(i,0)\leq 2\bar v s$, we simply use another triangle inequality for the sum over $j$ and bound $\norm{\comm{h_{i,j},\tilde A(s)}}\leq 2/\dist(i,j)^\alpha$.
Note that in this case, we have $\dist(i,j)\geq \dist(j,0)-2\bar v s \geq \dist(j,0)/2$.
Therefore, we have
\begin{align}
	\sum_{j:\dist({j,0})>r}\sum_{i:\dist({i,0})\leq 2\bar vs}\norm{\comm{h_{i,j},\tilde A(s)}}
	&\leq 4^{d}2^{\alpha+1} \bar v^{d} \sum_{j:\dist({j,0})>r}\frac{s^{d}}{\dist(0,j)^\alpha} \le \frac{Ks^d}{r^{\alpha-d}},\label{eq:loc-ob-bound-in-dsa}
\end{align}
for some constant $0<K<\infty$.
On the other hand, for $i$ such that $r\geq \dist(i,0)> 2\bar v s$, we further divide into two cases: $s\geq 1$ and $s<1$.
For $s\geq 1$, we use \cref{prop:multisite-bound} (note that $A$ is a single site operator):
\begin{align}
	\sum_{i:r\geq\dist({i,0})> 2\bar vs} \norm{\comm{\sum_{j:\dist({j,0})>r}h_{  i,  j},\tilde A(s)}}
	&\le K_1 \sum_{i:r\geq\dist({i,0})> 2\bar vs} \frac{1}{[r-\dist(i,0)]^{\alpha-d}}\frac{s^{d+1}\log^{2d}r}{\dist(i,0)^{\alpha-d}}\nonumber\\
	&\le K_2 s^{d+1} \frac{\log^{2d}r}{r^{\alpha-d}} ,\label{eq:loc-ob-bound-out}
\end{align}
where we have used Eq.~(\ref{eq:C2sum}) and defined another set of constants $0<K_{1,2}<\infty$.  Similarly, for $s<1$, we use a bound from Ref.~\cite{Hastings2006} to show that there exists $0<K_{3}<\infty$ such that
\begin{align}
	\sum_{j:\dist({j,0})>r}\sum_{i:\dist({i,0})> 2\bar vs} \norm{\comm{h_{  i,  j},\tilde A(s)}} \le \frac{K_3}{r^{\alpha-d}}.\label{eq:loc-ob-bound-out2}
\end{align}
Substituting Eq.~(\ref{eq:loc-ob-bound-in-dsa}), Eq.~(\ref{eq:loc-ob-bound-out}), and Eq.~(\ref{eq:loc-ob-bound-out2}) into Eq.~(\ref{eq:loc-ob-diff}) and integrating over time, we obtain Eq.~(\ref{eq:propsimulatinglocal}).
\end{proof}

We now analyze the cost of estimating $\avg{A(t)}$ using quantum algorithms, although we note that \cref{lem:loc-ob-app} applies equally well to classical simulation algorithms.
For simplicity, we assume that the Hamiltonian is time-independent in the following discussion.
In order for the error of the approximation to be at most a constant, we choose
\begin{align}
  	 r \propto \max\left\{t^{\frac{d+2}{\alpha-d}}\log t,t\right\}.
\end{align}
Therefore, to estimate $\avg{A(t)}$, it is sufficient to simulate the evolution of $\tilde A(t)$ on $N_r \propto r^d$ sites (instead of simulating the entire lattice).

We then compute $\avg{\tilde A(t)}$ by simulating $\mathrm{e}^{-\ii H_rt}$ using one of the existing quantum algorithms.
Using the $p$th-order product formula for simulating power-law Hamiltonians~\cite{Childs2019a}, we need
\begin{align}
	O\left((N_rt)^{\frac{\alpha}{\alpha-d}+o(1)}\right)
	=
 \max\left\lbrace O\left(t^{\frac{\alpha(\alpha+d+d^2)}{(\alpha-d)^2}+o(1)}\right),O\left(t^{\frac{\alpha(1+d)}{\alpha-d}+o(1)}\right)\right\rbrace
\end{align}
elementary quantum gates, where $o(1)$ denotes $p$-dependent constants that can be made arbitrarily small by increasing the order $p$.
For all $\alpha>2d+1$, this gate count is less than the estimate without using the Lieb-Robinson bound in Ref.~\cite{Childs2019a}.
In particular, in the limit $\alpha\rightarrow\infty$, the gate count reduces to $O\left(t^{1+d+o(1)}\right)$, which corresponds to the space-time volume inside a linear light cone.

We note that in estimating the gate count for computing $\avg{A(t)}$, we have implicitly assumed that we have access to many quantum copies of the initial state $\ket{\psi}$.
However, in scenarios where only a classical description of $\ket{\psi}$ is provided, we need to add the cost of preparing $\ket{\psi}$ to the total gate count of the simulation.

\section{Frobenius light cone}\label{sec:fro}
We now turn to the Frobenius light cone.  To motivate the development, let us consider the early time expansion of a time evolving operator $\mathcal{O}_i$, initially supported on lattice site $i$: \begin{equation}
    \mathcal{O}_i(t) = \sum_{n=0}^\infty \frac{(\mathcal{L}t)^n}{n!} \mathcal{O}_i = \mathcal{O}_i + \mathrm{i}t[H,\mathcal{O}_i] - \frac{t^2}{2}[H,[H,\mathcal{O}_i]] + \cdots
\end{equation}
For illustrative purposes, we have temporarily assumed $H$ is time-independent.  Suppose further that $H$ only contains nearest neighbor interactions.  Then $[H,\mathcal{O}_i]$ can only contain operators of the form $\mathcal{O}_{i-1}\mathcal{O}\mathcal{O}_{i+1}$, and $[H,[H,\mathcal{O}_i]]$ can contain terms no more complicated than $\mathcal{O}_{i-2}\mathcal{O}_{i-1}\mathcal{O}\mathcal{O}_{i+1}\mathcal{O}_{i+2}$, and so on.  It is natural to ask ``how much" of the operator can be written as a sum of products of single-site operators restricted to some given subset of the lattice $\Lambda$.  This question is naturally interpreted as follows:  upon expanding $|\mathcal{O}_i(t))$ in terms of the basis vectors of \cref{eq:operatorbasis}: \begin{equation}
    |\mathcal{O}_i(t)) :=  \sum_{\lbrace a_k\rbrace} c_{\lbrace a_k\rbrace}(t)|\bigotimes_k T^{a_k}_k),
\end{equation}
the coefficients $c_{\lbrace a_k\rbrace}(t)$ are analogous to the probability amplitudes of an ordinary quantum mechanical wave function.   As we will see, the coefficients $c_{\lbrace a_k\rbrace}(t)$ must be sufficiently small if any $a_k$ are non-identity, when the sites $i$ and $k$ are sufficiently far apart, at any fixed time $t$: this is, intuitively, what we will call the Frobenius light cone.

For mathematical convenience in the discussion that follows, we restrict our analysis to finite lattices.  It appears straightforward, if slightly tedious, to generalize to infinite lattices through an appropriate limiting procedure.  More significantly, we will focus our discussion to one-dimensional lattices, as only in one dimension have we developed the machinery powerful enough to qualitatively improve upon the results of Section \ref{sec:LR}.

\subsection{A vector space of operators}
We define a one dimensional lattice \begin{equation}
\Lambda := \lbrace i \in \mathbb{Z} :  0 \le i \le L\rbrace.
\end{equation}
For every site $i\in \Lambda$, we assume a finite dimensional local Hilbert space $\mathcal{H}_i$, obeying $\dim(\mathcal{H}_i)<\infty$.  The global Hilbert space is \begin{equation}
\mathcal{H} := \bigotimes_{i\in\Lambda} \mathcal{H}_i.
\end{equation}

Let $\mathcal{B}$ denote the set of Hermitian operators acting on $\mathcal{H}$.  We equip this space with the Frobenius inner product \begin{equation}
(A|B) := \frac{\mathrm{tr}(A B)}{\dim(\mathcal{H})} ,
\end{equation}
upon which $\mathcal{B}$ becomes a real inner product space;  we denote elements of this vector space $\mathcal{O}\in\mathcal{B}$ as $|\mathcal{O})$.
When $A = B$, the inner product reduces to the squared Frobenius norm of $A$: $(A|A) = \norm{A}_\textrm{F}^2$.
Note that for traceless operators $A$ and $B$, this inner product corresponds to the value of the thermal two-point connected correlation function at infinite temperature.
Let $\lbrace T^a_i \rbrace$ denote the generators of $\mathrm{U}(\dim(\mathcal{H}_i))$, with $a=0$ denoting the identity matrix.  These generators form a complete basis for $\mathcal{B}$: \begin{equation}
\mathcal{B} := \mathrm{span}\left\lbrace \bigotimes_{i\in \Lambda} T^{a_i}_i  \right\rbrace := \mathrm{span} \left\lbrace |a_{0}\cdots a_L ) \right\rbrace.
\end{equation}
We define the projectors \begin{equation}
\mathbb{Q}_x |a_{0}\cdots a_L) := \left\lbrace \begin{array}{ll}  |a_{0}\cdots a_L) &\ a_x \ne 0 \text{ and } a_y = 0 \text{ if } y>x \\ 0 &\ \text{otherwise} \end{array} \right..
\end{equation}
Hence $\mathbb{Q}_x$ selects the parts of an operator which all act on $x$, but on no site to the right of $x$.  Observe that we have orthogonality
and completeness: \begin{equation}
\mathbb{Q}_i\mathbb{Q}_j = \delta_{ij} \mathbb{Q}_j,\;\;\;\;\;\; \sum_{i\in \Lambda}\mathbb{Q}_i = 1.  \label{eq:sumQi1}
\end{equation}

Time evolution is generated by a (generally time-dependent) Hamiltonian $H(t): \mathbb{R} \rightarrow \mathcal{B}$.  We assume that $H$ is 2-local: \begin{equation}
H(t) = \sum_{\lbrace i,j\rbrace \subset \Lambda} H_{ij}(t) , \label{eq:sec4H}
\end{equation}
with power-law interactions of exponent $\alpha$.
By unitarity, \begin{equation}
(\mathcal{O}|\mathcal{L}(t)|\mathcal{O}) = 0,
\end{equation}
where $\mathcal{L}(t)$ was defined in Eq.~(\ref{eq:liouvillian});  hence $\mathcal{L}(t)$ generates orthogonal transformations on $\mathcal{B}$ and leaves the length of all operators invariant.

\subsection{The operator quantum walk}
Our goal is to understand the following scenario (\cref{fig:quantumwalk}):  given an operator $|\mathcal{O})$ starting at the left-most site, i.e. obeying $\mathbb{Q}_0|\mathcal{O}) = |\mathcal{O})$, how long does it take before most of the operator $|\mathcal{O}(t))$ consists of operator strings that act on sites $ \ge x$?  More precisely, define \begin{equation}
t_2^\delta(x) := \inf \left\lbrace t>0 : \text{for any } \mathbb{Q}_0 |\mathcal{O}_0) = |\mathcal{O}_0), \;\;\;   \delta  < \frac{\sum_{y: x\le y\le L} (\mathcal{O}_0(t)|\mathbb{Q}_y|\mathcal{O}_0(t)) }{(\mathcal{O}_0|\mathcal{O}_0)}  \right\rbrace
\end{equation}
to be the shortest time for which a fraction $\delta$ of the operator $|\mathcal{O}(t))$ can be supported on sites $\geq x$.   The assumption that the operator starts only on the left-most site is not restrictive---for an initial site $k \in \Lambda$, we can identify the lattice sites $k+m\sim k-m$ in order to ``fold" the one dimensional lattice to put the initial point $k$ at one boundary; such a change cannot modify Eq.~(\ref{eq:sec5alphadef}), except to adjust the value of $h$ by a factor $<4$.

We note that
\begin{align}
\label{eq:OTOC}
     \sup_{\mathcal O_x \in \mathcal{B}_x} \frac{\mathrm{tr}\left({\comm{\mathcal O_0(t),\mathcal O_x}^\dag\comm{\mathcal O_0(t),\mathcal O_x}}\right)}{\dim(\mathcal H)(\mathcal O_0|\mathcal O_0)}
     \leq 4\frac{(\mathcal O_0(t)|\mathbb P_x|\mathcal O_0(t))}{(\mathcal O_0|\mathcal O_0)}
     \leq 4\sum_{y:x\leq y\leq L}\frac{(\mathcal O_0(t)|\mathbb Q_y|\mathcal O_0(t))}{(\mathcal O_0|\mathcal O_0)},
\end{align}
where the left-most side corresponds to the out-of-time-order correlation function (OTOC) of an infinite-temperature state---a quantity known to herald the onset of many-body quantum chaos \cite{Shenker2014, Maldacena2016}.
From \cref{eq:OTOC}, it follows that a lower bound on $t_2^\delta(x)$ also bounds the evolution time of the OTOC and the growth of chaos.

The second main result of this paper is the following theorem:
\begin{thm}\label{thm:mbqw}
Given Hamiltonian evolution on $\mathcal{H}$ obeying Eq.~(\ref{eq:sec4H}) and Eq.~(\ref{eq:sec5alphadef}),  for any $x\in\Lambda$, $0<\delta\in\mathbb{R}$ and $\frac{3}{2} < \alpha \in \mathbb{R}$, there exist constants $0<K,K^\prime<\infty$ such that \begin{equation}
t_2^\delta(x) \ge K \times \left\lbrace \begin{array}{ll} x &\ \alpha > \frac{5}{2} \\ x^{\alpha - 3/2} (1+K^\prime\log x)^{-1} &\  \frac{3}{2} < \alpha \le \frac{5}{2} \end{array}\right..  \label{eq:sec4thm}
\end{equation}
\end{thm}
\begin{figure}[t]
 \includegraphics[width=0.85\textwidth]{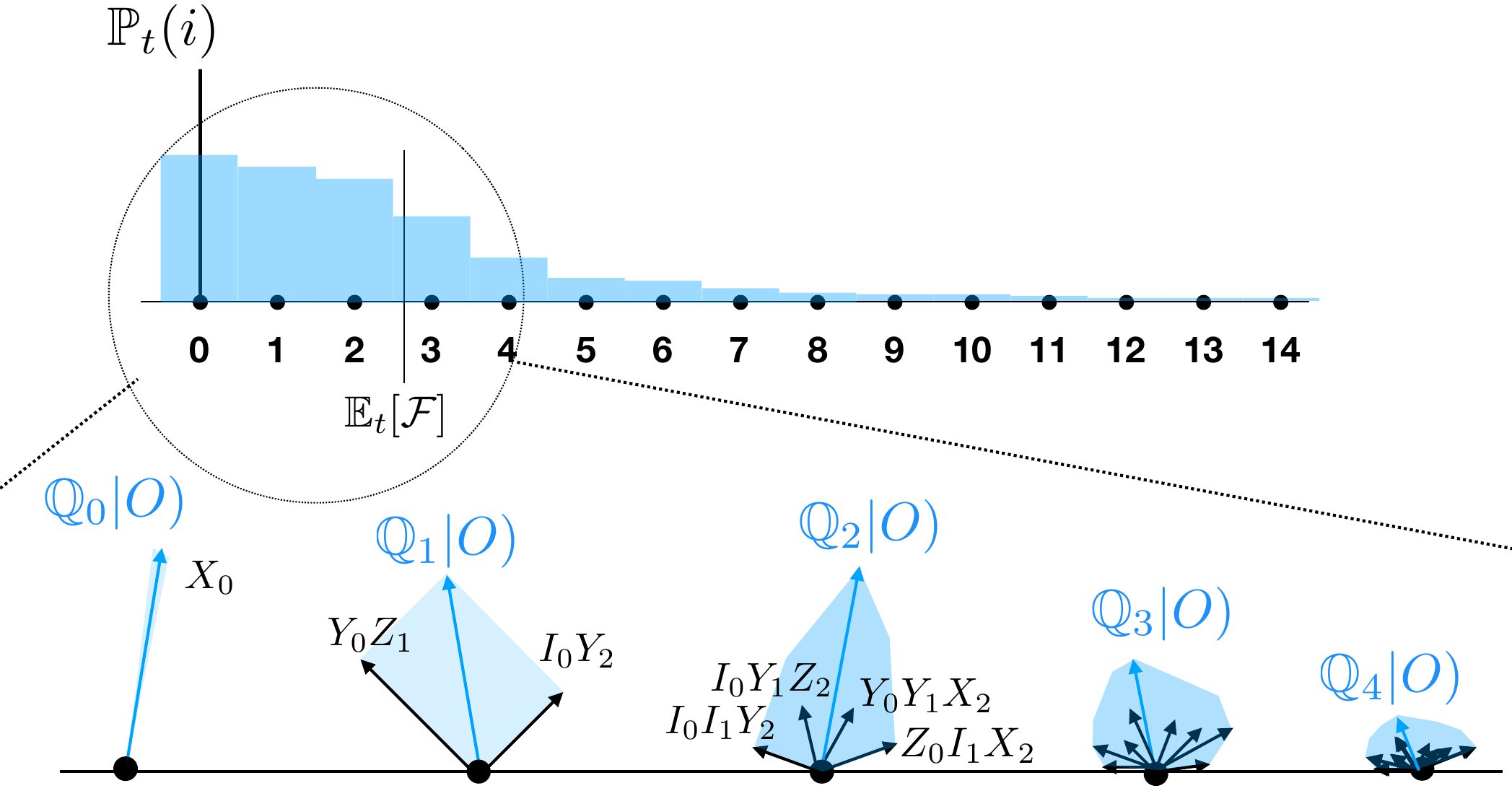} \caption{The $4^{L}$ dimensional space of operators can be decomposed into direct sum of $L$ subspaces $\{\mathbb{Q}_i\}$ by the position of the right-most occupied site.  By keeping track of only the ``average value" of the right-most site (depicted above), keeping in mind that an exponential number of orthogonal operators (depicted below) are contained on most of the sites, we reduce the quantum walk of many-body operators from an exponentially large space to a one dimensional line.
 } \label{fig:quantumwalk}
\end{figure}

\begin{proof}
We prove this theorem using the ``operator quantum walk" formalism introduced in Ref.~\cite{Lucas2019}.  For simplicity, we will first prove the theorem when $\alpha > \frac{5}{2}$, and then generalize to $\alpha\le\frac{5}{2}$ afterwards.   Consider the operator $\mathcal{F}$ acting on $\mathcal{B}$ defined by \begin{equation}
\mathcal{F} :=  \sum_{j\in \Lambda} j \mathbb{Q}_j.
\end{equation}
Our goal is to show that
 \begin{equation}
\lim_{L\rightarrow\infty} \lVert [\mathcal{F},\mathcal{L}(t)]\rVert_\infty \le C < \infty.  \label{eq:Clessinf}
\end{equation}

The reason Eq.~(\ref{eq:Clessinf}) is desirable is the following.  Without loss of generality, we normalize $(\mathcal{O}|\mathcal{O}) = 1$.  We then define a time-dependent probability distribution $\mathbb{P}_t$ on $\Lambda$ as \begin{equation}
\mathbb{P}_t(i \in \Lambda) := (\mathcal{O}(t)|\mathbb{Q}_i  |\mathcal{O}(t)),
\end{equation}
since by Eq.~(\ref{eq:sumQi1}) the probability distribution is properly normalized: $\mathbb{P}_t(\Lambda)=1$.   We may then reinterpret $t_2^\delta(x)$ as the first time where the probability that $i\ge x$ on the measure $\mathbb{P}_t$ is sufficiently large:  \begin{equation}
t_2^\delta(x) = \inf \left\lbrace t>0 :  \delta < \mathbb{P}_t( i \ge x)   \right\rbrace.
\end{equation}
We may then interpret $\mathcal{F}$ for $\alpha>\frac{5}{2}$ as a classical random variable that gives $i$ with probability $\mathbb{P}_t(i)$.
By Markov's inequality, \begin{equation}
\mathbb{P}_t( i \ge x) \le \frac{\mathbb{E}_t[\mathcal{F}]}{x}, \label{eq:markov}
\end{equation}
where $\mathbb{E}_t[\cdot]$ denotes expectation value on the measure $\mathbb{P}_t$.   If Eq.~(\ref{eq:Clessinf}) holds, then for any operator $\mathcal{O}_0$ in the domain of $\mathbb{Q}_0$, \begin{align}
\mathbb{E}_t[\mathcal{F}] &= \int\limits_0^t \mathrm{d}s \frac{\mathrm{d}}{\mathrm{d}s} (\mathcal{O}_0(s)| \mathcal{F}|\mathcal{O}_0(s)) = \int\limits_0^t \mathrm{d}s (\mathcal{O}_0(s)| [\mathcal{F}, \mathcal{L}(s)] |\mathcal{O}_0(s)) \notag \\
&  \le \int\limits_0^t \mathrm{d}s  \left| (\mathcal{O}_0(s)| [\mathcal{F}, \mathcal{L}(s)] |\mathcal{O}_0(s))  \right| \le Ct.  \label{eq:Ctmarkov}
\end{align}
Combining Eq.~(\ref{eq:markov}) and Eq.~(\ref{eq:Ctmarkov}), we see that Eq.~(\ref{eq:sec4thm}) holds with \begin{equation}
K = \frac{\delta}{C}.
\end{equation}

Hence, it remains to prove Eq.~(\ref{eq:Clessinf}).   To do so, it will be useful to define \begin{equation}
    \widetilde\Lambda := \Lambda-\lbrace 0\rbrace,
\end{equation}and a more refined set of complete, orthogonal projectors: for $S\subseteq\widetilde\Lambda$, \begin{equation}
\mathbb{R}_S |a_{0}\cdots a_{L}) := \left\lbrace \begin{array}{ll} |a_{0}\cdots a_{L})  &\  i>0 \text{ and } a_i \ne 0 \text{ if and only if } i\in S \\ 0 &\ \text{otherwise} \end{array}\right.,
\end{equation}
which projects onto the operators whose support is exactly the subset $S$.  We also define \begin{equation}
    \mathcal{F}_S := \max_{i\in S} i \label{eq:FSalphabig}
\end{equation}
to be the right-most occupied site.  Observe that $\mathcal{F}_S \mathbb{R}_S =\mathbb{R}_S\mathcal{F}\mathbb{R}_S $.   Since \begin{equation}
    \sum_{S\in \mathbb{Z}_2^{\widetilde\Lambda}} \mathbb{R}_S = \mathds{1},
\end{equation}
we may write, for any $\mathcal{O}\in\mathcal{B}$, \begin{equation}
    (\mathcal{O}|[\mathcal{F},\mathcal{L}]|\mathcal{O}) = \sum_{S,Q\in \mathbb{Z}_2^{\widetilde\Lambda}} (\mathcal{O}| \mathbb{R}_S [\mathcal{F},\mathcal{L}] \mathbb{R}_Q |\mathcal{O}) \le \sum_{S,Q\in \mathbb{Z}_2^{\widetilde\Lambda}} \sqrt{(\mathcal{O}|\mathbb{R}_S|\mathcal{O})(\mathcal{O}|\mathbb{R}_Q|\mathcal{O}) } \sup_{\mathcal{O},\mathcal{O}^\prime \in \mathcal{B}} \left|\frac{(\mathcal{F}_S-\mathcal{F}_Q)(\mathcal{O}|\mathbb{R}_S\mathcal{L}\mathbb{R}_Q|\mathcal{O}^\prime)}{\sqrt{(\mathcal{O}|\mathcal{O})(\mathcal{O}^\prime|\mathcal{O}^\prime)}}\right|.  \label{eq:OFLO}
\end{equation}
Next, we observe that the 2-locality of the Hamiltonian implies that $\mathbb{R}_S\mathcal{L}\mathbb{R}_Q \ne 0$ if and only if there exists a site $i\in\Lambda$ such that $S\cup \lbrace i \rbrace = Q$ or $Q\cup \lbrace i \rbrace = S$.

Suppose that $Q\cup \lbrace i \rbrace = S$, that $\mathcal{F}_Q = j$ and that $i>0$.   Then if $i < j$, $\mathcal{F}_S = \mathcal{F}_Q=j$; the right-most occupied site in $S$ and $Q$ has not changed, and hence the supremum in Eq.~(\ref{eq:OFLO}) vanishes.   Therefore, the supremum is only non-trivial when $i>j$.   By submultiplicativity of the operator norm, there exists $0<A<\infty$ such that
\begin{equation}
    \sup_{\mathcal{O},\mathcal{O}^\prime \in \mathcal{B}} \left|\frac{(\mathcal{F}_S-\mathcal{F}_Q)(\mathcal{O}|\mathbb{R}_S\mathcal{L}\mathbb{R}_Q|\mathcal{O}^\prime)}{\sqrt{(\mathcal{O}|\mathcal{O})(\mathcal{O}^\prime|\mathcal{O}^\prime)}}\right|
    \le 2 |i -j| \bigg\lVert \sum_{k\in Q} H_{ki} \bigg\rVert_\infty \le 2 |i - j| \sum_{k \in Q} \frac{h}{|i-k|^\alpha} \le \frac{A}{|i-\mathcal{F}_Q|^{\alpha-2}},
\end{equation}
where $A$ is a constant and, in the last step, we overestimated the sum by assuming all sites $\le j$ are included in the set $Q$.
A similar argument holds when $S\cup \lbrace i \rbrace = Q$.

It is now useful to interpret Eq.~(\ref{eq:OFLO}) as an auxiliary linear algebra problem.  Let us define $\varphi \in \mathbb{R}^{\mathbb{Z}_2^{\widetilde\Lambda}}$ as
\begin{equation}
    \varphi_S := \sqrt{(\mathcal{O}|\mathbb{R}_S|\mathcal{O})},
\end{equation}
and $M \in \mathbb{R}^{\mathbb{Z}_2^{\widetilde\Lambda}}\times\mathbb{R}^{\mathbb{Z}_2^{\widetilde\Lambda}}$ as \begin{equation}
    M_{SQ} = M_{QS} := \left\lbrace\begin{array}{ll} A |\mathcal{F}_S-\mathcal{F}_Q|^{2-\alpha} &\ \mathcal{F}_S \ne \mathcal{F}_Q \text{ and } S = Q \cup \lbrace m\rbrace \text{ or } Q = S \cup \lbrace m\rbrace \\  0 &\ \text{otherwise} \end{array}\right.. \label{eq:MSQQS}
\end{equation}
Since \begin{equation}
    (\mathcal{O}|[\mathcal{F},\mathcal{L}]|\mathcal{O}) \le \sup_{\varphi : \lVert \varphi \rVert_{\mathrm{F}} = 1} \sum_{S,Q}\varphi_SM_{SQ}\varphi_Q = \lVert M\rVert_\infty,
\end{equation}
it suffices to show that $\lVert M\rVert_\infty < \infty$.

To bound the maximal eigenvalue of $M$, we use the min-max Collatz-Weiland Theorem \cite{Meyer2000}.  To do that, we must first establish that $M$ is an irreducible matrix (non-negativity of the entries is guaranteed by Eq.~(\ref{eq:MSQQS}).  To show irreducibility, we observe that
\begin{equation}
    \left(M^{|S|}\right)_{\emptyset S} \ne 0;
\end{equation}
the sequence of subsets which satisfies this identity corresponds to sequentially adding the elements of $S$ from smallest to largest.  We conclude that (by non-negativity of all $M^n$) there exists an $n\in\mathbb{Z}^+$ such that $\left(M^n\right)_{SQ}> 0$ for all sets $S$ and $Q$.

We are now ready to apply the min-max Collatz-Weiland Theorem: \begin{equation}
    \lVert M\rVert_\infty = \inf_{\varphi \in \mathbb{R}^{\mathbb{Z}_2^{\widetilde\Lambda}} : \varphi_S > 0 } \max_S \frac{1}{\varphi_S} \sum_{Q\in\mathbb{Z}_2^{\widetilde\Lambda}} M_{SQ}\varphi_Q. \label{eq:CWbound}
\end{equation}
Clearly an upper bound to the maximal eigenvalue comes from choosing any trial vector $\varphi$ that we desire.  We make the following choice: writing \begin{equation}
    S = \lbrace n_1,\ldots, n_\ell \rbrace, \text{ with } n_i < n_{i+1}, \label{eq:Sordered}
\end{equation}we take $\varphi_\emptyset = 1$, and then define $n_0=0$ and
\begin{equation}
    \varphi_S := \prod_{i=1}^{|S|} \left(n_i - n_{i-1}\right)^{-\beta}, \label{eq:varphiansatz}
\end{equation}
where $\beta$ is a tunable parameter we will shortly fix.  Now we evaluate the right hand side of Eq.~(\ref{eq:CWbound}), defining $j=\mathcal{F}_S$: \begin{equation}
    \frac{1}{\varphi_S} \sum_{Q\in\mathbb{Z}_2^{\widetilde\Lambda}} M_{SQ}\varphi_Q = M_{S,S-\lbrace j\rbrace} \frac{\varphi_{S-\lbrace j\rbrace}}{\varphi_S} + \sum_{k \in \Lambda : k>j} M_{S,S\cup\lbrace k\rbrace} \frac{\varphi_{S\cup\lbrace k \rbrace}}{\varphi_S} \label{eq:CWboundeval}
\end{equation}
Using Eq.~(\ref{eq:MSQQS}), and assuming that $j_* = \mathcal{F}_{S-\lbrace j\rbrace}$, \begin{equation}
    M_{S,S-\lbrace j\rbrace} \frac{\varphi_{S-\lbrace j\rbrace}}{\varphi_S} \le A (j-j_*)^{\beta+2-\alpha}.
\end{equation}
We hence take \begin{equation}
    \beta = \alpha-2
\end{equation}
to ensure that this first term is finite.   Evaluating the second term of Eq.~(\ref{eq:CWboundeval}), \begin{equation}
    \sum_{k \in \Lambda : k>j} M_{S,S\cup\lbrace k\rbrace} \frac{\varphi_{S\cup\lbrace k \rbrace}}{\varphi_S} \le A \sum_{k=j+1}^\infty (k-j)^{2-\alpha-\beta} \le A_*,
\end{equation}
where \begin{equation}
    A_* := A \frac{2^{\alpha+\beta-3}}{\alpha+\beta-3} < \infty,
\end{equation}
so long as $\alpha>\frac{5}{2}$.  We conclude that $C\le A+A_* < \infty$, proving the theorem when $\alpha>\frac{5}{2}$.

We now return to the case $\frac{3}{2}<\alpha \le \frac{5}{2}$.   The proof is essentially identical with a few minor changes.  Firstly, we set $\mathcal{F}_{\lbrace 0 \rbrace} = 0$, and for non-empty sets we define \begin{equation}
    \mathcal{F}_S := \max_{j\in S} \frac{j^{\gamma}}{1+K^\prime \log j}, \label{eq:FSalphasmall}
\end{equation}
for a parameter $\gamma \in (0,1)$ that we will fix shortly.  We choose the parameter $K^\prime$ such that $\mathcal{F}_i$ is a convex function on $\mathbb{Z}^+$:  $|\mathcal{F}_i - \mathcal{F}_j| \le \mathcal{F}_{|i-j|}$.   Such a $K^\prime$ can be shown to exist by extending $\mathcal{F}$ to act on $[1,\infty)$, after which we use elementary calculus to demand that \begin{equation}
   0< \frac{\mathrm{d}\mathcal{F}(x)}{\mathrm{d}x} = \frac{1}{x^{1-\gamma} (1+K^\prime \log x)} \left(\gamma - \frac{K^\prime}{1+K^\prime \log x}\right), \label{eq:Fxd1}
\end{equation}along with \begin{equation}
    0> \frac{\mathrm{d}^2\mathcal{F}(x)}{\mathrm{d}x^2} = -\frac{1}{x^{2-\gamma} (1+K^\prime \log x)}\left(\left(1-\gamma + \frac{K^\prime}{1+K^\prime \log x} \right)\left(\gamma - \frac{K^\prime}{1+K^\prime \log x}\right)- \left(\frac{K^\prime}{1+K^\prime \log x}\right)^2\right). \label{eq:Fxd2}
\end{equation}
\cref{eq:Fxd1} and Eq.~(\ref{eq:Fxd2}) are both satisfied by the choice  \begin{equation}
    K^\prime = \frac{\gamma }{4}.
\end{equation}
We then find that convexity of $\mathcal{F}_i$ leads to the replacement of Eq.~(\ref{eq:MSQQS}) with \begin{equation}
    M_{SQ} = M_{QS} := \left\lbrace\begin{array}{ll} A |\mathcal{F}_S-\mathcal{F}_Q|^{\gamma+1-\alpha} (1+K^\prime \log |\mathcal{F}_S-\mathcal{F}_Q|)^{-1} &\ \mathcal{F}_S \ne \mathcal{F}_Q \\  0 &\ \text{otherwise} \end{array}\right..
\end{equation}
Lastly, we replace Eq.~(\ref{eq:varphiansatz}) with \begin{equation}
     \varphi_S := \prod_{i=1}^{|S|} \frac{\left(n_i - n_{i-1}\right)^{\gamma+1-\alpha} }{1+K^\prime \log (n_i - n_{i-1})}.
\end{equation}
These choices guarantee that \begin{equation}
    M_{S,S-\lbrace j\rbrace} \frac{\varphi_{S-\lbrace j\rbrace}}{\varphi_S} = A,
\end{equation}
as in the prior setting.  Then we find that \begin{equation}
    \sum_{k \in \Lambda : k>j} M_{S,S\cup\lbrace k\rbrace} \frac{\varphi_{S\cup\lbrace k \rbrace}}{\varphi_S} \le A \sum_{k=j+1}^\infty \frac{1}{(k-j)^{2(\alpha-1-\gamma)}(1+K^\prime \log (k-j))^2}. \label{eq:lastsumQWthm}
\end{equation}
Upon choosing $\gamma = \alpha-\frac{3}{2}$, we obtain that the sum above is finite.  Note that the logarithmic factors were required to obtain finiteness of Eq.~(\ref{eq:lastsumQWthm}). Hence we obtain $\lVert M\rVert_\infty <\infty$.  Lastly, we mimic the arguments of Eq.~(\ref{eq:Ctmarkov}) to complete the proof.
\end{proof}

We conjecture that in $d>1$, the Frobenius light cone is always linear if and only if \begin{equation}
    \alpha > \frac{3d}{2}+1. \label{eq:frobeniusconjecture}
\end{equation}
We expect that for $q$-local Hamiltonians with $q>2$, \cref{eq:frobeniusconjecture} holds only when a slightly stricter requirement than \cref{eq:sec5alphadef} is obeyed: for example, if $\norm{H_{\lbrace n_1,\ldots, n_q\rbrace}} \lesssim\prod_i |n_i-n_{i+1}|^{-\alpha}$ in one dimension.

The Frobenius light cone of Theorem \ref{thm:mbqw} is tight up to subalgebraic corrections, when applied to arbitrarily large operators.  This can be seen by considering a large operator of the form \begin{equation}
    \mathcal{O}_0 = \prod_{i=0}^{L/3}X^+_i  + \mathrm{h.c.}
\end{equation}
supported on the left-most $L/3$ sites of the lattice,
where $X^+_i = X_i + \mathrm{i}Y_i$.
We would like to spread this operator to the right-most $L/3$ sites of the lattice, which are at least a distance $L/3$ away from the initial support.
If the Hamiltonian is \begin{equation}
    H = \sum_{0\leq j \le \frac{L}{3}}\sum_{\frac{2L}{3}<k\leq L} \frac{Z_j Z_k}{L^\alpha},
\end{equation}
it is straightforward to show that the fraction of $\mathcal{O}_0(t)$ supported beyond $2L/3$ is (up to the first order in $t$)
\begin{align}
	\sum_{k>L/2}\mathbb Q_k\mathcal O_0(t) = \frac{t}{3L^{\alpha-1}}\mathcal O_0 \sum_{\frac{2L}{3}<k\leq L} Z_k.
\end{align}
The Frobenius norm of this fraction is
\begin{align}
	\sqrt{\frac{L}{3}\left(\frac{t}{3L^{\alpha-1}}\right)^2}\propto \frac{t}{L^{\alpha-3/2}}.
\end{align}
Therefore, our bound in \cref{thm:mbqw} is tight up to $O(1)$ factors.

\subsection{Quantum state transfer} \label{sec:qst}
An immediate consequence of this theorem is that the Lieb-Robinson light cone \emph{is not relevant} for infinite temperature many-body quantum chaos and the growth of operators.   A more practical application of the Frobenius light cone are tighter constraints on at least two different kinds of quantum state transfer.   For simplicity, we assume that $\dim(\mathcal{H}_i)=2$, and denote $|0_i\rangle$ and $|1_i\rangle$ as the eigenstates of the Pauli matrix $Z_i$ on $\mathcal{H}_i$.

A \emph{universal} notion of quantum state transfer from $i\in \Lambda$ to $j\in\Lambda$, which is independent of the background state, is to demand that there exists a Hamiltonian protocol $H(t)$ and a time $\tau\in\mathbb{R}$ such that \begin{equation}
    X^\alpha_i(\tau) = X^\alpha_j. \label{eq:strongstatetransfer}
\end{equation}
It is obvious that Theorem \ref{thm:mbqw} constrains the time at which Eq.~(\ref{eq:strongstatetransfer}) may hold; hence Eq.~(\ref{eq:strongstatetransfer}) cannot be performed at a time $\tau$ which scales slower than linearly in the distance $\mathcal{D}(i,j)$ when $\alpha>\frac{5}{2}, d=1$.

Another interesting scenario arises when we restrict to a time evolution operator $U(t)$ that obeys \begin{equation}
    U(t) \left( \bigotimes_{k\in\Lambda} |0_k\rangle \right) = \bigotimes_{k\in\Lambda} |0_k\rangle.  \label{eq:Utweakconstraint}
\end{equation}
Many protocols, including our own (Theorem \ref{thm:statetransfer}) and that of Ref. \cite{Eldredge2017}, are compatible with (\ref{eq:Utweakconstraint}).   With Eq. (\ref{eq:Utweakconstraint}), we now consider a quantum system whose initial condition is \begin{equation}
    |\psi(0)\rangle := |\phi_i\rangle \otimes \bigotimes_{k\in\Lambda-\lbrace i\rbrace } |0_k\rangle.
\end{equation}
for arbitrary $|\phi_i\rangle\in\mathcal{H}_i$.  Our goal is to find a time evolution operator $U(t)$ and a time $\tau$, such that $|\psi(t)\rangle = U(t)|\psi(0)\rangle$ and \begin{equation}
    \langle \psi(\tau)| Z_j |\psi(\tau)\rangle = \langle \phi_i| Z_i|\phi_i\rangle.  \label{eq:weakstatetransferdef}
\end{equation}
In particular, the probability of measuring a 0 or 1 on site $j$ at time $\tau$ is given by the probability of measuring it at time $t=0$ on site $i$.   This property must hold for \emph{all} $|\phi_i\rangle$ for a fixed $U(t)$, since the protocol must be able to transfer arbitrary states.
\begin{cor}
Let $\frac{3}{2}<\alpha\in\mathbb{R}$ and $x=\mathcal{D}(i,j)$.  Assuming (\ref{eq:Utweakconstraint}), there exist $0<K,K^\prime<\infty$ such that any state transfer algorithm runs in a time $\tau$ obeying \begin{equation}
    \tau > K \times  \left\lbrace \begin{array}{ll} x &\ \alpha > \frac{5}{2} \\ x^{\alpha - 3/2} (1+K^\prime\log x)^{-1} &\  \frac{3}{2} < \alpha \le \frac{5}{2} \end{array}\right..  \label{eq:sec4corollary}
\end{equation}\label{cor:statetransfer}
\end{cor}
\begin{proof}
We begin by observing that we may assume $|\phi_i\rangle = |1_i\rangle$ without loss of generality, since Eq.~(\ref{eq:weakstatetransferdef}) is trivially obeyed by Eq.~(\ref{eq:Utweakconstraint}).  Now the proof largely mirrors that of Theorem \ref{thm:mbqw}.  Without loss of generality, we may define lattice sites such that $i=0$ and $j>0$, as explained above. Define \begin{equation}
    |S\rangle := \left(\bigotimes_{k\in S}|1_k\rangle \otimes \bigotimes_{k\in S^{\mathrm{c}}}|0_k\rangle \right),
\end{equation}
and the observable $\mathcal{F}$ which acts on the mutual eigenbasis of $Z_i$ as \begin{equation}
    \mathcal{F} |S\rangle := \mathcal{F}_S |S\rangle,
\end{equation}
for any $S\subseteq\Lambda$; here $\mathcal{F}_S$ is given by Eq.~(\ref{eq:FSalphabig}) when $\alpha>\frac{5}{2}$ and Eq.~(\ref{eq:FSalphasmall}) when $\frac{3}{2}<\alpha\le\frac{5}{2}$.   For simplicity we only describe explicitly the case $\alpha > \frac{5}{2}$, as the other case follows from identical considerations.      We evaluate \begin{equation}
   \left| \frac{\mathrm{d}}{\mathrm{d}t} \langle \psi(t)|\mathcal{F}|\psi(t)\rangle \right| \le \left| -\mathrm{i} \langle \psi(t)|[\mathcal{F},H(t)]|\psi(t)\rangle \right| \le \lVert [\mathcal{F},H(t)]\rVert_\infty.
\end{equation}
As before, our goal is to show that $\lVert [\mathcal{F},H(t)]\rVert_\infty<\infty$.  Since $H$ is 2-local, we know that $H_{ij}(t)|0_i\rangle |0_j\rangle \propto |0_i\rangle|0_j\rangle$ by Eq.~(\ref{eq:Utweakconstraint}).  This implies that, as before $[\mathcal{F},H]$ can only be non-vanishing when $H$ serves to either add a new $|1\rangle$ to the right end of the state, or delete the right-most $|1\rangle$.  Hence $\langle S|[\mathcal{F},H(t)]|Q\rangle \ne 0$ only if $|S-S\cap Q| \le 1$ and $|Q-S\cap Q| \le 1$.   We define the matrix $M_{SQ} := \sup \langle S| [\mathcal{F},H(t)]|Q\rangle $, which equals \begin{equation}
    M_{SQ} = M_{QS} := \left\lbrace \begin{array}{ll} A|\mathcal{F}_S-\mathcal{F}_Q|^{2-\alpha} &\ S = Q\cup \lbrace m\rbrace \text{ or } Q = S\cup \lbrace m\rbrace \\  A|\mathcal{F}_S-\mathcal{F}_Q|^{1-\alpha} &\ \text{there exists } R\text{ with } S = R\cup \lbrace m\rbrace \text{ and } Q = R\cup \lbrace n\rbrace, \\ &\ \text{ and } Q\ne S \text{ and } \mathcal{F}_R < \min(\mathcal{F}_S,\mathcal{F}_Q),  \\ 0 &\ \text{otherwise} \end{array}\right.. \label{eq:statetransferM}
\end{equation}
We bound the maximal eigenvalue of $M$  using the Collatz-Wieland inequality Eq.~(\ref{eq:CWbound}), using the trial vector $\varphi_S$ given Eq.~(\ref{eq:varphiansatz}).  Observe that the first line of Eq.~(\ref{eq:statetransferM})  is identical to Eq.~(\ref{eq:MSQQS}); as such these terms in $M_{SQ}\varphi_Q$ are bounded by a constant as before.  The new terms we must deal with arise from the second line of Eq.~(\ref{eq:statetransferM}).    If  $S$ is given by Eq.~(\ref{eq:Sordered}), we find that \begin{equation}
\frac{1}{\varphi_S} \sum_{Q:|Q|=|S|} M_{SQ}\varphi_Q < A \sum_{m=n_{\ell-1}+1}^{\infty} \left(\frac{n_\ell - n_{\ell-1}}{m-n_{\ell-1}}\right)^{\alpha-2} \frac{1-\delta_{m,n_\ell}}{(m-n_{\ell-1})^{\alpha-1}} < A_{\mathrm{st}},
\end{equation}
for some constant $0<A_{\mathrm{st}}<\infty$, so long as $\alpha>\frac{5}{2}$.  We conclude that $M$ has a bounded maximal eigenvalue, independently of the lattice size.  We conclude there exists $0<K<\infty$ such that $\langle \psi(t)|\mathcal{F}|\psi(t)\rangle \le Kt$.

At time $\tau$, we must have \begin{equation}
    |\psi(\tau)\rangle = |1_j\rangle \otimes |\psi^\prime_{\Lambda-\lbrace j\rbrace}\rangle,
\end{equation}
for arbitrary state $|\psi^\prime\rangle$ acting on sites other than $j$.  Therefore, \begin{equation}
    \langle \psi(\tau)|\mathcal{F}|\psi(\tau)\rangle \ge j.
\end{equation}
Using Markov's inequality as in the proof of Theorem \ref{thm:mbqw}, we obtain Eq.~(\ref{eq:sec4corollary}).  The case $\alpha<\frac{5}{2}$ is proved analogously.
\end{proof}

\section{Free light cone}\label{sec:free}
In this section, we discuss bounds on the quantum dynamics of non-interacting many-body systems.

\subsection{Non-interacting Hamiltonians}
Consider a many-body quantum system defined on a $d$-dimensional lattice graph $\Lambda$; we assume the same properties of $\Lambda$ as in Section \ref{sec:preliminaries}.   Suppose that the many-body Hamiltonian takes the form \begin{equation}
    H(t) = \sum_{i,j\in\Lambda} h_{ij}(t)c^\dagger_i c_j, \label{eq:freehamiltonian}
\end{equation}
where $h(t):\mathbb{R}\rightarrow \mathbb{C}^{\Lambda\times\Lambda}$ is a Hermitian matrix, and $c^\dagger_i$ and $c_i$ represent either fermionic creation and annihilation operators:  \begin{equation}
    \lbrace c_j , c_i^\dagger\rbrace := \delta_{ij},
\end{equation}
or bosonic creation and annihilation operators:  \begin{equation}
    [c_j, c_i^\dagger  ] := \delta_{ij}. \label{eq:bosoncreateannihilate}
\end{equation}
The on site Hilbert space $\mathcal{H}_i$ obeys $\dim(\mathcal{H}_i) = 2$ in the fermionic case, and $\dim(\mathcal{H}_i)=\infty$ in the bosonic case.
We note, however, that in isolated bosonic systems, $\mathcal{H}_i$ can often be truncated so that $\dim(\mathcal{H}_i)$ is at most the number of excitations on the lattice and is therefore finite.

As is well known, the evolution of all operators in such a non-interacting theory is controlled by the Green's function of the single particle problem on the Hilbert space $\mathbb{C}^\Lambda$.    Time evolution on this space is generated by the Hamiltonian \begin{equation}
    H_{\mathrm{sp}}(t) := \sum_{i,j \in \Lambda} h_{ij}(t)|i\rangle\langle j|,
\end{equation}
where $\ket{i}$ denotes the state that has exactly one excitation at site $i\in \Lambda$.
The single particle time evolution matrix obeys the differential equation \begin{equation}
    \frac{\mathrm{d}}{\mathrm{d}t} U_{\mathrm{sp}}(t) := -\mathrm{i}H_{\mathrm{sp}}(t)U_{\mathrm{sp}}(t),
\end{equation} together with the initial condition $U_{\mathrm{sp}}(0)=1$.  For example, in the fermionic model, \begin{equation}
    c_i(t) = \sum_{j\in\Lambda}U_{\mathrm{sp},ij}(t)c_j,
\end{equation}which follows from observing that \begin{equation}
    \frac{\mathrm{d}}{\mathrm{d}t} c_i :=   \mathrm{i}[H(t),c_i] = \sum_{j\in\Lambda} \mathrm{i}h_{ij}(t)[c^\dagger_i c_j, c_i] = -\mathrm{i}\sum_{j\in\Lambda}h_{ij}(t)c_j.
\end{equation}
For simplicity in the discussion that follows, we drop the ``sp" subscript on $H$ and $U$.

\subsection{Quantum walks of a single particle}
Consider a normalized wave function $|\psi(t)\rangle := U(t)|\psi\rangle\in\mathbb{C}^\Lambda$, along with its canonical probability distribution $\mathbb{P}_t$ on $\Lambda$: \begin{equation}
    \mathbb{P}_t(A) := \sum_{i\in A} \frac{|\langle i|\psi(t)\rangle|^2}{\langle \psi|\psi\rangle}.
\end{equation}
Let us label an origin $0\in\Lambda$, and assume that $|\psi(0)\rangle  = |0\rangle$.
We now use the quantum walk framework to prove our third main result, on the concentration of $\mathbb{P}_t$ on lattice sites close to the origin.

\begin{thm}\label{thm:free}
If $\alpha>d+1$, $\epsilon>0$, and $r\in \mathbb{Z}^+$, there exist constants $0<K,u<\infty$ such that \begin{equation}
    \sum_{y\in\Lambda: \mathcal{D}(y,x)\ge r} \mathbb{P}_t(y) \le \frac{Kt}{(r-ut)^{\alpha-d-\epsilon}}.  \label{eq:QW5alphabig}
\end{equation}
When $d<\alpha\le d+1$, Eq.~(\ref{eq:QW5alphabig}) holds with $u=0$.
\end{thm}
\begin{proof}
We first prove Eq.~(\ref{eq:QW5alphabig}) when $\alpha>d+1$.  Define the Hermitian operator
\begin{subequations}
\begin{align}
 \langle x|\mathcal{F}(t)|y\rangle &:= \delta_{xy}\mathcal{F}(x,t),\\
\mathcal{F}(x,t)&:= \min\left(0,\mathcal{D}(x,0)-ut\right) .
\end{align}
   \end{subequations}

Our goal is to follow the proof of Theorem \ref{thm:mbqw}; first bounding the rate of change of an expectation value, and then employing Markov's inequality.  The operator whose expectation value we will bound in the time evolved wave function is $\mathcal{F}^{\beta}$; ultimately we will see that $\beta=\alpha-d-\epsilon$.

First, let us bound \begin{align}
    \left|\mathcal{F}(x)^{\beta} - \mathcal{F}(y)^{\beta}\right| &\le \beta\max(\mathcal{F}(x),\mathcal{F}(y))^{\beta-1} |\mathcal{F}(y)-\mathcal{F}(x)| \notag \\
    &\le \beta \mathcal{D}(x,y) \left(\mathcal{F}(x)^{\beta-1} + \mathcal{F}(y)^{\beta-1}\right). \label{eq:Falphad}
\end{align}
Then,
\begin{align}
   \frac{\mathrm{d}}{\mathrm{d}t} \langle \psi(t)|\mathcal{F}^{\beta}|\psi(t)\rangle &=- \mathrm{i}\langle \psi(t)|[\mathcal{F}^{\beta},H(t)]|\psi(t)\rangle  -u\beta\langle \psi(t)|\mathcal{F}^{\beta-1}|\psi(t)\rangle .\label{eq:Falphadbound}
\end{align}
Let us first bound the first term, using Eq.~(\ref{eq:Falphad}) and Eq.~(\ref{eq:Falphadbound}): \begin{align}
    |\psi(t)|[\mathcal{F}^{\beta},H(t)]|\psi(t)\rangle| &\le  2\sum_{\lbrace x,y\rbrace\subset \Lambda} |\langle x|[\mathcal{F}^{\beta},H(t)]|y\rangle|~ |\langle x|\psi\rangle\langle y|\psi\rangle| \notag \\
    &\le \sum_{x\in\Lambda}\sum_{y\in\Lambda-\lbrace x\rbrace} (\mathbb{P}_t(x)+ \mathbb{P}_t(y))  |\langle x|[\mathcal{F}^{\beta},H(t)]|y\rangle| \notag \\
    &\le \beta \sum_{x\in\Lambda}\sum_{y\in\Lambda-\lbrace x\rbrace} \mathbb{P}_t(x) \frac{2h}{\mathcal{D}(x,y)^{\alpha-1}} \left(\mathcal{F}(x)^{\beta-1} + \mathcal{F}(y)^{\beta-1}\right).
\end{align}
In the last line, we have used the symmetry of the sum under exchanging $x$ and $y$ to remove $\mathbb{P}_t(y)$.  Then we observe that\begin{equation}
    \mathcal{F}(y)^{\beta-1} \le (\mathcal{F}(x) + \mathcal{D}(x,y))^{\beta-1} \le 2^{\beta-1} \left(\mathcal{F}(x)^{\beta-1} + \mathcal{D}(x,y)^{\beta-1}\right).
\end{equation}
Hence, so long as we choose\begin{equation}
    \beta = \alpha-d-\epsilon, \label{eq:QW5beta}
\end{equation}
we conclude that there exist constants $0<K,A<\infty$ such that
\begin{align}
    |\psi(t)|[\mathcal{F}^{\beta},H(t)]|\psi(t)\rangle| &\le  \left(2+2^\beta\right) \sum_{x\in\Lambda} \mathbb{P}_t(x)\sum_{y\in\Lambda-\lbrace x\rbrace} \left(\frac{h}{\mathcal{D}(x,y)^{\alpha-\beta}} + \frac{h}{\mathcal{D}(x,y)^{\alpha-1}} \mathcal{F}(x)^{\beta-1}\right) \notag \\
    &\le \sum_{x\in\Lambda} \mathbb{P}_t(x) \left(K + A\mathcal{F}(x)^{\beta-1}\right) = K + A \langle \psi(t)|\mathcal{F}^{\beta-1}|\psi(t)\rangle,
\end{align}
where $K,A$ are constants.
Upon choosing \begin{equation}
    u = \frac{A}{\beta},
\end{equation}
Eq.~(\ref{eq:Falphadbound}) implies that \begin{equation}
    \langle \psi(t)|\mathcal{F}^\beta |\psi(t)\rangle \le Kt.
\end{equation}
Using Markov's inequality, and assuming $r>ut$, \begin{equation}
  \sum_{y\in\Lambda: \mathcal{D}(y,x)\ge r} \mathbb{P}_t(y)\le \frac{\mathbb{E}_t[\mathcal{F}^\beta]}{(r-ut)^\beta} \le \frac{Kt}{(r-ut)^\beta}. \label{eq:QW5markov}
\end{equation}
Combining Eq.~(\ref{eq:QW5beta}) and Eq.~(\ref{eq:QW5markov}) we obtain Eq.~(\ref{eq:QW5alphabig}).

Secondly, we study the case $\alpha \le d+1$.   Now we define \begin{equation}
    \langle x|\mathcal{F}|y\rangle := \delta_{xy} \mathcal{D}(x,0)^\beta,
\end{equation}
with $\beta$ given by Eq.~(\ref{eq:QW5beta}).  Observe that $\beta<1$.  In this limit, \begin{equation}
    \frac{\mathrm{d}}{\mathrm{d}t} \mathbb{E}_t[\mathcal{F}] \le \sum_{x\in\Lambda} \mathbb{P}_t(x) \sum_{y\in\Lambda-\lbrace x\rbrace} \frac{h}{\mathcal{D}(x,y)^\alpha} \left|\mathcal{D}(x,0)^\beta - \mathcal{D}(y,0)^\beta \right| \le \sum_{x\in\Lambda} \mathbb{P}_t(x) \sum_{y\in\Lambda-\lbrace x\rbrace} \frac{h}{\mathcal{D}(x,y)^{\alpha-\beta}},
\end{equation}
where in the last inequality we have used the convexity of $\mathcal{F}$ as a function of distance.  For any $\epsilon>0$, the sum over $y$ converges; hence there exists a $K<\infty$ such that \begin{equation}
    \frac{\mathrm{d}}{\mathrm{d}t} \mathbb{E}_t[\mathcal{F}] \le \sum_{x\in\Lambda} \mathbb{P}_t(x) \times K = K.
\end{equation}
Another application of Markov's inequality implies Eq.~(\ref{eq:QW5alphabig}).
\end{proof}

\subsection{Local simulation of a single particle}

An immediate application of Theorem \ref{thm:free} is to bound the error made by approximating time evolution via a truncated, local Hamiltonian, analogous to the discussion of Section \ref{sec:local-obs}.

\begin{cor}\label{cor:freesim}
For any $i\in\Lambda$, define $B_i^r := \lbrace j\in\Lambda : \mathcal{D}(j,i)\le r\rbrace $, and define $\tilde H(t)$ to be the restriction of a free bosonic Hamiltonian $H(t)$ [Eq.~(\ref{eq:freehamiltonian})] to $B_i^r \subset \Lambda$.  Then for any $\epsilon>0$, there exists $0<K,K^\prime<\infty$ such that for times\begin{equation}
    t<\frac{K^\prime}{n} r^{\min(1,(\alpha-d-\epsilon)/3)}, \label{eq:corollarytconstraint}
\end{equation}
we have
\begin{align}
	\norm{b_i^\dag(t) - \tilde b_i^\dag(t)} \le K n^{3/2}\left(\frac{t}{r^{\alpha-d}} + \frac{t^{3/2}}{r^{(\alpha-d-\epsilon)/2}}\right),   \label{eq:freebosonapprox}
\end{align}
where the norm is estimated in the subspace that has at most $n\geq 1$ excitations across the lattice and $\tilde b_i^\dagger(t)$ denotes time evolution with the restricted Hamiltonian $\tilde H(t)$.
\end{cor}

\begin{proof}
Without loss of generality, we assume $i=0$, the origin.
Observe that
\begin{align}
	\norm{b_{0}^\dag(t)-\tilde b_{0}^\dag(t)}
	&\leq\int_0^t \textrm{d}s\norm{ \comm{\tilde b_{0}^\dag(t),H(t)-\tilde H(t)}} \leq\int_0^t \textrm{d}s\norm{ \comm{\tilde b_{0}^\dag(t),\sum_{i:\dist(i,0)\leq r} \sum_{j:\dist(j,0)> r}h_{ij}{b_ib_j^\dag}}}.\label{eq:boson-diff}
\end{align}
Using Theorem \ref{thm:free},
\begin{align}
	\tilde b_{0}^\dag(t) = \sum_{i:\dist(i,0)\leq r} f_i(t) b_i^\dag,\label{eq:boson-lin-expand}
\end{align}
where the coefficients $f_i(t)$ satisfy, for some $0<C<\infty$ and arbitrary $\epsilon>0$,
\begin{align}
	\sum_{i:\dist(i,0)\geq x} \abs{f_i(t)}^2 \le \frac{Ct}{x^{\alpha-d-\epsilon}},
\end{align}
for all $x>0$ and all $t$ obeying Eq.~(\ref{eq:corollarytconstraint}).

We separate the sum over $i$ in \cref{eq:boson-diff} according to $\dist(i,0) \leq r/2$ and $r/2 < \dist(i,0) \leq r$.
In the former case, we have
\begin{align}
 	& \norm{ \comm{\tilde b_{0}^\dag(t),\sum_{i:\dist(i,0)\leq r/2} \sum_{j:\dist(j,0)> r}{h_{ij}b_ib_j^\dag}}}
 	\le 2\sqrt{n} \norm{\sum_{i:\dist(i,0)\leq r/2} \sum_{j:\dist(j,0)> r}{h_{ij}b_ib_j^\dag}}\nonumber\\
 	&\le 2n^{3/2} \max_{i:\dist(i,0)\leq r/2} \sum_{j:\dist(j,0)> r}\frac{1}{\dist(i,j)^{\alpha}}\le \frac{C_1 n^{3/2}}{r^{\alpha-d}},\label{eq:boson-inner}
\end{align}
where $0<C_1<\infty$ is a constant.
We have used the fact that $\abs{h_{ij}}\leq 1/\dist(i,j)^\alpha$ and that $\dist(j,i)\geq r/2$ for all $i$ such that $\dist(i,0)\leq r/2$.

On the other hand, for $r/2 < \dist(i,0) \leq r$,
\begin{align}
 	& \norm{ \comm{\tilde b_{0}^\dag(t),\sum_{i:r/2\leq \dist(i,0)\leq r} \sum_{j:\dist(j,0)> r}{h_{ij}b_ib_j^\dag}}} \le 2 \norm{\sum_{i:r/2< \dist(i,0)\leq r} f_i(t) b_i^\dag}
 	\norm{\sum_{i:r/2<\dist(i,0)\leq r} \sum_{j:\dist(j,0)> r}{h_{ij}b_ib_j^\dag}}\notag \\
 	&\lesssim \sqrt{\sum_{i:r/2<\dist(i,0)\leq r} \abs{f_i(t)}^2n}
 	\left(\max_{i:r/2<\dist(i,0)\leq r} \sum_{j:\dist(j,0)> r}\frac{n}{\dist(i,j)^{\alpha}}\right) \le C_2 \frac{n^{3/2}t^{1/2}}{r^{(\alpha-d-\epsilon)/2}},\label{eq:boson-outer}
\end{align}
for $0<C_2<\infty$.  Replacing \cref{eq:boson-inner,eq:boson-outer} into \cref{eq:boson-diff} and integrating over time, we arrive at Eq.~(\ref{eq:freebosonapprox}).
\end{proof}

\subsection{Single-particle state transfer}\label{sec:singlestatetransfer}

Our next goal is to prove the tightness of Theorem \ref{thm:free}, up to subalgebraic corrections.  This is achieved by the following theorem, which provides a rapid state-transfer protocol for a single particle.
\begin{thm} \label{thm:statetransfer}
For every $x\in\Lambda-\lbrace0\rbrace$ with $\mathcal{D}(x,0)>2$, there exists a constant $0<K<\infty$ and a Hermitian matrix $h(t):\mathbb{R}\rightarrow \mathbb{C}^{\Lambda\times\Lambda}$ obeying Eq.~(\ref{eq:sec5alphadef}), such that $\langle x|U(\tau)|0\rangle = 1$ at \begin{equation}
    \tau := K \times \left\lbrace \begin{array}{ll} \mathcal{D}(x,0) &\  \alpha \ge d+1 \\ \mathcal{D}(x,0)^{\alpha-d} &\ d<\alpha<d+1 \\ \log \mathcal{D}(x,0) &\ \alpha= d \\ 1 &\ \alpha < d  \end{array}\right..  \label{eq:perftransfer}
\end{equation}
\end{thm}

\begin{proof}
For $\alpha \ge d+1$, in order to transfer an excitation from $0$ to $x$, we simply use a sequence of nearest-neighbor hoppings, which would take time proportional to the distance $\dist(0,x)$.
Specifically, let $(y_0 := 0,y_1,\ldots, y_{\ell-1}, y_\ell := x)$ be a sequence of length $1+\mathcal{D}(x,0)$ such that the edge $(y_i,y_{i+1})$ is an edge of nearest neighbors in $\Lambda$; here $\ell := \mathcal{D}(x,0)$.
We then apply
\begin{equation}
    H(t) := \left\lbrace \begin{array}{ll} \mathrm{i}h|y_j\rangle\langle y_{j-1}| - \mathrm{i}h|y_{j-1}\rangle\langle y_j| &\ t \in [(j-1)\frac{\pi}{2h}, j\frac{\pi}{2h}) \\  0 &\ \text{elsewhere} \end{array} \right.,
\end{equation}
where $h$ is defined in \cref{eq:sec5alphadef}.
It is straightforward to verify that the Hamiltonian takes $\ket{y_{j-1}}$ to $\ket{y_{j}}$ at the end of the interval $[(j-1)\frac{\pi}{2h},j\frac{\pi}{2h}]$, for all $j = 1,\dots,\ell$.
As a result, we achieve perfect state transfer from site 0 to site $x$ at $t = \frac{\pi}{2h} \dist(0,x)$.
Therefore, \cref{thm:statetransfer} holds for $\alpha>d+1$ with
\begin{equation}
    K = \frac{\pi}{2h}.
\end{equation}

For $\alpha<d+1$, we use a state-transfer scheme depicted in \cref{fig:statetransfer}.  The scheme, as depicted, assumes that the lattice is a simple cubic lattice;  however, this protocol is easily applied to an arbitrary lattice graph, since we may always arrange the unit cells of the graph in the structure shown above.  As the generalization to other lattices is obvious, we describe only the case of a simple cubic lattice below.
We further assume the sites $0$ and $x$ are on the same axis of the lattice; more precisely, we assume that the path of shortest length connecting $0$ and $x$ is unique.  If $0$ and $x=(x_1,\ldots,x_d)$ do not satisfy this property, we use the protocol described below to to transfer the excitation from $0$ to $(x_1,0,\ldots,0)$, all the way to $(x_1,\ldots, x_d)$ in $d$ separate steps, increasing the total transfer time by at most a factor of $d$ compared to the protocol we describe below.

We define $q\in \mathbb{Z}^+$ as \begin{equation}
    q := \left\lfloor \log_2 \mathcal{D}(x,0) \right\rfloor + 1.  \label{eq:sec5qdef}
\end{equation}
Let $B_q\subset \mathbb R^d$ denote a cube of size $\mathcal{D}(x,0)$ such that the sites $0$ and $x$ are at two different corners of $B_q$.
We then recursively define a sequence of $q-1$ cubes, namely $B^{(0)}_{q-1},\dots,B^{(0)}_1$, satisfying
\begin{align}
     \{0\}\subset B^{(0)}_1 \subset B^{(0)}_2 \subset\cdots \subset B^{(0)}_{q-1} \subset{B^{(0)}_q} := B_q,
\end{align}
and the size of $B^{(0)}_s$ is $2^s\frac{\dist(x,0)}{2^q}$ for all $s = 1,\dots,q$.
Note that our definition ensures the size of $B_1$ is in $[1,2)$.
Similarly, we define the cubes $B_1^{(x)},\dots,B_{q}^{(x)}$ that contains $x$:
\begin{align}
     \{x\}\subset B^{(x)}_1 \subset B^{(x)}_2 \subset\cdots \subset B^{(x)}_{q-1} \subset{B^{(x)}_q} := B_q.
\end{align}

Our strategy is to first expand the state $|\psi(0)\rangle = \ket{0}$ to a coherent uniform superposition on $B^{(0)}_1$, which is subsequently expanded to coherent uniform superpositions on larger and larger cubes $B^{(0)}_2,\dots,B^{(0)}_q$.
After that, we reverse the process and contract the uniform superposition on $B^{(0)}_q = B^{(x)}_q$ onto the cubes $B^{(x)}_{q-1},\dots,B^{(x)}_{1}$, and finally onto site $\{x\}$.
We will argue that each expansion or contraction involving cubes of size $\ell$ takes time $\ell^{\alpha-d}$, where $\ell^{\alpha}$ is the penalty due to the power-law constraint and $\ell^{d}$ is the enhancement coming from the volumes of the cubes.
Summing over the values of $\ell$ results in a transfer time that scales as $\dist(0,x)^{\alpha-d}$ for $d<\alpha<d+1$.

To calculate the time it takes for each expansion or contraction, we invoke the following Lemma:

\begin{lem}\label{lem:superposition}
Let $A$ and $B$ be two disjoint subsets of $\Lambda$, and $0<C<\infty$.
Then if \begin{equation}
    |\psi(0)\rangle = \frac{1}{\sqrt{|A|}} \sum_{i\in A} |i\rangle,
\end{equation} there exists a free-particle Hamiltonian $H(t)$ defined in \cref{eq:freehamiltonian} with $\abs{h_{ij}}\leq C$ for all $i,j\in \Lambda$ such that for any $\theta \in \mathbb{R}$,  \begin{equation}
        |\psi(T)\rangle \propto \frac{\cos \theta }{\sqrt{|A|}} \sum_{i\in A} |i\rangle + \frac{\sin \theta }{\sqrt{|B|}} \sum_{i\in B} |i\rangle,  \label{eq:psiAB}
\end{equation}
at time \begin{equation}
    T \le \frac{\pi}{2C\sqrt{|B||A|}} .  \label{eq:lemmaT}
\end{equation}
\end{lem}
\begin{proof}
We prove the lemma by construction. Consider the Hamiltonian\begin{equation}
    H(t):= \mathrm{sgn}\left(\tan \theta\right) \sum_{k\in A}\sum_{j\in B} \mathrm{i}C \left(|j\rangle\langle k| - |k\rangle\langle j|\right). \label{eq:Hjk}
\end{equation}
Without loss of generality, we take $\theta \in [0,\frac{\pi}{2}]$; the generalization to other $\theta$ is straightforward.  By permutation symmetry, the wave function takes the form Eq.~(\ref{eq:psiAB}) with $\theta(t)$ a function of time.  Pick any $j\in B$.  We can explicitly evaluate \begin{equation}
   \frac{\mathrm{d}\theta}{\mathrm{d}t} = \frac{\sqrt{|B|} }{\cos\theta} \frac{\mathrm{d}\langle j|\psi(t)\rangle }{\mathrm{d}t} = -\mathrm{i}\frac{\sqrt{|B|} }{\cos\theta}\langle j|H|\psi(t)\rangle = C\sqrt{|B||A|}. \label{eq:dthetadt}
\end{equation}
Since the value of $\theta$ at which $|\psi(t)\rangle$ is given by Eq.~(\ref{eq:psiAB}) is in $[0,\frac{\pi}{2}]$, we conclude that Eq.~(\ref{eq:dthetadt}) implies Eq.~(\ref{eq:lemmaT}).
\end{proof}

By construction, the time $\tau$ of our perfect state transfer algorithm is given by \begin{equation}
    \tau = 2\sum_{s=1}^q T_{B,s},
\end{equation}
where $T_{B,s}$ is the time it takes to expand from $B^{(0)}_{s-1}$ to $B^{(0)}_{s}$, which is also the time it takes to contract $B^{(x)}_{s}$ into $B^{(x)}_{s-1}$.
To evaluate these times, we use Lemma~\ref{lem:superposition} with $C = 1/2^{s\alpha}$ to get
\begin{align}
    T_{B,s} \leq \pi\frac{2^{s\alpha}}{\sqrt{2^{(s-1)d}(2^{sd}-2^{({s-1})d})}}
    = \frac{2^d\pi}{\sqrt{2^d-1}} 2^{s(\alpha-d)},
\end{align}
for all $s = 1,\dots,q$.
Here we have lower bounded the number of sites in $B_s^{(0)}$ by $2^{sd}/2$.

For $\alpha\neq d$, summing over $s$ gives
\begin{align}
    \tau \leq \frac{2^{d+1}\pi}{\sqrt{2^d-1}}\frac{2^{(q+1)(\alpha-d)} -1}{2^{\alpha-d}-1}
    \leq
    \begin{cases}
    \frac{2^{d+3}\pi}{\sqrt{2^d-1}}\frac{1}{2^{\alpha-d}-1} \dist(x,0)^{\alpha-d} & (\alpha>d)\\
    \frac{2^{d+1}\pi}{\sqrt{2^d-1}}\frac{1}{1-2^{\alpha-d}}& (\alpha<d)
    \end{cases}.
\end{align}
On the other hand, at $\alpha = d$, we have
\begin{align}
    \tau \leq q\times \frac{2^d \pi}{\sqrt{2^d-1}} \leq \frac{2^d \pi}{\sqrt{2^d-1}} (1+\log\dist(x,0)).
\end{align}
Therefore, \cref{thm:statetransfer} follows.
\end{proof}

There are two important consequences of Theorem \ref{thm:statetransfer}.  Firstly, even a single quantum mechanical degree of freedom can perform state transfer as asymptotically well as the previously best known protocol in an interacting many-body system \cite{Eldredge2017} for $\alpha \geq d$.
Secondly, Theorem \ref{thm:statetransfer} proves that any possible improvement to Theorem \ref{thm:free} must be sub-algebraic.
Both the linear light cone and the superlinear polynomial light cones we have proved for free quantum systems with long-range interactions are now known to be optimal.
\Cref{thm:statetransfer} is also applicable to spin systems, since the spin degrees of freedom may be treated as hardcore bosons.
Similarly, the protocol applies to Hamiltonians with on-site and particle number conserving interactions such as the Bose-Hubbard model: the interactions have no effect since at all times during the protocol there is at most a single particle in the system.

As noted in the introduction, this state-transfer protocol is naturally realized in experiments whenever there is a conserved quantity.  For example, in a spin system with $z$-spin conservation, we can prepare the system in a highly polarized state with a single up-spin; the location of the up-spin represents the location of the single quantum degree of freedom, and our state-transfer protocol immediately applies.  In trapped ion crystals, it is natural to use a large transverse magnetic field to help restrict to this subspace \cite{Neyenhuise1700672}.  In addition, decoherence rates will be greatly reduced in the single-particle subspace, when compared to the GHZ states employed by Ref. \cite{Eldredge2017}.

A key feature of this state-transfer protocol is its remarkable robustness to error.  Here we give a heuristic argument that this will be the case; a complete analysis will be provided elsewhere \cite{2020arXiv200906587H}.  At step $n$ of the protocol above, there are $\mathcal{N}_n = 2^{dn}$ sites in each domain which are all mutually coupled; the coherent state transfer process leads to an enhancement in the transfer rate by a factor  of $\mathcal{N}_n$.  Now suppose that there is uncorrelated random error in the coefficients of $|j\rangle\langle k|$ in (\ref{eq:Hjk}).  Using random matrix theory \cite{mehta}, we estimate that these errors introduce lead to dephasing rates of order $\sqrt{\mathcal{N}_n}$.  If $|x\rangle$ is the target site for the state-transfer protocol, we estimate the loss in fidelity $\mathcal{F} = |\langle \psi(\tau)|x\rangle|^2$ by summing up the error after each step: \begin{equation}
    1-\mathcal{F} \sim \sum_{n=1}^m \frac{\epsilon}{\sqrt{\mathcal{N}_n}} \sim \epsilon \sum_{n=1}^m 2^{-dn/2} < \frac{ \epsilon}{2^{d/2}-1}.
\end{equation}
Here $\epsilon$ is related to the error in a single coupling in the state transfer process.  Therefore, the quantum coherent hopping of this state-transfer protocol renders it highly immune to imperfections in tunable coupling constants which are inevitable in any near term quantum simulator.    As $\epsilon\rightarrow 0$, the fidelity $\mathcal{F}\rightarrow 1$.

\subsection{Efficient early time classical boson sampling}\label{sec:boson-sampling}

The boson sampling problem was proposed by Aaronson and Arkhipov~\cite{Aaronson2011} as a potential candidate for the demonstration of quantum supremacy.
While simulating the dynamics of bosons hopping on a lattice is generally a difficult task for classical computers,
early-time evolutions where the bosons do not have enough time to hop too far from their initial positions can be simulated efficiently~\cite{Deshpande2018,Muraleedharan2018,Maskara2019}.
In particular, Ref.~\cite{Deshpande2018} considered a scenario where bosons were initially located at equal distances on a lattice and allowed to move in the lattice using only nearest-neighbor hoppings.
Using the Lieb-Robinson bounds, the authors constructed an early-time classical sampler that efficiently captures the dynamics of the bosons up to time $t_*$ that scales polynomially with the system size and thereby demonstrated a dynamical phase transition in the computational complexity.

The early-time classical sampler was later generalized to more complicated systems with power-law hoppings~\cite{Maskara2019}.
However, the easiness timescale $t_*$ in this case only scales polynomially with the system size for $\alpha>2d$ and scales logarithmically with the system size when $d+1<\alpha<2d$.
In this section, we show that the tight free-particle bound in this paper
immediately imply that $t_*$ scales polynomially with the system size for all $\alpha>d$, i.e. an exponentially longer easiness timescale in the regime $\alpha\in (d,2d]$ compared to the previous results~\cite{Maskara2019}.

For pedagogical reasons, we only describe here the high-level ideas behind the construction of the early-time boson sampler and argue for its efficiency using the technical results of Ref.~\cite{Maskara2019}.
We consider $N$ bosons hopping on a $d$-dimensional lattice under the Hamiltonian
\begin{align}
	H(t) = \sum_{i,j} J_{i,j}(t) b_{i}^\dag b_{j},
\end{align}
where $b_i$ is the bosonic annihilation operator on site $i$, $J_{i,j}(t)\leq 1/\dist(i,j)^\alpha$ are the hopping strengths, and the sums are over all sites $i,j$ on the lattice.
We assume that the lattice has $M \propto N^{\beta}$ sites in total, where $\beta\geq 1$ is a constant.

The bosons are initially located on evenly spaced sites on the lattice so that the minimum distance between nearest occupied sites is $2L \propto (M/N)^{1/d} \propto N^{(\beta-1)/d}$, as shown in Figure \ref{fig:initstate}.
Denote these initial positions by $j_1,\dots,j_N$.
We can write the initial state in terms of the creation operators:
\begin{align}
	\ket{\psi(0)} = \prod_{k=1}^{N} b^\dag_{j_k} \ket{\vac},
\end{align}
where $\ket{\vac}$ is the vacuum state.

The aim of boson sampling is to sample the positions of the bosons at a later time $t$.
The idea of the early-time boson sampler in Refs.~\cite{Deshpande2018,Maskara2019} is that each boson primarily hops within its causal light cone, i.e. a bubble of radius $r(t)$ centered on its initial position.
For small enough time, $r(t)<L$ and the bosons do not interfere with each other.
The state of the system at this time can be approximated by a product state over the bubbles and therefore the positions of the bosons can be efficiently simulated by simulating the dynamics of each boson independently.

Let $U(t) = \mathcal T\exp(-\ii\int_0^t \mathrm{d}s H(s) )$ be the evolution unitary generated by $H$ at time $t$.
By inserting pairs of $\mathbb I  = U^\dag U$ in between the creation operators, the state of the system at time $t$ can be written as
\begin{align}
	\ket{\psi(t)} = \prod_{k=1}^{N} U(t) b^\dag_{j_k} U^\dag(t) \ket{\vac}.
\end{align}
Here, the evolution of the state can be simplified into independent evolutions of the creation operators $b^\dag_{j_k}(t) \equiv U(t) b^\dag_{j_k} U^\dag(t)$.
Using our free-particle bound in \cref{thm:free}, we can approximate $b^\dag_{j_k}(t)$ by its evolution within a light cone originated from $j_k$:
\begin{align}
	b^\dag_{j_k}(t) \approx U_k(t) b^\dag_{j_k} U^\dag_k(t)\equiv \tilde b^\dag_{j_k}(t),\label{eq:boson-trunc}
\end{align}
where $U_k(t) = \mathcal T \exp[-\ii\int_0^t\mathrm{d}s H_k(s)]$ and $H_k$ is the Hamiltonian constructed from $H$ by taking only the hoppings between sites that are at most a distance $L$ from $j_k$.
Using Corollary \ref{cor:freesim}, the error of this approximation is $\OO{(Nt)^{3/2}/L^{(\alpha-d-\epsilon)/2}}$, where $\epsilon$ is an arbitrarily small positive constant and we have assumed $t\geq 1$ without loss of generality.
Repeating the approximations for all $k=1,\dots, N$, we thereby show that the state $\ket{\psi(t)}$ is approximately $\ket{\phi(t)} = \prod_{k} \tilde b_{j_k}(t)\ket{\vac}$.

\begin{figure}[t]
 \includegraphics[width=0.45\textwidth]{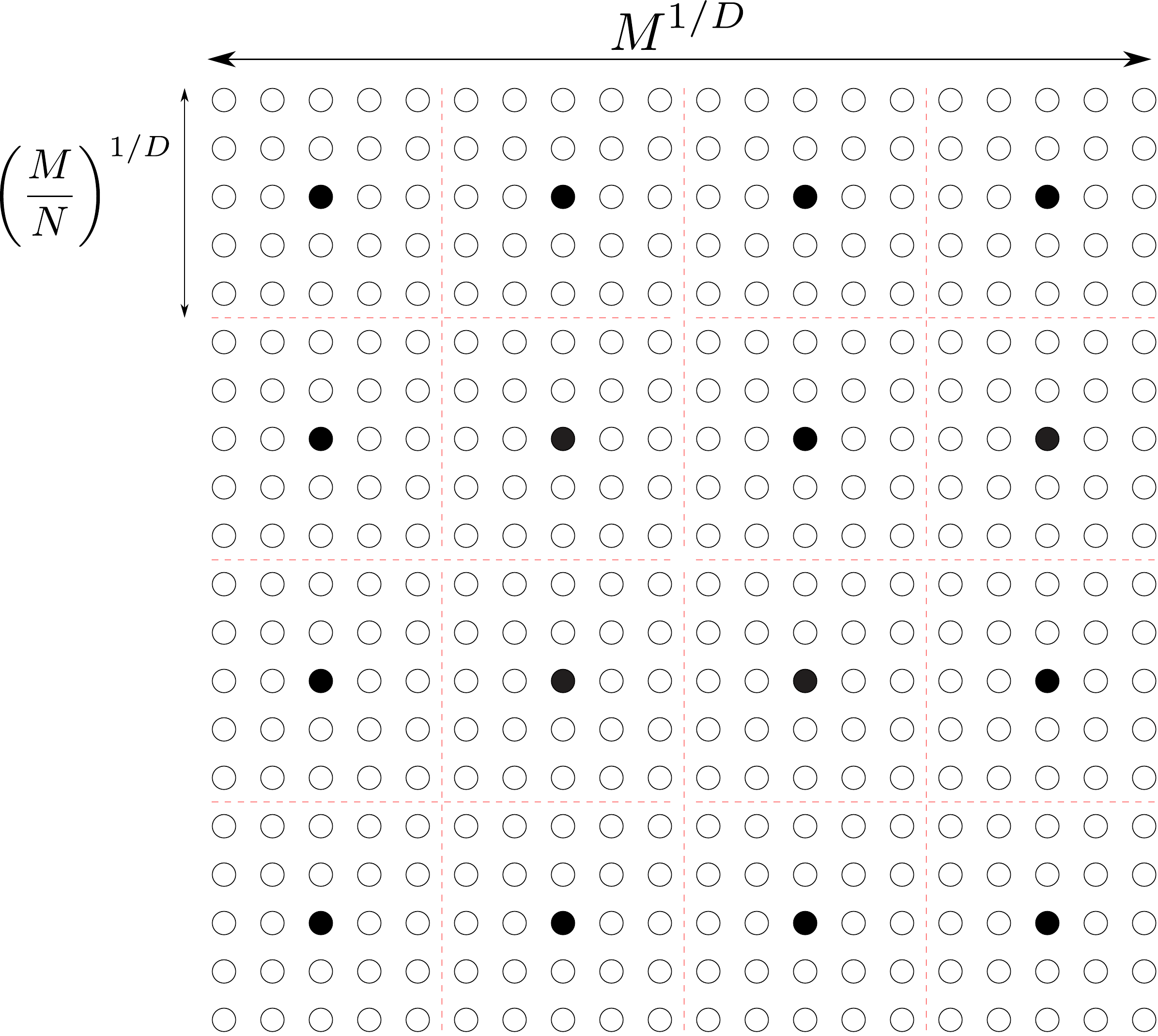}
 \caption{A depiction of the initial state in Ref.~\cite{Deshpande2018}.
Empty circles represent empty lattice site and filled circles represent initially occupied sites.
 } \label{fig:initstate}
\end{figure}

Since the operators $\tilde b_{j_k}(t)$ are supported on distinct regions, the bosons from different regions do not interfere with each other.
Therefore the probability distribution for the positions of the bosons in $\ket{\phi(t)}$ is simply the product of probability distributions of each boson hopping independently.
Thus, at later time, the positions of the bosons in $\ket{\phi(t)}$ can be efficiently sampled on a classical computer.

Note that the state $\ket{\phi(t)}$ only approximates $\ket{\psi(t)}$ up to some time $t_*$.
To estimate $t_*$, we calculate the total error of the approximation.
A simple calculation~\cite{Maskara2019} shows that the total error of approximating the $N$ original bosons $\{b^\dag(t)\}$ by the confined ones $\{\tilde b^\dag(t)\}$ would be $\OO{N^{5/2}t^{3/2}/L^{(\alpha-d-\epsilon)/2}}$---$N$ times the error of approximating each $b^\dag(t)$ by the corresponding $\tilde b^\dag(t)$.

Requiring that the total error of the approximation is at most a small constant, we obtain
\begin{align}
	t_* \propto L^{\frac{\alpha-d-\epsilon}{3}}N^{-\frac{5}{3}}\propto
	N^{\frac{(\beta-1)(\alpha-d-\epsilon)}{3d}-\frac{5}{3}},
\end{align}
where we have replaced $L\propto N^{(\beta-1)/d}$ from our assumption.
Therefore, by choosing a small enough $\epsilon$, the easiness time $t_*$ increases polynomially with $N$ for all $\alpha>d(1+\frac{5}{\beta-1})$.
In particular, the condition becomes
$\alpha>d$ in the limit of large $\beta$.
Therefore, our free-particle bound has improved the easiness time $t_*$ exponentially compared to Ref.~\cite{Maskara2019} in the regime $\alpha\in(d,2d]$.

\section{Generating topologically ordered states} \label{sec:topo}

We now study the minimum time it takes to create topologically ordered states from topologically trivial ones.
Before we present our result, we shall define
topologically ordered states and topologically trivial states following the definitions in Refs.\,\cite{Bravyi2006,Hastings2010}.
Suppose that the finite lattice $\Lambda$ has diameter $L$ and consists of $\Theta(L^d)$ sites.
Let $\{|\psi_{1}\rangle,\cdots|\psi_{k}\rangle\}$ be a
set of orthonormal quantum many-body states and define
\begin{align}
	\epsilon = \sup_{\mathcal O} \max_{1\leq i,j\leq k} \left\{\left|\langle\psi_{i}|\mathcal{O}|\psi_{i}\rangle-\langle\psi_{j}|\mathcal{O}|\psi_{j}\rangle\right|,
	2\left|\langle\psi_{i}|\mathcal{O}|\psi_{j}\rangle\right|
	\right\},
\end{align}
where the supremum is taken over unit-norm operators $\mathcal{O}$ supported on a subset of the lattice with diameter $l\leq bL$ for a constant $b<1$.
Roughly speaking, $\epsilon$ quantifies the ability to distinguish between the states $\{\ket{\psi_1},\dots,\ket{\psi_k}\}$ using observables that are supported on only a fraction of the lattice.
We say that the states are topologically ordered with diameter $bL$ if there exist constants $0<b<1$ and $c,\beta>0$ such that $\epsilon \le c L^{-\beta}$, and are topologically trivial if $\epsilon$ is independent of $L$ for all $b<1$ \cite{Gong2016}.
We now use the Lieb-Robinson bound to bound the minimum time it takes to convert between topologically ordered and topologically trivial states.

\begin{prop}\label{thm:topo}
	Consider a time-dependent Hamiltonian $H$ with long-range interactions of exponent $\alpha$ on ${\Lambda}$.  Let $U(t)$ be the evolution unitary generated by $H$ at time $t$, let $\{|\psi_{1}\rangle,\cdots|\psi_{k}\rangle\}$ be a set of topologically ordered states, and let $\{|\phi_{1}\rangle,\cdots|\phi_{k}\rangle\}$ be a set of topologically trivial states.  If $\alpha>2d+1$ and there is a time $0<\tau<\infty$ such that $|\psi_i\rangle = U(\tau)|\phi_i\rangle$, then there exists an $L$-independent constant $0<K<\infty$ such that $\tau > K \tau^*$, where
	\begin{align}\label{eq:taustar}
		\tau^* := \begin{cases}
			L & \text{when }\alpha>3d+1,\\
			L^{\frac{\alpha-2d}{d+1}}/\log^{\frac{2d}{d+1}}L & \text{when } 2d+1< \alpha\leq 3d+1.
		\end{cases}
	\end{align}
\end{prop}

\begin{proof}
Suppose $\{\ket{\psi_k}\}$ are topologically ordered with diameter $l' = bL$. 
Consider an arbitrary operator $\mathcal O$ with a support diameter at most $l=  l'/2$ and let $\mathcal O(t)\equiv U(t)\mathcal OU^{\dagger}(t)$ be the evolved version of $\mathcal O$.
We further introduce $\mathcal O(t,l^{\prime}) = \mathrm{tr}_{\mathcal B_{l'}^c} \mathcal O(t)$ as the version of $\mathcal O(t)$ truncated to a ball $\mathcal B_{l'}$ of diameter $l^{\prime}$.
Using the triangle inequality, we have
\begin{align}
\left|\langle\phi_{i}|\mathcal O|\phi_{i}\rangle-\langle\phi_{j}|\mathcal O|\phi_{j}\rangle\right|
=\left|\langle\psi_{i}|\mathcal O(\tau)|\psi_{i}\rangle-\langle\psi_{j}|\mathcal O(\tau)|\psi_{j}\rangle\right|
\le2\left\Vert \mathcal O(\tau)-\mathcal O(\tau,l^{\prime})\right\Vert +\left|\langle\psi_{i}|\mathcal O(\tau,l')|\psi_{i}\rangle-\langle\psi_{j}|\mathcal O(\tau,l')|\psi_{j}\rangle\right|.
\end{align}
Similarly, the off-diagonal elements satisfy
\begin{align}
    2\left|\langle\phi_{i}|\mathcal O|\phi_{j}\rangle\right| = 2\left|\langle\psi_{i}|\mathcal{O}(\tau)|\psi_{j}\rangle\right| \leq 2\left\Vert\mathcal{O}(\tau)-\mathcal{O}(\tau,l')\right\Vert + 2\left|\langle\psi_{i}|\mathcal O(\tau,l')|\psi_{j}\rangle\right|.
\end{align}
By our assumption of the absence of topological order in $\ket{\phi_{k}}$, we know there exists some operator $\mathcal{O}$ and some pair $i, j$ such that one of either $2\left|\langle\phi_{i}|\mathcal O|\phi_{j}\rangle\right|$ or $\left|\langle\phi_{i}|\mathcal O|\phi_{i}\rangle-\langle\phi_{j}|\mathcal O|\phi_{j}\rangle\right|$ is a constant $a_{1}$. Similarly, by our assumption on the presence of topological order in $\ket{\psi_{k}}$, we know that for any operator $\mathcal{O}$, both $2\left|\langle\psi_{i}|\mathcal O(\tau,l')|\psi_{j}\rangle\right|$ and $\left|\langle\psi_{i}|\mathcal O(\tau,l')|\psi_{i}\rangle-\langle\psi_{j}|\mathcal O(\tau,l')|\psi_{j}\rangle\right|$ are upper-bounded by $a_{2}/L^{\beta}$ for some constants $a_{2}, \beta$. Combining these two facts ensures that there exists some operator $\mathcal{O}$ such that
\begin{equation}
    a_1 -\frac{a_2}{L^\beta} \le 2 \left\Vert \mathcal O(\tau)-\mathcal O(\tau,l^{\prime})\right\Vert. \label{eq:topoineq2}
\end{equation}

We may bound the right-hand side of this expression using a standard trick from Ref.~\cite{Bravyi2006}:
\begin{equation}
    \norm{\mathcal{O}(\tau) - \mathcal{O}(\tau,l')} = \norm{\mathcal{O}(\tau) - \int_{\mathrm{supp}(U) \subset \mathcal{B}_{l'}^{c}}\mathrm{d}U U^{\dag}\mathcal{O}(\tau)U} \leq \int_{\mathrm{supp}(U) \subset \mathcal{B}_{l'}^{c}}\mathrm{d}U\norm{[\mathcal{O}(\tau), U]} 
    \leq 2\norm{\mathbb{P}_{\mathcal{B}_{l'}^{c}}\mathcal{O}(\tau)}.
\end{equation}
The first step rewrites the partial trace of $\mathcal{O}(\tau)$ as an integral over the Haar measure of unitaries whose support is contained within $\mathcal{B}_{l'}^{c}$, sites in the complement of a ball of radius $l'$ around $\mathcal{O}(\tau)$, and the second is a consequence of liberal use of the triangle inequality and unitary invariance of the operator norm. The third replaces the commutator norm with the projection norm, which is now $U$-independent, allowing us to evaluate the integral. Thus, using Proposition~\ref{prop:multisite-bound} for $\alpha>2d+1$ and assuming $\tau <bL/4\bar{v}$, where $\bar v$ is a constant, there exist $0<C_{1,2}<\infty$ such that

\begin{align}
	\left\Vert \mathcal O(\tau)-\mathcal O(\tau,l^{\prime})\right\Vert  \leq C_1L^{d}\frac{\tau^{d+1}\log^{2d}(l'-l)}{(l'-l-\bar{v}\tau)^{\alpha-d}}
	=C_2 \tau^{d+1}\frac{L^{d}\log^{2d}L}{(\frac{bL}{2}-\bar{v}\tau)^{\alpha-d}}, \label{eq:topoineq}
\end{align}
where the factor $L^d$ accounts for the support size of $\mathcal O$.
For $\alpha>3d+1$, Eq.~(\ref{eq:topoineq}) vanishes as $L$ increases, in contradiction with Eq.~(\ref{eq:topoineq2}), unless $\tau$ scales linearly in $L$. For $2d+1 < \alpha \leq 3d+1$, $\tau$ must scale at least as $L^{\frac{\alpha-2d}{d+1}}/\log^{\frac{2d}{d+1}}L$ in order for Eq.~(\ref{eq:topoineq}) to no longer be nonvanishing as $L$ increases. These two results complete the proof of Eq.~(\ref{eq:taustar}).
\end{proof}

\section{Clustering of correlations}\label{sec:clustering-corr}
In addition to the dynamics of quantum systems, the Lieb-Robinson bounds also have implications for the eigenstates of a Hamiltonian.
In Ref.~\cite{Hastings2006}, the authors show that if a time-independent power-law Hamiltonian with an exponent $\alpha$ has spectral gap $\Delta>0$, the correlations between distant sites in the ground state of the system also decay with the distance as a power law with an exponent lower bounded by
\begin{align}
	\alpha' = \frac{\alpha}{1+\tilde v\Delta^{-2}},
\end{align}
where $\tilde v$ is a constant that depends on $\alpha$.

The bound in Ref.~\cite{Hastings2006} has a undesirable feature: for a given value of $\alpha$, varying the gap $\Delta$ also changes the minimum exponent $\alpha'$.
Although this leads to an intuitive implication that larger energy gaps result in faster correlation decay, there is no known example where ground state correlations decay at a slower rate than a power law with an exponent
$\alpha$.
Indeed, we shall show that the cause of this problem is tied to the previous lack of an algebraic light cone in the quench dynamics.
In particular, by using the Lieb-Robinson bounds with algebraic light cones~\cite{Foss-Feig2015,Tran2019,Else2018,Chen2019,Kuwahara2019}, we show for all $\alpha>2d$ that the ground state correlations must decay as a power law with the exponent lower bounded by the exponent of the Hamiltonian.

\begin{prop}\label{lem:clustering}
Let $H$ be a power-law Hamiltonian with an exponent $\alpha$; let $A,B$ be local operators obeying $\lVert A\rVert, \lVert B\rVert \le 1$, supported on $X,Y$ such that $\abs{X}=\abs{Y}=1$ and $\dist(X,Y)=r>0$.
Assume that $H$ has a unique ground state $\ket{\psi_0}$ and spectral gap $\Delta$ to the first excited state.
Define $\mathcal C(r)  \coloneqq  \langle\psi_{0}|AB|\psi_{0}\rangle-\langle\psi_{0}|A|\psi_{0}\rangle\langle\psi_{0}|B|\psi_{0}\rangle$ to be the connected correlator between $A,B$ in the ground state.
Then
\begin{align}
	\left|\mathcal C(r)\right|
&\le\left[\frac{2^{\gamma-1}c\Gamma(\frac{\gamma}{2})}{\pi}\frac{\alpha^{\gamma/2}
}{\Delta^{\gamma}}+1\right]
\frac{\log^{\gamma/2}r}{r^{\alpha}},
\end{align}
where $c$ is a constant independent of $\alpha$, $\gamma = \alpha(\alpha-d+1)/(\alpha-2d)$, and $\Gamma(\cdot)$ is the Gamma function.
\end{prop}
\begin{proof}
First we rewrite
\begin{align}
	\mathcal C(r)  = \sum_{k>0} \bra{\psi_0} A\ket{\psi_k}\bra{\psi_k}B\ket{\psi_0},\label{eq:Cr}
\end{align}
where the sum is over the excited states $\ket{\psi_k}$ of the Hamiltonian.
Our strategy is to relate $\mathcal C(r)$ to the commutator norm $\left\Vert [A(t),B]\right\Vert $, which we then bound using a Lieb-Robinson bound.
To relate $\mathcal C(r)$ to $\left\Vert [A(t),B]\right\Vert $, it is natural
to first consider the value of $[A(t),B]$ in the ground state, whose
magnitude is bounded by $\left\Vert [A(t),B]\right\Vert $:
\begin{align}
\langle\psi_{0}|[A(t),B]|\psi_{0}\rangle
=  \langle\psi_{0}|A(t)B|\psi_{0}\rangle-\text{h.c.} =  \sum_{k>0}\mathrm{e}^{\ii E_{k}t}\langle\psi_{0}|A|\psi_{k}\rangle\langle\psi_{k}|B|\psi_{0}\rangle-\text{h.c.},\label{eq:AB}
\end{align}
where $E_k$ are the eigenvalues of the Hamiltonian and we have set ground state energy $E_{0}=0$ so that $E_{k}>0$ for all $k>0$.  Note that the $k=0$ terms cancel between the first term and its Hermitian conjugate.

By observation, we note that if we could replace the terms $\mathrm{e}^{\ii E_{k}t}$ in \cref{eq:AB} by a unit step function
$\Theta(E_{k})$ that satisfies $\Theta(E_{k})=1$ and $\Theta(-E_{k})=0$, we immediately obtain the expression of $\mathcal C(r)$ in \cref{eq:Cr}.
In fact, this replacement is easy to achieve using the identity
\begin{equation}
\lim_{\epsilon\rightarrow0^{+}}\frac{1}{2\pi \ii}\int\limits_{-\infty}^{\infty}\mathrm{d}t\frac{\mathrm{e}^{\ii E_{k}t}}{t-\ii \epsilon}=\Theta(E_{k}).\label{eq:stepint}
\end{equation}
Therefore, we have
\begin{equation}
\lim_{\epsilon\rightarrow0^{+}}\frac{1}{2\pi \ii}\int\limits_{-\infty}^{\infty}\mathrm{d}t\frac{\langle\psi_{0}|[A(t),B]|\psi_{0}\rangle}{t-\ii \epsilon}=\mathcal C(r),\label{eq:lim}
\end{equation}
and we obtain the relation
\begin{align}
\left|\mathcal C(r)\right|
&=\left|\lim_{\epsilon\rightarrow0^{+}}\frac{1}{2\pi \ii}\int\limits_{-\infty}^{\infty}\mathrm{d}t\frac{\langle\psi_{0}|[A(t),B]|\psi_{0}\rangle}{t-\ii\epsilon}\right| \le\frac{1}{\pi}\int\limits_{0}^{\infty}\mathrm{d}t \frac{\left\Vert [A(t),B]\right\Vert }{t}.\label{eq:Cint}
\end{align}

Unfortunately, this relation is not useful; the right-hand side of \cref{eq:Cint} diverges even when the commutator $\left\Vert [A(t),B]\right\Vert $ does not increase with time.
The failure of such a simple treatment is not surprising
as we have not used the crucial assumption on the existence of a finite energy
gap ($E_{k}\ge\Delta$).

Intuitively, to make the integral in \cref{eq:Cint} converge, we can multiply the integrand by a function that decays quickly with $t$.
A natural choice is a Gaussian function $\mathrm{e}^{-(\upsilon t/2)^{2}}$, where $\upsilon>0$ is an adjustable parameter; it decays with time quickly enough to make the integral converge and
its Fourier transformation is rather easy to handle.
By multiplying this function to the integrand in \cref{eq:stepint},
we arrive at a convolution of the step function with the Gaussian function:
\begin{align}
&\lim_{\epsilon\rightarrow0^{+}}\frac{1}{2\pi \ii}\int\limits_{-\infty}^{\infty}\mathrm{d}t\frac{\mathrm{e}^{\ii E_{k}t}\mathrm{e}^{-(\upsilon t/2)^{2}}}{t-i\epsilon}  =  \frac{1}{\sqrt{\pi}\upsilon}\int\limits_{-\infty}^{\infty}\Theta(E_{k}-E)\mathrm{e}^{-E^{2}/\upsilon^{2}}dE
\eqqcolon f(E_{k}).
\end{align}
It is easy to verify that $f(E_{k})=1-g(E_{k})$
and $f(-E_{k})=0+g(E_{k})$ for some positive function $g(E_{k})\le\frac{1}{2}\mathrm{e}^{-(E_{k}/\upsilon)^{2}}$.
Thus, $f(E_{k})$ closely resembles the step function $\Theta(E_{k})$,
albeit with a smoother transition from $0$ to $1$.

Inserting this convolution into \cref{eq:lim}, we have:
\begin{align}
&\lim_{\epsilon\rightarrow0^{+}}\frac{1}{2\pi \ii}\int\limits_{-\infty}^{\infty}\mathrm{d}t\frac{\langle\psi_{0}|[A(t),B]|\psi_{0}\rangle \mathrm{e}^{-(\upsilon t/2)^{2}}}{t-\ii\epsilon}  =  \mathcal C(r)-\sum_{k>0}g(E_{k})[\langle\psi_{0}|A|\psi_{n}\rangle\langle\psi_{n}|B|\psi_{0}\rangle+\text{h.c.}].
\end{align}
Using a Cauchy-Schwarz inequality, we can then bound the absolute value of the sum over $k$ in the right-hand side by
\begin{align}
	 \sum_{k>0}2g(E_{k})\left|\langle\psi_{0}|A|\psi_{n}\rangle\langle\psi_{n}|B|\psi_{0}\rangle\right|\le \mathrm{e}^{-(\Delta/\upsilon)^{2}},
\end{align}
where we have used that $E_k \ge \Delta$.
Thus we arrive at our desired relation:
\begin{equation}
\left|\mathcal C(r)\right|\le\frac{1}{\pi}\int\limits_{0}^{\infty}\mathrm{d}t\frac{\mathrm{e}^{-(\upsilon t/2)^{2}}}{t}\left\Vert [A(t),B]\right\Vert +\mathrm{e}^{-(\Delta/\upsilon)^{2}}.\label{eq:relation}
\end{equation}
Finally, we bound the commutator norm using the Lieb-Robinson bound in Ref.~\cite{Foss-Feig2015},
\begin{align}
	\norm{\comm{A(t),B}}\leq c \frac{t^\gamma}{r^\alpha},
\end{align}
where $c$ is a constant and $\gamma = \alpha(\alpha-d+1)/(\alpha-2d)$.
We obtain:
\begin{align}
\left|\mathcal C(r)\right|\leq\frac{2^{\gamma-1}c\Gamma(\frac{\gamma}{2})}{\pi}
\frac{1}{\upsilon^{\gamma}r^{\alpha}} +\mathrm{e}^{-(\Delta/\upsilon)^{2}},\label{eq:relation-after-LR}
\end{align}
where $\Gamma(\cdot)$ is the Gamma function.
By choosing $\upsilon = {\Delta}/{\sqrt{\alpha \log r}} $, we get
\begin{align}
\left|\mathcal C(r)\right|
&\le\left[\frac{2^{\gamma-1}c\Gamma(\frac{\gamma}{2})}{\pi}\frac{\alpha^{\gamma/2}
}{\Delta^{\gamma}}+1\right]
\frac{\log^{\gamma/2}r}{r^{\alpha}}. \label{eq:cluster-corr-bound}
\end{align}
Therefore, the correlators in the ground state of a power-law Hamiltonian with $\alpha>2d$ also decay with the distance as a power law (up to a logarithmic correction) with the same exponent $\alpha$ as that of the Hamiltonian.
In particular, this exponent is independent of the energy gap $\Delta$, in contrast to the previous result in Ref.~\cite{Hastings2006}.
\end{proof}

Note that in \cref{eq:relation-after-LR}, we have used an algebraic light cone bound from \cite{Foss-Feig2015} instead of the tighter bounds in recent works~\cite{Tran2019,Else2018,Chen2019,Kuwahara2019},
because the bounds in Refs.~\cite{Tran2019,Else2018,Chen2019} decay with the distance slower than $1/r^\alpha$ while the bound in Ref.~\cite{Kuwahara2019} does not hold for $2d<\alpha\leq 2d+1$.

\section{Conclusion}
We have demonstrated a \emph{hierarchy of linear light cones}---a sequence of metrics and protocols under which the emergent locality that arises in local quantum many-body systems breaks down at different exponents $\alpha$ of long-range interactions.  The most general such light cone---the Lieb-Robinson light cone that bounds commutator norms---can become superlinear for any $\alpha<2d+1$.  We conjecture that the Frobenius light cone that controls many-body chaos and state transfer can only be superlinear when $\alpha<1+\frac{3}{2}d$, and proved this result in $d=1$ using the operator quantum walk formalism.  Finally, in non-interacting systems, we proved both linear ($\alpha>d+1$) and superlinear ($d<\alpha\le d+1$) light cones along with the optimality of these bounds.  As such, we close a number of long-standing questions in the literature on the limits and capabilities of quantum dynamics with long-range interactions.

Besides state transfer and many-body chaos, we have also demonstrated a wide range of applications of these (nearly) tight light cones.
We proved that the growth of connected correlations obey the same light cone as that of the Lieb-Robinson bound.
In the context of digital quantum simulation, we used the Lieb-Robinson bound to construct an approximation for the time-evolved version of a local observable, and thereby reduced cost of simulating the observable on quantum computers for all $\alpha>2d+1$.
Similarly, using the free light cone, we constructed an efficient early-time classical boson sampler for all $\alpha>d$, exponentially improving the previous best estimate in some regime of $\alpha$.
Additionally, we bounded the time it takes to generate topologically ordered states using power-law interactions.
Finally, we tightened the minimum correlation-decay rate in the ground state of a gapped power-law Hamiltonian.

The hierarchy of linear light cones revealed in this paper has important implications both on the capabilities of quantum technologies exploiting long-range interactions, as well as on the nature of quantum information dynamics and thermalization in these systems.  A complete understanding of quantum chaos and state transfer, at the very least, requires the construction of a new mathematical framework beyond the Lieb-Robinson bounds, perhaps along the lines of our operator quantum walk.
It is an interesting open question whether and how the hierarchy of different notions of locality revealed in this manuscript reveals itself in aspects of quantum chaos besides OTOCs, perhaps including entanglement dynamics or eigenvalue statistics.
The tightness of the superlinear polynomial light cone $t\sim r^{\alpha-1}$, found in $d=1$ for $2<\alpha<3$ in Ref. \cite{Chen2019} as well as the existence of algebraic light cone below $\alpha=2d$ in ~\cite{Foss-Feig2015,Tran2019,Else2018}, remains an open problem.  Lastly, it
 also remains an important future challenge to obtain the Frobenius light cone in two or more dimensions, as well as to rigorously study the light cone that controls the decoherence of a quantum system subject to long-range random noise, which was conjectured to be linear for  $\alpha>d+\frac{1}{2}$ \cite{Zhou2019b}.

\section*{Acknowledgments}
We thank Ana Maria Rey for useful discussions.  MCT, AYG, AD, AE, and AVG acknowledge support by the DoE ASCR Quantum Testbed Pathfinder program (award No.\ DE-SC0019040), DoE ASCR FAR-QC (award No.\ DE-SC0020312),  NSF PFCQC program, AFOSR MURI, AFOSR, ARO MURI, DoE BES Materials and Chemical Sciences Research for Quantum Information Science program (award No.\ DE-SC0019449),  ARL CDQI, and NSF PFC at JQI.
CFC is supported by the Millikan graduate fellowship at Caltech.
AYG is supported by the NSF Graduate Research Fellowship Program under Grant No. DGE-1840340.
AE also acknowledges funding from the DoD.
ZXG is supported by the NSF RAISE-TAQS program under Grant No. CCF-1839232.

\bibliography{lieb-robinson}

\begin{thebibliography}{57}%
\makeatletter
\providecommand \@ifxundefined [1]{%
 \@ifx{#1\undefined}
}%
\providecommand \@ifnum [1]{%
 \ifnum #1\expandafter \@firstoftwo
 \else \expandafter \@secondoftwo
 \fi
}%
\providecommand \@ifx [1]{%
 \ifx #1\expandafter \@firstoftwo
 \else \expandafter \@secondoftwo
 \fi
}%
\providecommand \natexlab [1]{#1}%
\providecommand \enquote  [1]{``#1''}%
\providecommand \bibnamefont  [1]{#1}%
\providecommand \bibfnamefont [1]{#1}%
\providecommand \citenamefont [1]{#1}%
\providecommand \href@noop [0]{\@secondoftwo}%
\providecommand \href [0]{\begingroup \@sanitize@url \@href}%
\providecommand \@href[1]{\@@startlink{#1}\@@href}%
\providecommand \@@href[1]{\endgroup#1\@@endlink}%
\providecommand \@sanitize@url [0]{\catcode `\\12\catcode `\$12\catcode
  `\&12\catcode `\#12\catcode `\^12\catcode `\_12\catcode `\%12\relax}%
\providecommand \@@startlink[1]{}%
\providecommand \@@endlink[0]{}%
\providecommand \url  [0]{\begingroup\@sanitize@url \@url }%
\providecommand \@url [1]{\endgroup\@href {#1}{\urlprefix }}%
\providecommand \urlprefix  [0]{URL }%
\providecommand \Eprint [0]{\href }%
\providecommand \doibase [0]{http://dx.doi.org/}%
\providecommand \selectlanguage [0]{\@gobble}%
\providecommand \bibinfo  [0]{\@secondoftwo}%
\providecommand \bibfield  [0]{\@secondoftwo}%
\providecommand \translation [1]{[#1]}%
\providecommand \BibitemOpen [0]{}%
\providecommand \bibitemStop [0]{}%
\providecommand \bibitemNoStop [0]{.\EOS\space}%
\providecommand \EOS [0]{\spacefactor3000\relax}%
\providecommand \BibitemShut  [1]{\csname bibitem#1\endcsname}%
\let\auto@bib@innerbib\@empty
\bibitem [{\citenamefont {Lieb}\ and\ \citenamefont
  {Robinson}(1972)}]{Lieb1972}%
  \BibitemOpen
  \bibfield  {author} {\bibinfo {author} {\bibfnamefont {Elliott~H.}\
  \bibnamefont {Lieb}}\ and\ \bibinfo {author} {\bibfnamefont {Derek~W.}\
  \bibnamefont {Robinson}},\ }\bibfield  {title} {\enquote {\bibinfo {title}
  {The finite group velocity of quantum spin systems},}\ }\href {\doibase
  10.1007/BF01645779} {\bibfield  {journal} {\bibinfo  {journal} {Commun. Math.
  Phys.}\ }\textbf {\bibinfo {volume} {28}},\ \bibinfo {pages} {251--257}
  (\bibinfo {year} {1972})}\BibitemShut {NoStop}%
\bibitem [{\citenamefont {Aikawa}\ \emph {et~al.}(2012)\citenamefont {Aikawa},
  \citenamefont {Frisch}, \citenamefont {Mark}, \citenamefont {Baier},
  \citenamefont {Rietzler}, \citenamefont {Grimm},\ and\ \citenamefont
  {Ferlaino}}]{Aikawa2012}%
  \BibitemOpen
  \bibfield  {author} {\bibinfo {author} {\bibfnamefont {K.}~\bibnamefont
  {Aikawa}}, \bibinfo {author} {\bibfnamefont {A.}~\bibnamefont {Frisch}},
  \bibinfo {author} {\bibfnamefont {M.}~\bibnamefont {Mark}}, \bibinfo {author}
  {\bibfnamefont {S.}~\bibnamefont {Baier}}, \bibinfo {author} {\bibfnamefont
  {A.}~\bibnamefont {Rietzler}}, \bibinfo {author} {\bibfnamefont
  {R.}~\bibnamefont {Grimm}}, \ and\ \bibinfo {author} {\bibfnamefont
  {F.}~\bibnamefont {Ferlaino}},\ }\bibfield  {title} {\enquote {\bibinfo
  {title} {Bose-einstein condensation of erbium},}\ }\href {\doibase
  10.1103/PhysRevLett.108.210401} {\bibfield  {journal} {\bibinfo  {journal}
  {Phys. Rev. Lett.}\ }\textbf {\bibinfo {volume} {108}},\ \bibinfo {pages}
  {210401} (\bibinfo {year} {2012})}\BibitemShut {NoStop}%
\bibitem [{\citenamefont {Saffman}\ \emph {et~al.}(2010)\citenamefont
  {Saffman}, \citenamefont {Walker},\ and\ \citenamefont
  {Mølmer}}]{Saffman2010}%
  \BibitemOpen
  \bibfield  {author} {\bibinfo {author} {\bibfnamefont {M.}~\bibnamefont
  {Saffman}}, \bibinfo {author} {\bibfnamefont {T.~G.}\ \bibnamefont {Walker}},
  \ and\ \bibinfo {author} {\bibfnamefont {K.}~\bibnamefont {Mølmer}},\
  }\bibfield  {title} {\enquote {\bibinfo {title} {Quantum information with
  {{Rydberg}} atoms},}\ }\href {\doibase 10.1103/RevModPhys.82.2313} {\bibfield
   {journal} {\bibinfo  {journal} {Rev. Mod. Phys.}\ }\textbf {\bibinfo
  {volume} {82}},\ \bibinfo {pages} {2313--2363} (\bibinfo {year}
  {2010})}\BibitemShut {NoStop}%
\bibitem [{\citenamefont {Britton}\ \emph {et~al.}(2012)\citenamefont
  {Britton}, \citenamefont {Sawyer}, \citenamefont {Keith}, \citenamefont
  {Wang}, \citenamefont {Freericks}, \citenamefont {Uys}, \citenamefont
  {Biercuk},\ and\ \citenamefont {Bollinger}}]{Britton2012}%
  \BibitemOpen
  \bibfield  {author} {\bibinfo {author} {\bibfnamefont {Joseph~W.}\
  \bibnamefont {Britton}}, \bibinfo {author} {\bibfnamefont {Brian~C.}\
  \bibnamefont {Sawyer}}, \bibinfo {author} {\bibfnamefont {Adam~C.}\
  \bibnamefont {Keith}}, \bibinfo {author} {\bibfnamefont {C.-C.~Joseph}\
  \bibnamefont {Wang}}, \bibinfo {author} {\bibfnamefont {James~K.}\
  \bibnamefont {Freericks}}, \bibinfo {author} {\bibfnamefont {Hermann}\
  \bibnamefont {Uys}}, \bibinfo {author} {\bibfnamefont {Michael~J.}\
  \bibnamefont {Biercuk}}, \ and\ \bibinfo {author} {\bibfnamefont {John~J.}\
  \bibnamefont {Bollinger}},\ }\bibfield  {title} {\enquote {\bibinfo {title}
  {Engineered two-dimensional {{Ising}} interactions in a trapped-ion quantum
  simulator with hundreds of spins},}\ }\href {\doibase 10.1038/nature10981}
  {\bibfield  {journal} {\bibinfo  {journal} {Nature}\ }\textbf {\bibinfo
  {volume} {484}},\ \bibinfo {pages} {489--492} (\bibinfo {year}
  {2012})}\BibitemShut {NoStop}%
\bibitem [{\citenamefont {Yan}\ \emph {et~al.}(2013)\citenamefont {Yan},
  \citenamefont {Moses}, \citenamefont {Gadway}, \citenamefont {Covey},
  \citenamefont {Hazzard}, \citenamefont {Rey}, \citenamefont {Jin},\ and\
  \citenamefont {Ye}}]{Yan2013}%
  \BibitemOpen
  \bibfield  {author} {\bibinfo {author} {\bibfnamefont {Bo}~\bibnamefont
  {Yan}}, \bibinfo {author} {\bibfnamefont {Steven~A.}\ \bibnamefont {Moses}},
  \bibinfo {author} {\bibfnamefont {Bryce}\ \bibnamefont {Gadway}}, \bibinfo
  {author} {\bibfnamefont {Jacob~P.}\ \bibnamefont {Covey}}, \bibinfo {author}
  {\bibfnamefont {Kaden R.~A.}\ \bibnamefont {Hazzard}}, \bibinfo {author}
  {\bibfnamefont {Ana~Maria}\ \bibnamefont {Rey}}, \bibinfo {author}
  {\bibfnamefont {Deborah~S.}\ \bibnamefont {Jin}}, \ and\ \bibinfo {author}
  {\bibfnamefont {Jun}\ \bibnamefont {Ye}},\ }\bibfield  {title} {\enquote
  {\bibinfo {title} {Observation of dipolar spin-exchange interactions with
  lattice-confined polar molecules},}\ }\href {\doibase 10.1038/nature12483}
  {\bibfield  {journal} {\bibinfo  {journal} {Nature}\ }\textbf {\bibinfo
  {volume} {501}},\ \bibinfo {pages} {521--525} (\bibinfo {year}
  {2013})}\BibitemShut {NoStop}%
\bibitem [{\citenamefont {Yao}\ \emph {et~al.}(2012)\citenamefont {Yao},
  \citenamefont {Jiang}, \citenamefont {Gorshkov}, \citenamefont {Maurer},
  \citenamefont {Giedke}, \citenamefont {Cirac},\ and\ \citenamefont
  {Lukin}}]{Yao2012}%
  \BibitemOpen
  \bibfield  {author} {\bibinfo {author} {\bibfnamefont {N.Y.}\ \bibnamefont
  {Yao}}, \bibinfo {author} {\bibfnamefont {L.}~\bibnamefont {Jiang}}, \bibinfo
  {author} {\bibfnamefont {A.V.}\ \bibnamefont {Gorshkov}}, \bibinfo {author}
  {\bibfnamefont {P.C.}\ \bibnamefont {Maurer}}, \bibinfo {author}
  {\bibfnamefont {G.}~\bibnamefont {Giedke}}, \bibinfo {author} {\bibfnamefont
  {J.I.}\ \bibnamefont {Cirac}}, \ and\ \bibinfo {author} {\bibfnamefont
  {M.D.}\ \bibnamefont {Lukin}},\ }\bibfield  {title} {\enquote {\bibinfo
  {title} {Scalable architecture for a room temperature solid-state quantum
  information processor},}\ }\href {\doibase 10.1038/ncomms1788} {\bibfield
  {journal} {\bibinfo  {journal} {Nat Commun}\ }\textbf {\bibinfo {volume}
  {3}},\ \bibinfo {pages} {1--8} (\bibinfo {year} {2012})}\BibitemShut
  {NoStop}%
\bibitem [{\citenamefont {Douglas}\ \emph {et~al.}(2015)\citenamefont
  {Douglas}, \citenamefont {Habibian}, \citenamefont {Hung}, \citenamefont
  {Gorshkov}, \citenamefont {Kimble},\ and\ \citenamefont
  {Chang}}]{Douglas2015}%
  \BibitemOpen
  \bibfield  {author} {\bibinfo {author} {\bibfnamefont {J.~S.}\ \bibnamefont
  {Douglas}}, \bibinfo {author} {\bibfnamefont {H.}~\bibnamefont {Habibian}},
  \bibinfo {author} {\bibfnamefont {C.-L.}\ \bibnamefont {Hung}}, \bibinfo
  {author} {\bibfnamefont {A.~V.}\ \bibnamefont {Gorshkov}}, \bibinfo {author}
  {\bibfnamefont {H.~J.}\ \bibnamefont {Kimble}}, \ and\ \bibinfo {author}
  {\bibfnamefont {D.~E.}\ \bibnamefont {Chang}},\ }\bibfield  {title} {\enquote
  {\bibinfo {title} {Quantum many-body models with cold atoms coupled to
  photonic crystals},}\ }\href {\doibase 10.1038/nphoton.2015.57} {\bibfield
  {journal} {\bibinfo  {journal} {Nat. Photonics}\ }\textbf {\bibinfo {volume}
  {9}},\ \bibinfo {pages} {326--331} (\bibinfo {year} {2015})}\BibitemShut
  {NoStop}%
\bibitem [{\citenamefont {Eldredge}\ \emph {et~al.}(2017)\citenamefont
  {Eldredge}, \citenamefont {Gong}, \citenamefont {Young}, \citenamefont
  {Moosavian}, \citenamefont {{Foss-Feig}},\ and\ \citenamefont
  {Gorshkov}}]{Eldredge2017}%
  \BibitemOpen
  \bibfield  {author} {\bibinfo {author} {\bibfnamefont {Zachary}\ \bibnamefont
  {Eldredge}}, \bibinfo {author} {\bibfnamefont {Zhe-Xuan}\ \bibnamefont
  {Gong}}, \bibinfo {author} {\bibfnamefont {Jeremy~T.}\ \bibnamefont {Young}},
  \bibinfo {author} {\bibfnamefont {Ali~Hamed}\ \bibnamefont {Moosavian}},
  \bibinfo {author} {\bibfnamefont {Michael}\ \bibnamefont {{Foss-Feig}}}, \
  and\ \bibinfo {author} {\bibfnamefont {Alexey~V.}\ \bibnamefont {Gorshkov}},\
  }\bibfield  {title} {\enquote {\bibinfo {title} {Fast {{Quantum State
  Transfer}} and {{Entanglement Renormalization Using Long}}-{{Range
  Interactions}}},}\ }\href {\doibase 10.1103/PhysRevLett.119.170503}
  {\bibfield  {journal} {\bibinfo  {journal} {Phys. Rev. Lett.}\ }\textbf
  {\bibinfo {volume} {119}},\ \bibinfo {pages} {170503} (\bibinfo {year}
  {2017})}\BibitemShut {NoStop}%
\bibitem [{\citenamefont {Linke}\ \emph {et~al.}(2017)\citenamefont {Linke},
  \citenamefont {Maslov}, \citenamefont {Roetteler}, \citenamefont {Debnath},
  \citenamefont {Figgatt}, \citenamefont {Landsman}, \citenamefont {Wright},\
  and\ \citenamefont {Monroe}}]{linkeExperimentalComparisonTwo2017}%
  \BibitemOpen
  \bibfield  {author} {\bibinfo {author} {\bibfnamefont {Norbert~M.}\
  \bibnamefont {Linke}}, \bibinfo {author} {\bibfnamefont {Dmitri}\
  \bibnamefont {Maslov}}, \bibinfo {author} {\bibfnamefont {Martin}\
  \bibnamefont {Roetteler}}, \bibinfo {author} {\bibfnamefont {Shantanu}\
  \bibnamefont {Debnath}}, \bibinfo {author} {\bibfnamefont {Caroline}\
  \bibnamefont {Figgatt}}, \bibinfo {author} {\bibfnamefont {Kevin~A.}\
  \bibnamefont {Landsman}}, \bibinfo {author} {\bibfnamefont {Kenneth}\
  \bibnamefont {Wright}}, \ and\ \bibinfo {author} {\bibfnamefont
  {Christopher}\ \bibnamefont {Monroe}},\ }\bibfield  {title} {\enquote
  {\bibinfo {title} {Experimental comparison of two quantum computing
  architectures},}\ }\href {\doibase 10.1073/pnas.1618020114} {\bibfield
  {journal} {\bibinfo  {journal} {Proceedings of the National Academy of
  Sciences}\ }\textbf {\bibinfo {volume} {114}},\ \bibinfo {pages} {3305--3310}
  (\bibinfo {year} {2017})}\BibitemShut {NoStop}%
\bibitem [{\citenamefont {Deshpande}\ \emph {et~al.}(2018)\citenamefont
  {Deshpande}, \citenamefont {Fefferman}, \citenamefont {Tran}, \citenamefont
  {{Foss-Feig}},\ and\ \citenamefont {Gorshkov}}]{Deshpande2018}%
  \BibitemOpen
  \bibfield  {author} {\bibinfo {author} {\bibfnamefont {Abhinav}\ \bibnamefont
  {Deshpande}}, \bibinfo {author} {\bibfnamefont {Bill}\ \bibnamefont
  {Fefferman}}, \bibinfo {author} {\bibfnamefont {Minh~C.}\ \bibnamefont
  {Tran}}, \bibinfo {author} {\bibfnamefont {Michael}\ \bibnamefont
  {{Foss-Feig}}}, \ and\ \bibinfo {author} {\bibfnamefont {Alexey~V.}\
  \bibnamefont {Gorshkov}},\ }\bibfield  {title} {\enquote {\bibinfo {title}
  {Dynamical {{Phase Transitions}} in {{Sampling Complexity}}},}\ }\href
  {\doibase 10.1103/PhysRevLett.121.030501} {\bibfield  {journal} {\bibinfo
  {journal} {Phys. Rev. Lett.}\ }\textbf {\bibinfo {volume} {121}},\ \bibinfo
  {pages} {030501} (\bibinfo {year} {2018})}\BibitemShut {NoStop}%
\bibitem [{\citenamefont {Tran}\ \emph
  {et~al.}(2019{\natexlab{a}})\citenamefont {Tran}, \citenamefont {Ehrenberg},
  \citenamefont {Guo}, \citenamefont {Titum}, \citenamefont {Abanin},\ and\
  \citenamefont {Gorshkov}}]{Tran2019}%
  \BibitemOpen
  \bibfield  {author} {\bibinfo {author} {\bibfnamefont {Minh~C.}\ \bibnamefont
  {Tran}}, \bibinfo {author} {\bibfnamefont {Adam}\ \bibnamefont {Ehrenberg}},
  \bibinfo {author} {\bibfnamefont {Andrew~Y.}\ \bibnamefont {Guo}}, \bibinfo
  {author} {\bibfnamefont {Paraj}\ \bibnamefont {Titum}}, \bibinfo {author}
  {\bibfnamefont {Dmitry~A.}\ \bibnamefont {Abanin}}, \ and\ \bibinfo {author}
  {\bibfnamefont {Alexey~V.}\ \bibnamefont {Gorshkov}},\ }\bibfield  {title}
  {\enquote {\bibinfo {title} {Locality and {{Heating}} in {{Periodically
  Driven}}, {{Power}}-law {{Interacting Systems}}},}\ }\href {\doibase
  10.1103/PhysRevA.100.052103} {\bibfield  {journal} {\bibinfo  {journal}
  {Phys. Rev. A}\ }\textbf {\bibinfo {volume} {100}},\ \bibinfo {pages}
  {052103} (\bibinfo {year} {2019}{\natexlab{a}})}\BibitemShut {NoStop}%
\bibitem [{\citenamefont {Landsman}\ \emph {et~al.}(2019)\citenamefont
  {Landsman}, \citenamefont {Figgatt}, \citenamefont {Schuster}, \citenamefont
  {Linke}, \citenamefont {Yoshida}, \citenamefont {Yao},\ and\ \citenamefont
  {Monroe}}]{landsmanVerifiedQuantumInformation2019}%
  \BibitemOpen
  \bibfield  {author} {\bibinfo {author} {\bibfnamefont {K.~A.}\ \bibnamefont
  {Landsman}}, \bibinfo {author} {\bibfnamefont {C.}~\bibnamefont {Figgatt}},
  \bibinfo {author} {\bibfnamefont {T.}~\bibnamefont {Schuster}}, \bibinfo
  {author} {\bibfnamefont {N.~M.}\ \bibnamefont {Linke}}, \bibinfo {author}
  {\bibfnamefont {B.}~\bibnamefont {Yoshida}}, \bibinfo {author} {\bibfnamefont
  {N.~Y.}\ \bibnamefont {Yao}}, \ and\ \bibinfo {author} {\bibfnamefont
  {C.}~\bibnamefont {Monroe}},\ }\bibfield  {title} {\enquote {\bibinfo {title}
  {Verified quantum information scrambling},}\ }\href {\doibase
  10.1038/s41586-019-0952-6} {\bibfield  {journal} {\bibinfo  {journal}
  {Nature}\ }\textbf {\bibinfo {volume} {567}},\ \bibinfo {pages} {61--65}
  (\bibinfo {year} {2019})}\BibitemShut {NoStop}%
\bibitem [{\citenamefont {Wineland}\ \emph {et~al.}(1992)\citenamefont
  {Wineland}, \citenamefont {Bollinger}, \citenamefont {Itano}, \citenamefont
  {Moore},\ and\ \citenamefont {Heinzen}}]{PhysRevA.46.R6797}%
  \BibitemOpen
  \bibfield  {author} {\bibinfo {author} {\bibfnamefont {D.~J.}\ \bibnamefont
  {Wineland}}, \bibinfo {author} {\bibfnamefont {J.~J.}\ \bibnamefont
  {Bollinger}}, \bibinfo {author} {\bibfnamefont {W.~M.}\ \bibnamefont
  {Itano}}, \bibinfo {author} {\bibfnamefont {F.~L.}\ \bibnamefont {Moore}}, \
  and\ \bibinfo {author} {\bibfnamefont {D.~J.}\ \bibnamefont {Heinzen}},\
  }\bibfield  {title} {\enquote {\bibinfo {title} {Spin squeezing and reduced
  quantum noise in spectroscopy},}\ }\href {\doibase 10.1103/PhysRevA.46.R6797}
  {\bibfield  {journal} {\bibinfo  {journal} {Phys. Rev. A}\ }\textbf {\bibinfo
  {volume} {46}},\ \bibinfo {pages} {R6797--R6800} (\bibinfo {year}
  {1992})}\BibitemShut {NoStop}%
\bibitem [{\citenamefont {Foss-Feig}\ \emph
  {et~al.}(2016{\natexlab{a}})\citenamefont {Foss-Feig}, \citenamefont {Gong},
  \citenamefont {Gorshkov},\ and\ \citenamefont
  {Clark}}]{fossfeig2016entanglement}%
  \BibitemOpen
  \bibfield  {author} {\bibinfo {author} {\bibfnamefont {Michael}\ \bibnamefont
  {Foss-Feig}}, \bibinfo {author} {\bibfnamefont {Zhe-Xuan}\ \bibnamefont
  {Gong}}, \bibinfo {author} {\bibfnamefont {Alexey~V.}\ \bibnamefont
  {Gorshkov}}, \ and\ \bibinfo {author} {\bibfnamefont {Charles~W.}\
  \bibnamefont {Clark}},\ }\href@noop {} {\enquote {\bibinfo {title}
  {Entanglement and spin-squeezing without infinite-range interactions},}\ }
  (\bibinfo {year} {2016}{\natexlab{a}}),\ \Eprint
  {http://arxiv.org/abs/1612.07805} {arXiv:1612.07805} \BibitemShut {NoStop}%
\bibitem [{\citenamefont {Hastings}\ and\ \citenamefont
  {Koma}(2006)}]{Hastings2006}%
  \BibitemOpen
  \bibfield  {author} {\bibinfo {author} {\bibfnamefont {Matthew~B.}\
  \bibnamefont {Hastings}}\ and\ \bibinfo {author} {\bibfnamefont {Tohru}\
  \bibnamefont {Koma}},\ }\bibfield  {title} {\enquote {\bibinfo {title}
  {Spectral {{Gap}} and {{Exponential Decay}} of {{Correlations}}},}\ }\href
  {\doibase 10.1007/s00220-006-0030-4} {\bibfield  {journal} {\bibinfo
  {journal} {Commun. Math. Phys.}\ }\textbf {\bibinfo {volume} {265}},\
  \bibinfo {pages} {781--804} (\bibinfo {year} {2006})}\BibitemShut {NoStop}%
\bibitem [{\citenamefont {{Foss-Feig}}\ \emph {et~al.}(2015)\citenamefont
  {{Foss-Feig}}, \citenamefont {Gong}, \citenamefont {Clark},\ and\
  \citenamefont {Gorshkov}}]{Foss-Feig2015}%
  \BibitemOpen
  \bibfield  {author} {\bibinfo {author} {\bibfnamefont {Michael}\ \bibnamefont
  {{Foss-Feig}}}, \bibinfo {author} {\bibfnamefont {Zhe-Xuan}\ \bibnamefont
  {Gong}}, \bibinfo {author} {\bibfnamefont {Charles~W.}\ \bibnamefont
  {Clark}}, \ and\ \bibinfo {author} {\bibfnamefont {Alexey~V.}\ \bibnamefont
  {Gorshkov}},\ }\bibfield  {title} {\enquote {\bibinfo {title} {Nearly
  {{Linear Light Cones}} in {{Long}}-{{Range Interacting Quantum Systems}}},}\
  }\href {\doibase 10.1103/PhysRevLett.114.157201} {\bibfield  {journal}
  {\bibinfo  {journal} {Phys. Rev. Lett.}\ }\textbf {\bibinfo {volume} {114}},\
  \bibinfo {pages} {157201} (\bibinfo {year} {2015})}\BibitemShut {NoStop}%
\bibitem [{\citenamefont {Tran}\ \emph
  {et~al.}(2019{\natexlab{b}})\citenamefont {Tran}, \citenamefont {Guo},
  \citenamefont {Su}, \citenamefont {Garrison}, \citenamefont {Eldredge},
  \citenamefont {{Foss-Feig}}, \citenamefont {Childs},\ and\ \citenamefont
  {Gorshkov}}]{Tran2019a}%
  \BibitemOpen
  \bibfield  {author} {\bibinfo {author} {\bibfnamefont {Minh~C.}\ \bibnamefont
  {Tran}}, \bibinfo {author} {\bibfnamefont {Andrew~Y.}\ \bibnamefont {Guo}},
  \bibinfo {author} {\bibfnamefont {Yuan}\ \bibnamefont {Su}}, \bibinfo
  {author} {\bibfnamefont {James~R.}\ \bibnamefont {Garrison}}, \bibinfo
  {author} {\bibfnamefont {Zachary}\ \bibnamefont {Eldredge}}, \bibinfo
  {author} {\bibfnamefont {Michael}\ \bibnamefont {{Foss-Feig}}}, \bibinfo
  {author} {\bibfnamefont {Andrew~M.}\ \bibnamefont {Childs}}, \ and\ \bibinfo
  {author} {\bibfnamefont {Alexey~V.}\ \bibnamefont {Gorshkov}},\ }\bibfield
  {title} {\enquote {\bibinfo {title} {Locality and digital quantum simulation
  of power-law interactions},}\ }\href {\doibase 10.1103/PhysRevX.9.031006}
  {\bibfield  {journal} {\bibinfo  {journal} {Phys. Rev. X}\ }\textbf {\bibinfo
  {volume} {9}},\ \bibinfo {pages} {031006} (\bibinfo {year}
  {2019}{\natexlab{b}})}\BibitemShut {NoStop}%
\bibitem [{\citenamefont {Chen}\ and\ \citenamefont {Zhou}(2019)}]{Chen2019b}%
  \BibitemOpen
  \bibfield  {author} {\bibinfo {author} {\bibfnamefont {Xiao}\ \bibnamefont
  {Chen}}\ and\ \bibinfo {author} {\bibfnamefont {Tianci}\ \bibnamefont
  {Zhou}},\ }\bibfield  {title} {\enquote {\bibinfo {title} {Quantum chaos
  dynamics in long-range power law interaction systems},}\ }\href {\doibase
  10.1103/PhysRevB.100.064305} {\bibfield  {journal} {\bibinfo  {journal}
  {Phys. Rev. B}\ }\textbf {\bibinfo {volume} {100}},\ \bibinfo {pages}
  {064305} (\bibinfo {year} {2019})}\BibitemShut {NoStop}%
\bibitem [{\citenamefont {Luitz}\ and\ \citenamefont
  {Bar~Lev}(2019)}]{Luitz2019}%
  \BibitemOpen
  \bibfield  {author} {\bibinfo {author} {\bibfnamefont {David~J.}\
  \bibnamefont {Luitz}}\ and\ \bibinfo {author} {\bibfnamefont {Yevgeny}\
  \bibnamefont {Bar~Lev}},\ }\bibfield  {title} {\enquote {\bibinfo {title}
  {Emergent locality in systems with power-law interactions},}\ }\href
  {\doibase 10.1103/PhysRevA.99.010105} {\bibfield  {journal} {\bibinfo
  {journal} {Phys. Rev. A}\ }\textbf {\bibinfo {volume} {99}},\ \bibinfo
  {pages} {010105} (\bibinfo {year} {2019})}\BibitemShut {NoStop}%
\bibitem [{\citenamefont {Secular}\ \emph {et~al.}(2019)\citenamefont
  {Secular}, \citenamefont {Gourianov}, \citenamefont {Lubasch}, \citenamefont
  {Dolgov}, \citenamefont {Clark},\ and\ \citenamefont {Jaksch}}]{Secular2019}%
  \BibitemOpen
  \bibfield  {author} {\bibinfo {author} {\bibfnamefont {Paul}\ \bibnamefont
  {Secular}}, \bibinfo {author} {\bibfnamefont {Nikita}\ \bibnamefont
  {Gourianov}}, \bibinfo {author} {\bibfnamefont {Michael}\ \bibnamefont
  {Lubasch}}, \bibinfo {author} {\bibfnamefont {Sergey}\ \bibnamefont
  {Dolgov}}, \bibinfo {author} {\bibfnamefont {Stephen~R.}\ \bibnamefont
  {Clark}}, \ and\ \bibinfo {author} {\bibfnamefont {Dieter}\ \bibnamefont
  {Jaksch}},\ }\bibfield  {title} {\enquote {\bibinfo {title} {Parallel
  time-dependent variational principle algorithm for matrix product states},}\
  }\href@noop {} {\  (\bibinfo {year} {2019})},\ \Eprint
  {http://arxiv.org/abs/1912.06127} {arXiv:1912.06127} \BibitemShut {NoStop}%
\bibitem [{\citenamefont {Chen}\ and\ \citenamefont
  {Lucas}(2019{\natexlab{a}})}]{Chen2019}%
  \BibitemOpen
  \bibfield  {author} {\bibinfo {author} {\bibfnamefont {Chi-Fang}\
  \bibnamefont {Chen}}\ and\ \bibinfo {author} {\bibfnamefont {Andrew}\
  \bibnamefont {Lucas}},\ }\bibfield  {title} {\enquote {\bibinfo {title}
  {Finite speed of quantum scrambling with long range interactions},}\ }\href
  {\doibase 10.1103/PhysRevLett.123.250605} {\bibfield  {journal} {\bibinfo
  {journal} {Phys. Rev. Lett.}\ }\textbf {\bibinfo {volume} {123}},\ \bibinfo
  {pages} {250605} (\bibinfo {year} {2019}{\natexlab{a}})}\BibitemShut
  {NoStop}%
\bibitem [{\citenamefont {Kuwahara}\ and\ \citenamefont
  {Saito}(2019)}]{Kuwahara2019}%
  \BibitemOpen
  \bibfield  {author} {\bibinfo {author} {\bibfnamefont {Tomotaka}\
  \bibnamefont {Kuwahara}}\ and\ \bibinfo {author} {\bibfnamefont {Keiji}\
  \bibnamefont {Saito}},\ }\bibfield  {title} {\enquote {\bibinfo {title}
  {Strictly linear light cones in long-range interacting systems of arbitrary
  dimensions},}\ }\href@noop {} {\  (\bibinfo {year} {2019})},\ \Eprint
  {http://arxiv.org/abs/1910.14477} {arXiv:1910.14477} \BibitemShut {NoStop}%
\bibitem [{\citenamefont {Cirac}\ \emph {et~al.}(1997)\citenamefont {Cirac},
  \citenamefont {Zoller}, \citenamefont {Kimble},\ and\ \citenamefont
  {Mabuchi}}]{Cirac1997}%
  \BibitemOpen
  \bibfield  {author} {\bibinfo {author} {\bibfnamefont {J.~I.}\ \bibnamefont
  {Cirac}}, \bibinfo {author} {\bibfnamefont {P.}~\bibnamefont {Zoller}},
  \bibinfo {author} {\bibfnamefont {H.~J.}\ \bibnamefont {Kimble}}, \ and\
  \bibinfo {author} {\bibfnamefont {H.}~\bibnamefont {Mabuchi}},\ }\bibfield
  {title} {\enquote {\bibinfo {title} {Quantum {{State Transfer}} and
  {{Entanglement Distribution}} among {{Distant Nodes}} in a {{Quantum
  Network}}},}\ }\href {\doibase 10.1103/PhysRevLett.78.3221} {\bibfield
  {journal} {\bibinfo  {journal} {Phys. Rev. Lett.}\ }\textbf {\bibinfo
  {volume} {78}},\ \bibinfo {pages} {3221--3224} (\bibinfo {year}
  {1997})}\BibitemShut {NoStop}%
\bibitem [{\citenamefont {Epstein}\ and\ \citenamefont
  {Whaley}(2017)}]{Epstein2017}%
  \BibitemOpen
  \bibfield  {author} {\bibinfo {author} {\bibfnamefont {Jeffrey~M.}\
  \bibnamefont {Epstein}}\ and\ \bibinfo {author} {\bibfnamefont {K.~Birgitta}\
  \bibnamefont {Whaley}},\ }\bibfield  {title} {\enquote {\bibinfo {title}
  {Quantum speed limits for quantum-information-processing tasks},}\ }\href
  {\doibase 10.1103/PhysRevA.95.042314} {\bibfield  {journal} {\bibinfo
  {journal} {Phys. Rev. A}\ }\textbf {\bibinfo {volume} {95}},\ \bibinfo
  {pages} {042314} (\bibinfo {year} {2017})}\BibitemShut {NoStop}%
\bibitem [{\citenamefont {Aaronson}\ and\ \citenamefont
  {Arkhipov}(2011)}]{Aaronson2011}%
  \BibitemOpen
  \bibfield  {author} {\bibinfo {author} {\bibfnamefont {Scott}\ \bibnamefont
  {Aaronson}}\ and\ \bibinfo {author} {\bibfnamefont {Alex}\ \bibnamefont
  {Arkhipov}},\ }\bibfield  {title} {\enquote {\bibinfo {title} {The
  computational complexity of linear optics},}\ }in\ \href {\doibase
  10.1145/1993636.1993682} {\emph {\bibinfo {booktitle} {Proceedings of the
  Forty-Third Annual {{ACM}} Symposium on {{Theory}} of Computing}}}\ (\bibinfo
   {publisher} {{ACM Press}},\ \bibinfo {address} {{New York, New York, USA}},\
  \bibinfo {year} {2011})\ p.\ \bibinfo {pages} {333}\BibitemShut {NoStop}%
\bibitem [{\citenamefont {Maskara}\ \emph {et~al.}(2019)\citenamefont
  {Maskara}, \citenamefont {Deshpande}, \citenamefont {Tran}, \citenamefont
  {Ehrenberg}, \citenamefont {Fefferman},\ and\ \citenamefont
  {Gorshkov}}]{Maskara2019}%
  \BibitemOpen
  \bibfield  {author} {\bibinfo {author} {\bibfnamefont {Nishad}\ \bibnamefont
  {Maskara}}, \bibinfo {author} {\bibfnamefont {Abhinav}\ \bibnamefont
  {Deshpande}}, \bibinfo {author} {\bibfnamefont {Minh~C.}\ \bibnamefont
  {Tran}}, \bibinfo {author} {\bibfnamefont {Adam}\ \bibnamefont {Ehrenberg}},
  \bibinfo {author} {\bibfnamefont {Bill}\ \bibnamefont {Fefferman}}, \ and\
  \bibinfo {author} {\bibfnamefont {Alexey~V.}\ \bibnamefont {Gorshkov}},\
  }\bibfield  {title} {\enquote {\bibinfo {title} {Complexity phase diagram for
  interacting and long-range bosonic {{Hamiltonians}}},}\ }\href@noop {} {\
  (\bibinfo {year} {2019})},\ \Eprint {http://arxiv.org/abs/1906.04178}
  {arXiv:1906.04178} \BibitemShut {NoStop}%
\bibitem [{\citenamefont {Maldacena}(1999)}]{Maldacena:1997re}%
  \BibitemOpen
  \bibfield  {author} {\bibinfo {author} {\bibfnamefont {Juan~Martin}\
  \bibnamefont {Maldacena}},\ }\bibfield  {title} {\enquote {\bibinfo {title}
  {{The Large N limit of superconformal field theories and supergravity}},}\
  }\href {\doibase 10.1023/A:1026654312961} {\bibfield  {journal} {\bibinfo
  {journal} {Int. J. Theor. Phys.}\ }\textbf {\bibinfo {volume} {38}},\
  \bibinfo {pages} {1113--1133} (\bibinfo {year} {1999})},\ \Eprint
  {http://arxiv.org/abs/hep-th/9711200} {arXiv:hep-th/9711200} \BibitemShut
  {NoStop}%
\bibitem [{\citenamefont {Agarwal}\ \emph {et~al.}(1997)\citenamefont
  {Agarwal}, \citenamefont {Puri},\ and\ \citenamefont
  {Singh}}]{PhysRevA.56.2249}%
  \BibitemOpen
  \bibfield  {author} {\bibinfo {author} {\bibfnamefont {G.~S.}\ \bibnamefont
  {Agarwal}}, \bibinfo {author} {\bibfnamefont {R.~R.}\ \bibnamefont {Puri}}, \
  and\ \bibinfo {author} {\bibfnamefont {R.~P.}\ \bibnamefont {Singh}},\
  }\bibfield  {title} {\enquote {\bibinfo {title} {Atomic schr\"odinger cat
  states},}\ }\href {\doibase 10.1103/PhysRevA.56.2249} {\bibfield  {journal}
  {\bibinfo  {journal} {Phys. Rev. A}\ }\textbf {\bibinfo {volume} {56}},\
  \bibinfo {pages} {2249--2254} (\bibinfo {year} {1997})}\BibitemShut {NoStop}%
\bibitem [{\citenamefont {Lewis-Swan}\ \emph {et~al.}(2018)\citenamefont
  {Lewis-Swan}, \citenamefont {Norcia}, \citenamefont {Cline}, \citenamefont
  {Thompson},\ and\ \citenamefont {Rey}}]{PhysRevLett.121.070403}%
  \BibitemOpen
  \bibfield  {author} {\bibinfo {author} {\bibfnamefont {Robert~J.}\
  \bibnamefont {Lewis-Swan}}, \bibinfo {author} {\bibfnamefont {Matthew~A.}\
  \bibnamefont {Norcia}}, \bibinfo {author} {\bibfnamefont {Julia R.~K.}\
  \bibnamefont {Cline}}, \bibinfo {author} {\bibfnamefont {James~K.}\
  \bibnamefont {Thompson}}, \ and\ \bibinfo {author} {\bibfnamefont
  {Ana~Maria}\ \bibnamefont {Rey}},\ }\bibfield  {title} {\enquote {\bibinfo
  {title} {Robust spin squeezing via photon-mediated interactions on an optical
  clock transition},}\ }\href {\doibase 10.1103/PhysRevLett.121.070403}
  {\bibfield  {journal} {\bibinfo  {journal} {Phys. Rev. Lett.}\ }\textbf
  {\bibinfo {volume} {121}},\ \bibinfo {pages} {070403} (\bibinfo {year}
  {2018})}\BibitemShut {NoStop}%
\bibitem [{\citenamefont {Foss-Feig}\ \emph
  {et~al.}(2016{\natexlab{b}})\citenamefont {Foss-Feig}, \citenamefont {Gong},
  \citenamefont {Gorshkov},\ and\ \citenamefont {Clark}}]{foss-feig16b}%
  \BibitemOpen
  \bibfield  {author} {\bibinfo {author} {\bibfnamefont {Michael}\ \bibnamefont
  {Foss-Feig}}, \bibinfo {author} {\bibfnamefont {Zhe-Xuan}\ \bibnamefont
  {Gong}}, \bibinfo {author} {\bibfnamefont {Alexey~V.}\ \bibnamefont
  {Gorshkov}}, \ and\ \bibinfo {author} {\bibfnamefont {Charles~W.}\
  \bibnamefont {Clark}},\ }\bibfield  {title} {\enquote {\bibinfo {title}
  {Entanglement and spin-squeezing without infinite-range interactions},}\
  }\href@noop {} {\bibfield  {journal} {\bibinfo  {journal} {arXiv:1612.07805}\
  } (\bibinfo {year} {2016}{\natexlab{b}})}\BibitemShut {NoStop}%
\bibitem [{\citenamefont {He}\ \emph {et~al.}(2019)\citenamefont {He},
  \citenamefont {Perlin}, \citenamefont {Muleady}, \citenamefont {Lewis-Swan},
  \citenamefont {Hutson}, \citenamefont {Ye},\ and\ \citenamefont
  {Rey}}]{PhysRevResearch.1.033075}%
  \BibitemOpen
  \bibfield  {author} {\bibinfo {author} {\bibfnamefont {P.}~\bibnamefont
  {He}}, \bibinfo {author} {\bibfnamefont {M.~A.}\ \bibnamefont {Perlin}},
  \bibinfo {author} {\bibfnamefont {S.~R.}\ \bibnamefont {Muleady}}, \bibinfo
  {author} {\bibfnamefont {R.~J.}\ \bibnamefont {Lewis-Swan}}, \bibinfo
  {author} {\bibfnamefont {R.~B.}\ \bibnamefont {Hutson}}, \bibinfo {author}
  {\bibfnamefont {J.}~\bibnamefont {Ye}}, \ and\ \bibinfo {author}
  {\bibfnamefont {A.~M.}\ \bibnamefont {Rey}},\ }\bibfield  {title} {\enquote
  {\bibinfo {title} {Engineering spin squeezing in a 3d optical lattice with
  interacting spin-orbit-coupled fermions},}\ }\href {\doibase
  10.1103/PhysRevResearch.1.033075} {\bibfield  {journal} {\bibinfo  {journal}
  {Phys. Rev. Research}\ }\textbf {\bibinfo {volume} {1}},\ \bibinfo {pages}
  {033075} (\bibinfo {year} {2019})}\BibitemShut {NoStop}%
\bibitem [{\citenamefont {Gong}\ \emph {et~al.}(2014)\citenamefont {Gong},
  \citenamefont {Foss-Feig}, \citenamefont {Michalakis},\ and\ \citenamefont
  {Gorshkov}}]{PhysRevLett.113.030602}%
  \BibitemOpen
  \bibfield  {author} {\bibinfo {author} {\bibfnamefont {Zhe-Xuan}\
  \bibnamefont {Gong}}, \bibinfo {author} {\bibfnamefont {Michael}\
  \bibnamefont {Foss-Feig}}, \bibinfo {author} {\bibfnamefont {Spyridon}\
  \bibnamefont {Michalakis}}, \ and\ \bibinfo {author} {\bibfnamefont
  {Alexey~V.}\ \bibnamefont {Gorshkov}},\ }\bibfield  {title} {\enquote
  {\bibinfo {title} {Persistence of locality in systems with power-law
  interactions},}\ }\href {\doibase 10.1103/PhysRevLett.113.030602} {\bibfield
  {journal} {\bibinfo  {journal} {Phys. Rev. Lett.}\ }\textbf {\bibinfo
  {volume} {113}},\ \bibinfo {pages} {030602} (\bibinfo {year}
  {2014})}\BibitemShut {NoStop}%
\bibitem [{\citenamefont {Richerme}\ \emph {et~al.}(2014)\citenamefont
  {Richerme}, \citenamefont {Gong}, \citenamefont {Lee}, \citenamefont {Senko},
  \citenamefont {Smith}, \citenamefont {{Foss-Feig}}, \citenamefont
  {Michalakis}, \citenamefont {Gorshkov},\ and\ \citenamefont
  {Monroe}}]{richermeNonlocalPropagationCorrelations2014}%
  \BibitemOpen
  \bibfield  {author} {\bibinfo {author} {\bibfnamefont {Philip}\ \bibnamefont
  {Richerme}}, \bibinfo {author} {\bibfnamefont {Zhe-Xuan}\ \bibnamefont
  {Gong}}, \bibinfo {author} {\bibfnamefont {Aaron}\ \bibnamefont {Lee}},
  \bibinfo {author} {\bibfnamefont {Crystal}\ \bibnamefont {Senko}}, \bibinfo
  {author} {\bibfnamefont {Jacob}\ \bibnamefont {Smith}}, \bibinfo {author}
  {\bibfnamefont {Michael}\ \bibnamefont {{Foss-Feig}}}, \bibinfo {author}
  {\bibfnamefont {Spyridon}\ \bibnamefont {Michalakis}}, \bibinfo {author}
  {\bibfnamefont {Alexey~V.}\ \bibnamefont {Gorshkov}}, \ and\ \bibinfo
  {author} {\bibfnamefont {Christopher}\ \bibnamefont {Monroe}},\ }\bibfield
  {title} {\enquote {\bibinfo {title} {Non-local propagation of correlations in
  quantum systems with long-range interactions},}\ }\href {\doibase
  10.1038/nature13450} {\bibfield  {journal} {\bibinfo  {journal} {Nature}\
  }\textbf {\bibinfo {volume} {511}},\ \bibinfo {pages} {198--201} (\bibinfo
  {year} {2014})}\BibitemShut {NoStop}%
\bibitem [{\citenamefont {Osborne}(2006)}]{Osborne2006}%
  \BibitemOpen
  \bibfield  {author} {\bibinfo {author} {\bibfnamefont {Tobias~J.}\
  \bibnamefont {Osborne}},\ }\bibfield  {title} {\enquote {\bibinfo {title}
  {Efficient {{Approximation}} of the {{Dynamics}} of {{One}}-{{Dimensional
  Quantum Spin Systems}}},}\ }\href {\doibase 10.1103/PhysRevLett.97.157202}
  {\bibfield  {journal} {\bibinfo  {journal} {Phys. Rev. Lett.}\ }\textbf
  {\bibinfo {volume} {97}},\ \bibinfo {pages} {157202} (\bibinfo {year}
  {2006})}\BibitemShut {NoStop}%
\bibitem [{\citenamefont {Haah}\ \emph {et~al.}(2018)\citenamefont {Haah},
  \citenamefont {Hastings}, \citenamefont {Kothari},\ and\ \citenamefont
  {Low}}]{Haah}%
  \BibitemOpen
  \bibfield  {author} {\bibinfo {author} {\bibfnamefont {Jeongwan}\
  \bibnamefont {Haah}}, \bibinfo {author} {\bibfnamefont {Matthew~B.}\
  \bibnamefont {Hastings}}, \bibinfo {author} {\bibfnamefont {Robin}\
  \bibnamefont {Kothari}}, \ and\ \bibinfo {author} {\bibfnamefont {Guang~Hao}\
  \bibnamefont {Low}},\ }\bibfield  {title} {\enquote {\bibinfo {title}
  {Quantum algorithm for simulating real time evolution of lattice
  hamiltonians},}\ }in\ \href {\doibase 10.1109/FOCS.2018.00041} {\emph
  {\bibinfo {booktitle} {59th {IEEE} Annual Symposium on Foundations of
  Computer Science, {FOCS} 2018, Paris, France, October 7-9, 2018}}},\ \bibinfo
  {editor} {edited by\ \bibinfo {editor} {\bibfnamefont {Mikkel}\ \bibnamefont
  {Thorup}}}\ (\bibinfo  {publisher} {{IEEE} Computer Society},\ \bibinfo
  {year} {2018})\ pp.\ \bibinfo {pages} {350--360}\BibitemShut {NoStop}%
\bibitem [{\citenamefont {Bravyi}\ \emph {et~al.}(2006)\citenamefont {Bravyi},
  \citenamefont {Hastings},\ and\ \citenamefont {Verstraete}}]{Bravyi2006}%
  \BibitemOpen
  \bibfield  {author} {\bibinfo {author} {\bibfnamefont {S.}~\bibnamefont
  {Bravyi}}, \bibinfo {author} {\bibfnamefont {M.~B.}\ \bibnamefont
  {Hastings}}, \ and\ \bibinfo {author} {\bibfnamefont {F.}~\bibnamefont
  {Verstraete}},\ }\bibfield  {title} {\enquote {\bibinfo {title}
  {Lieb-{{Robinson Bounds}} and the {{Generation}} of {{Correlations}} and
  {{Topological Quantum Order}}},}\ }\href {\doibase
  10.1103/PhysRevLett.97.050401} {\bibfield  {journal} {\bibinfo  {journal}
  {Phys. Rev. Lett.}\ }\textbf {\bibinfo {volume} {97}},\ \bibinfo {pages}
  {050401} (\bibinfo {year} {2006})}\BibitemShut {NoStop}%
\bibitem [{\citenamefont {Hastings}(2007)}]{Hastings_2007}%
  \BibitemOpen
  \bibfield  {author} {\bibinfo {author} {\bibfnamefont {M~B}\ \bibnamefont
  {Hastings}},\ }\bibfield  {title} {\enquote {\bibinfo {title} {An area law
  for one-dimensional quantum systems},}\ }\href {\doibase
  10.1088/1742-5468/2007/08/p08024} {\bibfield  {journal} {\bibinfo  {journal}
  {Journal of Statistical Mechanics: Theory and Experiment}\ }\textbf {\bibinfo
  {volume} {2007}},\ \bibinfo {pages} {P08024–P08024} (\bibinfo {year}
  {2007})}\BibitemShut {NoStop}%
\bibitem [{\citenamefont {Abanin}\ \emph {et~al.}(2015)\citenamefont {Abanin},
  \citenamefont {De~Roeck},\ and\ \citenamefont {Huveneers}}]{Abanin2015}%
  \BibitemOpen
  \bibfield  {author} {\bibinfo {author} {\bibfnamefont {Dmitry~A.}\
  \bibnamefont {Abanin}}, \bibinfo {author} {\bibfnamefont {Wojciech}\
  \bibnamefont {De~Roeck}}, \ and\ \bibinfo {author} {\bibfnamefont
  {François}\ \bibnamefont {Huveneers}},\ }\bibfield  {title} {\enquote
  {\bibinfo {title} {Exponentially {{Slow Heating}} in {{Periodically Driven
  Many}}-{{Body Systems}}},}\ }\href {\doibase 10.1103/PhysRevLett.115.256803}
  {\bibfield  {journal} {\bibinfo  {journal} {Phys. Rev. Lett.}\ }\textbf
  {\bibinfo {volume} {115}} (\bibinfo {year} {2015}),\
  10.1103/PhysRevLett.115.256803}\BibitemShut {NoStop}%
\bibitem [{\citenamefont {Abanin}\ \emph {et~al.}(2017)\citenamefont {Abanin},
  \citenamefont {De~Roeck}, \citenamefont {Ho},\ and\ \citenamefont
  {Huveneers}}]{Abanin2017a}%
  \BibitemOpen
  \bibfield  {author} {\bibinfo {author} {\bibfnamefont {Dmitry}\ \bibnamefont
  {Abanin}}, \bibinfo {author} {\bibfnamefont {Wojciech}\ \bibnamefont
  {De~Roeck}}, \bibinfo {author} {\bibfnamefont {Wen~Wei}\ \bibnamefont {Ho}},
  \ and\ \bibinfo {author} {\bibfnamefont {Francois}\ \bibnamefont
  {Huveneers}},\ }\bibfield  {title} {\enquote {\bibinfo {title} {A {{Rigorous
  Theory}} of {{Many}}-{{Body Prethermalization}} for {{Periodically Driven}}
  and {{Closed Quantum Systems}}},}\ }\href {\doibase
  10.1007/s00220-017-2930-x} {\bibfield  {journal} {\bibinfo  {journal}
  {Commun. Math. Phys.}\ }\textbf {\bibinfo {volume} {354}},\ \bibinfo {pages}
  {809--827} (\bibinfo {year} {2017})}\BibitemShut {NoStop}%
\bibitem [{\citenamefont {Kuwahara}\ \emph {et~al.}(2016)\citenamefont
  {Kuwahara}, \citenamefont {Mori},\ and\ \citenamefont
  {Saito}}]{Kuwahara2016}%
  \BibitemOpen
  \bibfield  {author} {\bibinfo {author} {\bibfnamefont {Tomotaka}\
  \bibnamefont {Kuwahara}}, \bibinfo {author} {\bibfnamefont {Takashi}\
  \bibnamefont {Mori}}, \ and\ \bibinfo {author} {\bibfnamefont {Keiji}\
  \bibnamefont {Saito}},\ }\bibfield  {title} {\enquote {\bibinfo {title}
  {Floquet–{{Magnus}} theory and generic transient dynamics in periodically
  driven many-body quantum systems},}\ }\href {\doibase
  10.1016/j.aop.2016.01.012} {\bibfield  {journal} {\bibinfo  {journal} {Ann.
  Phys.}\ }\textbf {\bibinfo {volume} {367}},\ \bibinfo {pages} {96--124}
  (\bibinfo {year} {2016})}\BibitemShut {NoStop}%
\bibitem [{\citenamefont {Gong}\ \emph {et~al.}(2017)\citenamefont {Gong},
  \citenamefont {Foss-Feig}, \citenamefont {Brand\~ao},\ and\ \citenamefont
  {Gorshkov}}]{PhysRevLett.119.050501}%
  \BibitemOpen
  \bibfield  {author} {\bibinfo {author} {\bibfnamefont {Zhe-Xuan}\
  \bibnamefont {Gong}}, \bibinfo {author} {\bibfnamefont {Michael}\
  \bibnamefont {Foss-Feig}}, \bibinfo {author} {\bibfnamefont {Fernando G.
  S.~L.}\ \bibnamefont {Brand\~ao}}, \ and\ \bibinfo {author} {\bibfnamefont
  {Alexey~V.}\ \bibnamefont {Gorshkov}},\ }\bibfield  {title} {\enquote
  {\bibinfo {title} {Entanglement area laws for long-range interacting
  systems},}\ }\href {\doibase 10.1103/PhysRevLett.119.050501} {\bibfield
  {journal} {\bibinfo  {journal} {Phys. Rev. Lett.}\ }\textbf {\bibinfo
  {volume} {119}},\ \bibinfo {pages} {050501} (\bibinfo {year}
  {2017})}\BibitemShut {NoStop}%
\bibitem [{\citenamefont {Shenker}\ and\ \citenamefont
  {Stanford}(2014)}]{Shenker2014}%
  \BibitemOpen
  \bibfield  {author} {\bibinfo {author} {\bibfnamefont {Stephen~H.}\
  \bibnamefont {Shenker}}\ and\ \bibinfo {author} {\bibfnamefont {Douglas}\
  \bibnamefont {Stanford}},\ }\bibfield  {title} {\enquote {\bibinfo {title}
  {Black holes and the butterfly effect},}\ }\href {\doibase
  10.1007/JHEP03(2014)067} {\bibfield  {journal} {\bibinfo  {journal} {J. High
  Energ. Phys.}\ }\textbf {\bibinfo {volume} {2014}},\ \bibinfo {pages} {67}
  (\bibinfo {year} {2014})}\BibitemShut {NoStop}%
\bibitem [{\citenamefont {Maldacena}\ \emph {et~al.}(2016)\citenamefont
  {Maldacena}, \citenamefont {Shenker},\ and\ \citenamefont
  {Stanford}}]{Maldacena2016}%
  \BibitemOpen
  \bibfield  {author} {\bibinfo {author} {\bibfnamefont {Juan}\ \bibnamefont
  {Maldacena}}, \bibinfo {author} {\bibfnamefont {Stephen~H.}\ \bibnamefont
  {Shenker}}, \ and\ \bibinfo {author} {\bibfnamefont {Douglas}\ \bibnamefont
  {Stanford}},\ }\bibfield  {title} {\enquote {\bibinfo {title} {A bound on
  chaos},}\ }\href {\doibase 10.1007/JHEP08(2016)106} {\bibfield  {journal}
  {\bibinfo  {journal} {Journal of High Energy Physics}\ }\textbf {\bibinfo
  {volume} {2016}},\ \bibinfo {pages} {106} (\bibinfo {year}
  {2016})}\BibitemShut {NoStop}%
\bibitem [{\citenamefont {G\"arttner}\ \emph {et~al.}(2017)\citenamefont
  {G\"arttner}, \citenamefont {Bohnet}, \citenamefont {Safavi-Naini},
  \citenamefont {Wall}, \citenamefont {Bollinger},\ and\ \citenamefont
  {Rey}}]{Garttner_2017}%
  \BibitemOpen
  \bibfield  {author} {\bibinfo {author} {\bibfnamefont {Martin}\ \bibnamefont
  {G\"arttner}}, \bibinfo {author} {\bibfnamefont {Justin~G.}\ \bibnamefont
  {Bohnet}}, \bibinfo {author} {\bibfnamefont {Arghavan}\ \bibnamefont
  {Safavi-Naini}}, \bibinfo {author} {\bibfnamefont {Michael~L.}\ \bibnamefont
  {Wall}}, \bibinfo {author} {\bibfnamefont {John~J.}\ \bibnamefont
  {Bollinger}}, \ and\ \bibinfo {author} {\bibfnamefont {Ana~Maria}\
  \bibnamefont {Rey}},\ }\bibfield  {title} {\enquote {\bibinfo {title}
  {Measuring out-of-time-order correlations and multiple quantum spectra in a
  trapped-ion quantum magnet},}\ }\href {\doibase 10.1038/nphys4119} {\bibfield
   {journal} {\bibinfo  {journal} {Nature Physics}\ }\textbf {\bibinfo {volume}
  {13}},\ \bibinfo {pages} {781–786} (\bibinfo {year} {2017})}\BibitemShut
  {NoStop}%
\bibitem [{\citenamefont {Li}\ \emph {et~al.}(2017)\citenamefont {Li},
  \citenamefont {Fan}, \citenamefont {Wang}, \citenamefont {Ye}, \citenamefont
  {Zeng}, \citenamefont {Zhai}, \citenamefont {Peng},\ and\ \citenamefont
  {Du}}]{Li_2017}%
  \BibitemOpen
  \bibfield  {author} {\bibinfo {author} {\bibfnamefont {Jun}\ \bibnamefont
  {Li}}, \bibinfo {author} {\bibfnamefont {Ruihua}\ \bibnamefont {Fan}},
  \bibinfo {author} {\bibfnamefont {Hengyan}\ \bibnamefont {Wang}}, \bibinfo
  {author} {\bibfnamefont {Bingtian}\ \bibnamefont {Ye}}, \bibinfo {author}
  {\bibfnamefont {Bei}\ \bibnamefont {Zeng}}, \bibinfo {author} {\bibfnamefont
  {Hui}\ \bibnamefont {Zhai}}, \bibinfo {author} {\bibfnamefont {Xinhua}\
  \bibnamefont {Peng}}, \ and\ \bibinfo {author} {\bibfnamefont {Jiangfeng}\
  \bibnamefont {Du}},\ }\bibfield  {title} {\enquote {\bibinfo {title}
  {Measuring out-of-time-order correlators on a nuclear magnetic resonance
  quantum simulator},}\ }\href {\doibase 10.1103/physrevx.7.031011} {\bibfield
  {journal} {\bibinfo  {journal} {Physical Review X}\ }\textbf {\bibinfo
  {volume} {7}} (\bibinfo {year} {2017}),\
  10.1103/physrevx.7.031011}\BibitemShut {NoStop}%
\bibitem [{\citenamefont {Chen}\ and\ \citenamefont
  {Lucas}(2019{\natexlab{b}})}]{chen2019operator}%
  \BibitemOpen
  \bibfield  {author} {\bibinfo {author} {\bibfnamefont {Chi-Fang}\
  \bibnamefont {Chen}}\ and\ \bibinfo {author} {\bibfnamefont {Andrew}\
  \bibnamefont {Lucas}},\ }\bibfield  {title} {\enquote {\bibinfo {title}
  {Operator growth bounds from graph theory},}\ }\href@noop {} {\  (\bibinfo
  {year} {2019}{\natexlab{b}})},\ \Eprint {http://arxiv.org/abs/1905.03682}
  {arXiv:1905.03682} \BibitemShut {NoStop}%
\bibitem [{\citenamefont {Childs}\ \emph {et~al.}(2019)\citenamefont {Childs},
  \citenamefont {Su}, \citenamefont {Tran}, \citenamefont {Wiebe},\ and\
  \citenamefont {Zhu}}]{Childs2019a}%
  \BibitemOpen
  \bibfield  {author} {\bibinfo {author} {\bibfnamefont {Andrew~M.}\
  \bibnamefont {Childs}}, \bibinfo {author} {\bibfnamefont {Yuan}\ \bibnamefont
  {Su}}, \bibinfo {author} {\bibfnamefont {Minh~C.}\ \bibnamefont {Tran}},
  \bibinfo {author} {\bibfnamefont {Nathan}\ \bibnamefont {Wiebe}}, \ and\
  \bibinfo {author} {\bibfnamefont {Shuchen}\ \bibnamefont {Zhu}},\ }\bibfield
  {title} {\enquote {\bibinfo {title} {A {{Theory}} of {{Trotter Error}}},}\
  }\href@noop {} {\  (\bibinfo {year} {2019})},\ \Eprint
  {http://arxiv.org/abs/1912.08854} {arXiv:1912.08854} \BibitemShut {NoStop}%
\bibitem [{\citenamefont {Lucas}(2019)}]{Lucas2019}%
  \BibitemOpen
  \bibfield  {author} {\bibinfo {author} {\bibfnamefont {Andrew}\ \bibnamefont
  {Lucas}},\ }\bibfield  {title} {\enquote {\bibinfo {title} {Non-perturbative
  dynamics of the operator size distribution in the
  {{Sachdev}}-{{Ye}}-{{Kitaev}} model},}\ }\href@noop {} {\  (\bibinfo {year}
  {2019})},\ \Eprint {http://arxiv.org/abs/1910.09539} {arXiv:1910.09539}
  \BibitemShut {NoStop}%
\bibitem [{\citenamefont {Meyer}(2000)}]{Meyer2000}%
  \BibitemOpen
  \bibfield  {author} {\bibinfo {author} {\bibfnamefont {C.~D.}\ \bibnamefont
  {Meyer}},\ }\href@noop {} {\emph {\bibinfo {title} {Matrix Analysis and
  Applied Linear Algebra}}}\ (\bibinfo  {publisher} {{Society for Industrial
  and Applied Mathematics}},\ \bibinfo {address} {{Philadelphia}},\ \bibinfo
  {year} {2000})\BibitemShut {NoStop}%
\bibitem [{\citenamefont {Neyenhuis}\ \emph {et~al.}(2017)\citenamefont
  {Neyenhuis}, \citenamefont {Zhang}, \citenamefont {Hess}, \citenamefont
  {Smith}, \citenamefont {Lee}, \citenamefont {Richerme}, \citenamefont {Gong},
  \citenamefont {Gorshkov},\ and\ \citenamefont {Monroe}}]{Neyenhuise1700672}%
  \BibitemOpen
  \bibfield  {author} {\bibinfo {author} {\bibfnamefont {Brian}\ \bibnamefont
  {Neyenhuis}}, \bibinfo {author} {\bibfnamefont {Jiehang}\ \bibnamefont
  {Zhang}}, \bibinfo {author} {\bibfnamefont {Paul~W.}\ \bibnamefont {Hess}},
  \bibinfo {author} {\bibfnamefont {Jacob}\ \bibnamefont {Smith}}, \bibinfo
  {author} {\bibfnamefont {Aaron~C.}\ \bibnamefont {Lee}}, \bibinfo {author}
  {\bibfnamefont {Phil}\ \bibnamefont {Richerme}}, \bibinfo {author}
  {\bibfnamefont {Zhe-Xuan}\ \bibnamefont {Gong}}, \bibinfo {author}
  {\bibfnamefont {Alexey~V.}\ \bibnamefont {Gorshkov}}, \ and\ \bibinfo
  {author} {\bibfnamefont {Christopher}\ \bibnamefont {Monroe}},\ }\bibfield
  {title} {\enquote {\bibinfo {title} {Observation of prethermalization in
  long-range interacting spin chains},}\ }\href {\doibase
  10.1126/sciadv.1700672} {\bibfield  {journal} {\bibinfo  {journal} {Science
  Advances}\ }\textbf {\bibinfo {volume} {3}} (\bibinfo {year} {2017}),\
  10.1126/sciadv.1700672}\BibitemShut {NoStop}%
\bibitem [{\citenamefont {et~al.}()}]{hongfuture}%
  \BibitemOpen
  \bibfield  {author} {\bibinfo {author} {\bibfnamefont {Y.~Hong}\ \bibnamefont
  {et~al.}},\ }\href@noop {} {\ }\bibinfo {note} {In preparation}\BibitemShut
  {NoStop}%
\bibitem [{\citenamefont {Mehta}(2004)}]{mehta}%
  \BibitemOpen
  \bibfield  {author} {\bibinfo {author} {\bibfnamefont {M.~L.}\ \bibnamefont
  {Mehta}},\ }\href@noop {} {\emph {\bibinfo {title} {Random Matrices}}},\
  \bibinfo {edition} {3rd}\ ed.\ (\bibinfo  {publisher} {{Academic Press}},\
  \bibinfo {year} {2004})\BibitemShut {NoStop}%
\bibitem [{\citenamefont {Muraleedharan}\ \emph {et~al.}(2018)\citenamefont
  {Muraleedharan}, \citenamefont {Miyake},\ and\ \citenamefont
  {Deutsch}}]{Muraleedharan2018}%
  \BibitemOpen
  \bibfield  {author} {\bibinfo {author} {\bibfnamefont {Gopikrishnan}\
  \bibnamefont {Muraleedharan}}, \bibinfo {author} {\bibfnamefont {Akimasa}\
  \bibnamefont {Miyake}}, \ and\ \bibinfo {author} {\bibfnamefont {Ivan~H.}\
  \bibnamefont {Deutsch}},\ }\bibfield  {title} {\enquote {\bibinfo {title}
  {Quantum computational supremacy in the sampling of bosonic random walkers on
  a one-dimensional lattice},}\ }\href {\doibase 10.1088/1367-2630/ab0610}
  {\bibfield  {journal} {\bibinfo  {journal} {New J. Phys.}\ }\textbf {\bibinfo
  {volume} {21}},\ \bibinfo {pages} {055003} (\bibinfo {year}
  {2018})}\BibitemShut {NoStop}%
\bibitem [{\citenamefont {Hastings}(2010)}]{Hastings2010}%
  \BibitemOpen
  \bibfield  {author} {\bibinfo {author} {\bibfnamefont {M.~B.}\ \bibnamefont
  {Hastings}},\ }\bibfield  {title} {\enquote {\bibinfo {title} {Locality in
  {{Quantum Systems}}},}\ }\href@noop {} {\  (\bibinfo {year} {2010})},\
  \Eprint {http://arxiv.org/abs/1008.5137} {arXiv:1008.5137} \BibitemShut
  {NoStop}%
\bibitem [{\citenamefont {Gong}\ \emph {et~al.}(2016)\citenamefont {Gong},
  \citenamefont {Maghrebi}, \citenamefont {Hu}, \citenamefont {Wall},
  \citenamefont {{Foss-Feig}},\ and\ \citenamefont {Gorshkov}}]{Gong2016}%
  \BibitemOpen
  \bibfield  {author} {\bibinfo {author} {\bibfnamefont {Z.-X.}\ \bibnamefont
  {Gong}}, \bibinfo {author} {\bibfnamefont {M.~F.}\ \bibnamefont {Maghrebi}},
  \bibinfo {author} {\bibfnamefont {A.}~\bibnamefont {Hu}}, \bibinfo {author}
  {\bibfnamefont {M.~L.}\ \bibnamefont {Wall}}, \bibinfo {author}
  {\bibfnamefont {M.}~\bibnamefont {{Foss-Feig}}}, \ and\ \bibinfo {author}
  {\bibfnamefont {A.~V.}\ \bibnamefont {Gorshkov}},\ }\bibfield  {title}
  {\enquote {\bibinfo {title} {Topological phases with long-range
  interactions},}\ }\href {\doibase 10.1103/PhysRevB.93.041102} {\bibfield
  {journal} {\bibinfo  {journal} {Phys. Rev. B}\ }\textbf {\bibinfo {volume}
  {93}},\ \bibinfo {pages} {041102} (\bibinfo {year} {2016})}\BibitemShut
  {NoStop}%
\bibitem [{\citenamefont {Else}\ \emph {et~al.}(2018)\citenamefont {Else},
  \citenamefont {Machado}, \citenamefont {Nayak},\ and\ \citenamefont
  {Yao}}]{Else2018}%
  \BibitemOpen
  \bibfield  {author} {\bibinfo {author} {\bibfnamefont {Dominic~V.}\
  \bibnamefont {Else}}, \bibinfo {author} {\bibfnamefont {Francisco}\
  \bibnamefont {Machado}}, \bibinfo {author} {\bibfnamefont {Chetan}\
  \bibnamefont {Nayak}}, \ and\ \bibinfo {author} {\bibfnamefont {Norman~Y.}\
  \bibnamefont {Yao}},\ }\bibfield  {title} {\enquote {\bibinfo {title} {An
  improved {{Lieb}}-{{Robinson}} bound for many-body {{Hamiltonians}} with
  power-law interactions},}\ }\href@noop {} {\  (\bibinfo {year} {2018})},\
  \Eprint {http://arxiv.org/abs/1809.06369} {arXiv:1809.06369} \BibitemShut
  {NoStop}%
\bibitem [{\citenamefont {Zhou}\ \emph {et~al.}(2019)\citenamefont {Zhou},
  \citenamefont {Xu}, \citenamefont {Chen}, \citenamefont {Guo},\ and\
  \citenamefont {Swingle}}]{Zhou2019b}%
  \BibitemOpen
  \bibfield  {author} {\bibinfo {author} {\bibfnamefont {Tianci}\ \bibnamefont
  {Zhou}}, \bibinfo {author} {\bibfnamefont {Shenglong}\ \bibnamefont {Xu}},
  \bibinfo {author} {\bibfnamefont {Xiao}\ \bibnamefont {Chen}}, \bibinfo
  {author} {\bibfnamefont {Andrew}\ \bibnamefont {Guo}}, \ and\ \bibinfo
  {author} {\bibfnamefont {Brian}\ \bibnamefont {Swingle}},\ }\bibfield
  {title} {\enquote {\bibinfo {title} {{The Operator L\'evy Flight: Light Cones
  in Chaotic Long-Range Interacting Systems}},}\ }\href@noop {} {\  (\bibinfo
  {year} {2019})},\ \Eprint {http://arxiv.org/abs/1909.08646}
  {arXiv:1909.08646} \BibitemShut {NoStop}%
\end{thebibliography}%
\end{document}